\documentclass[11pt]{article}
\usepackage{smile}

\usepackage{fullpage}
\usepackage{titling}

\usepackage{float}
\usepackage{subcaption}
\usepackage{epstopdf} % converts an EPåS file to a PDF file

\usepackage{diagbox} % making table heads with diagonal lines
\usepackage[font={small,it}]{caption}  % set font size and Italian of caption

\usepackage{bbm}
\usepackage{mathtools}
\usepackage{mathrsfs} 
\usepackage{cancel} % \cancel{} draw diagonal lines to cancel an item
\usepackage{lipsum} % select from paragraphs
\usepackage{marvosym} % Euro currency symbol, zodiac etc
\usepackage{xr} % Cross-referencing among tex files
\usepackage{ulem} %strike out comments
\usepackage[colorlinks=true,
            linkcolor=blue,
            urlcolor=blue,
            citecolor=blue]{hyperref} % set hyper link for citations,equations

\newcommand{\T}{^\top}

\newcommand{\sgn}{\mathop{\mathrm{sign}}}

\newcommand*{\rank}{\operatorname{rank}}

\renewcommand*{\PP}{\mathbb{P}}
\newcommand*{\var}{\operatorname{Var}}

\newcommand*{\argsup}{\operatornamewithlimits{argsup}}

\newcommand*{\col}{\operatorname{col}}

\begin{document}

\title{\huge A New Perspective on Debiasing Linear Regressions}
\author{Yufei Yi \and Matey Neykov}
\date{Department of Statistics \& Data Science\\ Carnegie Mellon University}

\maketitle{}

\begin{abstract}
In this paper, we propose an abstract procedure for debiasing constrained or regularized potentially high-dimensional linear models. It is elementary to show that the proposed procedure can produce $\frac{1}{\sqrt{n}}$-confidence intervals for individual coordinates (or even bounded contrasts) in models with unknown covariance, provided that the covariance has bounded spectrum. While the proof of the statistical guarantees of our procedure is simple, its implementation requires more care due to the complexity of the optimization programs we need to solve. We spend the bulk of this paper giving examples in which the proposed algorithm can be implemented in practice. One fairly general class of instances which are amenable to applications of our procedure include convex constrained least squares. We are able to translate the procedure to an abstract algorithm over this class of models, and we give concrete examples where efficient polynomial time methods for debiasing exist. Those include the constrained version of the group LASSO, regression under monotone constraints, regression with positive monotone constraints and non-negative least squares. We also demonstrate that our method can debias Minkowski gauge selectors such as the ones proposed by \cite{cai2016geometric} under a certain condition. This solves an open problem posed by \cite{cai2016geometric} on how to debias such selectors when the covariance is unknown. In addition, we show that our abstract procedure can be applied to efficiently debias group LASSO, SLOPE and square-root SLOPE, among other popular regularized procedures under certain assumptions. We provide thorough simulation results in support of our theoretical findings.
\end{abstract}

%%%%%%%%%%%%%%%%%%%%
%% Sec 1: Intro
%%%%%%%%%%%%%%%%%%%%
\section{Introduction}

Linear regression is a pillar in statistics. Due to its simplicity and interpretability, it is possibly the most widely known and used statistical modeling and estimation technique both within and outside the field of statistics. The amount of literature on linear regression is vast, and ever-growing. In addition, with the big data boom, high-dimensional regression has steadily become an indispensable tool in practice, and has been in the focus of statisticians and practitioners for the past number of years. By far the most widely used estimator for the linear model is the ordinary least squares estimator (OLS). Unfortunately, OLS does not allow the practitioner to build in prior knowledge on the coefficients of interest. However, prior knowledge, e.g. sparsity, can be crucial for performing reasonable estimation especially in modern large datasets like genome-wide association studies where the number of samples can be smaller than the number of covariates. Incorporating prior knowledge (in a frequentist sense) may come at a price --- it is not immediately obvious how to perform inference since the resulting estimator might not have a closed form, in contrast to the OLS, and in addition the estimated coefficients are likely biased. In this paper we tackle questions of this flavor: we suggest an abstract procedure which can perform inference for certain estimators in linear models which are ``non-OLS'', such as some convex constraint least squares estimators and some reguralized estimators such as the Sorted L-One Penalized Estimator (SLOPE) and square-root SLOPE.

As we mentioned, parameter estimation in high-dimensional statistical models typically requires solving a regularized (or constrained) optimization problem. Regularization is necessitated in order to help fight the curse of dimensionality. Since the resulting estimators are non-linear, it is difficult to directly characterize their limiting distributions. A notable exception where asymptotic results have been obtained for regularized estimators, is the LASSO estimator \cite{tibshirani1996regression} (and more generally the so called Bridge estimators) see \cite{Knight2000Asymptotics}; however, importantly, these asymptotic results are valid in the fixed dimensional setting and not in the high-dimensional setting, and moreover, are difficult to apply to draw inference or construct confidence intervals since the limiting distribution is not pivotal. This underscores that performing statistical inference is non-trivial in the high-dimensional setting. In a low-dimensional setting (where the need for regularization is less apparent), one can use large sample theory on an unregularized estimator (such as the OLS) to get an asymptotic result \citep{van2000asymptotic}. Even in low-dimensional settings however, if one chooses to use a constrained likelihood or more generally a constrained $M$-estimator, e.g., the asymptotic distribution may be highly non-trivial \cite{chernoff1954distribution, self1987asymptotic, geyer1994asymptotics}. A high-dimensional setting only exacerbates this issue, since as we mentioned, it necessitates the regularization. 

In high-dimensional models, one is often interested in one of three directions: oracle inequalities \citep{bunea2007sparsity, van2008high, bickel2009simultaneous}, variable selection \citep{meinshausen2006high, zhao2006model, fan2008sure}, and statistical inference \citep{van2014asymptotically, neykov2018unified, feng2019high}. The latter reference list is far from complete and we refer the reader to the excellent books by \cite{buhlmann2011statistics} and \cite{wainwright2019high} for a full introduction to high-dimensional statistics. Since in this paper we focus on the inference direction, below we review in depth only articles which are related to this direction. 

At first, the efforts of statisticians were naturally devoted to enable performing inference in the high-dimensional linear model, as it has ubiquitous applications in a variety of fields such as statistical gentics, bioinformatics, econometrics, finance, among many others. For instance, high-dimensional problems have been recently recognized in signal processing \citep{lustig2008compressed}, genetics \citep{peng2010regularized} and collaborative filtering \citep{koren2009matrix}. Early approaches of high-dimensional statistical inference were based on variable selection consistency \citep{wasserman2009high, meinshausen2010stability, shah2013variable}, which only works for sparse signal vectors. Specifically, the estimator is computed on the oracle set only, so the statistical inference is reduced to a low-dimensional setting. A limitation of this approach is that the variable selection consistency requires the magnitude of all non-zero coefficients to be greater than a threshold \citep{wainwright2009information, zhang2010nearly}, which may be unrealistic in many applications. The above reasoning motivated various approaches for deriving tractable and pivotal distributions for high-dimensional models which can be used to construct confidence intervals and draw inferences for individual coefficients. While there are approaches which consider a \textit{conditional hypothesis test} of the coefficients from a LASSO  \cite[among others]{lockhart2014significance, lee2014exact, lee2016exact}, in this paper we follow a line of work initiated by \cite{zhang2014confidence, van2014asymptotically, javanmard2014confidence, belloni2014inference, belloni2015uniform} where it was proposed how to correct the LASSO estimate (often called debiasing) in order to achieve asymptotic normality on individual coefficients. These works spurred a lot of follow-ups including \cite[among others]{ning2017general, jankova2015confidence, neykov2018unified, javanmard2018debiasing, jankova2018semiparametric}. Until recently, the majority of debiasing methods focused exclusively on $\ell_1$ penalized (generalized) linear models. Of note there is a recent exception which can handle more general penalties than the $\ell_1$ \citep{bellec2019second}. A notable limitation of this work however is that this debiasing scheme works only in the asymptotic regime $p/n \rightarrow \gamma$ for some constant $\gamma$, and furthermore it requires the knowledge of the covariance matrix $\bSigma$ of the predictors. Further, some other more recent works integrate a degrees-of-freedom adjustment to the debiasing procedure \citep{bellec2019biasing, celentano2020lasso}. This is something that we do not exploit in the current work, although we think there may be some promising connection between this idea and our algorithm. Finally we would like mention the work of \cite{bradic2018testability, zhu2018linear} which studies how one can perform inference in linear models where sparsity may be absent. This is related to our work in the sense that some models which we consider, like the monotone regression, are non-sparse. However, there is a big difference in the settings in that the algorithms given in \cite{bradic2018testability, zhu2018linear} work without having to respect the prior knowledge that the coefficients are monotone, e.g.

In this paper, we propose an abstract debiasing procedure for some regularized or constrained linear models. We illustrate that our procedure is applicable to convex constrained least squares with unknown covariance, in cases when the convex constraint set $K$ has a simple geometric structure. In addition, we demonstrate that our approach can successfully debias SLOPE and square-root SLOPE under the assumption that we have a known upper bound on the sparsity of the signal. Our debiasing approach relies on solving a cascade of two optimization problems. The first optimization restricts the initial coefficient estimator to have a small tangent cone, which is used to facilitate the second optimization program.
The second optimization is inspired by the work of \cite{javanmard2014confidence}. Specifically, the constraint set of this convex program is designed in such a way so that any feasible solution can be used for debiasing. Next the objective function is selected to minimize the variance of the limiting distribution of debiased estimator. 
Our second optimization uses a newly-designed constraint set in comparison with the LASSO debiasing approach from \cite[Algorithm 1]{javanmard2014confidence}. In the case of convex constrained least squares for example, our debiasing constraint is designed to respect the geometry of the constraint set $K$, which turns out to be the key for generalizing the debiasing from $\ell_1$-regularized problems to general constraint problems.
 
\subsection{Notation and Definitions}
% Gaussian Complexity 
Here we introduce some notation and concepts which will be used throughout the paper. Given a set $T\subset\mathbb{R}^p$, define its Gaussian complexity as 
\begin{align*}
    w(T)& = \mathbb{E}\sup\limits_{\xb\in T}\,\,\langle \gb,\xb\rangle, \quad\text{where}\,\,\gb \sim \mathcal{N}(0,\Ib_p).
\end{align*}
$w(T)$ is the expectation of maximum magnitude of the canonical Gaussian process on $T$. The Gaussian complexity is a basic geometric property of $T$. It measures the size of $T$ and is related to the metric entropy of $T$ \citep[Theorem~8.1.13]{vershynin2018high}. In additition to $w(T)$, throughout the paper we denote with $\overline w(T)$ any known (and ideally easily computable and as small as possible) upper bound of $w(T)$, i.e., $\overline w(T)$ satisfies:
\begin{align}\label{gaussian:width:upper:bound}
    w(T) \leq \overline w(T).
\end{align}
% Tangent Cone
Next we formalize the concept of a tangent cone which is frequently used in optimization. The tangent cone of a convex set $K\subset\mathbb{R}^p$ at $\xb\in K$ consists of all the possible directions from which a sequence in $K$ can converge to $\xb$. It is defined as 
\begin{align*}
    \mathcal{T}_K(\xb) & = \{ t(\vb-\xb)\,\,:\, t\geq 0, \vb\in K\}. 
\end{align*}
% Projection to convex cone
The projection of a vector $\vb \in \mathbb{R}^p$ onto a convex set $K\subset\mathbb{R}^p$ is defined as 
\begin{align*}
    \Pi_K(\vb) &= \arg\min\limits_{\xb\in K} \|\vb-\xb\|,
\end{align*}
where here and throughout we will use $\|\cdot\|$ as a shorthand for the Euclidean norm $\|\cdot\|_2$. Furthermore let $\|\cdot\|_{\operatorname{op}}$ denote the operator norm of a matrix. In addition we will also use $\wedge$ and $\vee$ as a shorthand for $\min$ and $\max$ of two numbers respectively, and $[n] = \{1,\ldots, n\}$ for an integer $n \in \NN$. We also make use of standard asymptotic notation: we write $X_n = o_p(1)$ if $\PP(|X_n| > \epsilon) \rightarrow 0$ for all $\epsilon > 0$, and $X_n = O_p(1)$ if for any $\epsilon >0$ there exists an $M > 0$ and a finite $N > 0$ such that $\PP(|X_n| > M) < \epsilon$ for all $n > N$. We write $X_n = o_p(a_n)$ if $X_n/a_n = o_p(1)$, and $X_n = O_p(a_n)$ if $X_n/a_n = O_p(1)$ for some non-zero sequence $\{a_n\}$. Furthermore, given two non-negative sequences $\{a_n\}, \{b_n\}$ we write $a_n = O(b_n)$ (or $a_n \lesssim b_n$) if there exists a constant $C < \infty$ such that for all $n > N$ for some $N \in \mathbb{N}$,  $a_n \leq C b_n$, $a_n = o(b_n)$ if $a_n/b_n \rightarrow 0$, and $a_n \asymp b_n$ if there exists positive constants $c$ and $C$ such that $c < a_n/b_n < C$. Finally, throughout the paper we will use $\eb^{(j)}$ to denote a vector with $0$ entries except on the $j$-th position where $\eb^{(j)}$ has an entry $1$. 

%%%%%%%%%%%
\subsection{Problem Formulation}
Suppose that we are given $n$ i.i.d. observations from a linear model
\begin{align}
Y_i = \bX_i\T \bbeta^* + \varepsilon_i, i \in [n],%\quad \bbeta^* \in K 
\label{gaussian_model}
\end{align}
where the predictors $\bX_i$ are also considered i.i.d. and random. For simplicity we assume that every observation $\bX_i$ is zero-mean (i.e. the covariates are centered). This can always be achieved at the price of splitting the data evenly and subtracting the $Y_i$ and $\bX_i$ values from the first half from those values of the second half (this not only ensures that $\bX_i$ will be zero-mean but also preserves other subsequent assumptions that we make on the data). In addition we will require that $\bX_i$ is a sub-Gaussian random variable (see Definition \ref{sub:gaussian:def}) with covariance $\bSigma$. Furthermore, for the most part of the manuscript we will assume that $\varepsilon_i \sim N(0,\sigma^2)$ in order to simplify the presentation. In Section \ref{subGaussian:noise:section} we elaborate on a slight modification of our procedure, inspired by \cite{javanmard2014confidence}, that can handle general sub-Gaussian noise. Additionally, for the most part we require that $\varepsilon_i$ are independent of $\bX_i$.

Suppose now that instead of fitting OLS to \eqref{gaussian_model}, a practitioner fits a reguralized or constrained least squares estimator. An example where such a situation may arise is when the practitioner has prior knowledge that $\bbeta^* \in K$ for some fixed and known convex set $K$. In such a setting the practitioner may opt for outputting the following natural estimate of $\bbeta^*$:
\begin{align}
\hat \bbeta = \argmin_{\bbeta \in K} n^{-1}\sum_{i \in [n]}(Y_i - \bX_i\T \bbeta)^2. 
\label{constrained_ls}
\end{align}
In addition, especially in settings when $p \gg n$ and an assumption on the sparsity of $\bbeta^*$ is appropriate, the practitioner may opt for running a regularized procedure such as LASSO \citep{tibshirani1996regression}, SLOPE \citep{bogdan2015slope} or square-root SLOPE \citep{stucky2017sharp}. Unlike the OLS, constraints or reguralizations incur bias on $\hat\bbeta$, and make the limiting distribution of $\hat\bbeta$ complicated. Thus performing statistical inference on $\hat\bbeta$ becomes non-straightforward. 

The goal of the present paper is to develop what became known as debiasing techniques for $\hat \bbeta$ in such scenarios. In particular we would like to construct confidence intervals for any bounded contrast of $\bbeta^*$ (i.e. $\bgamma\T \bbeta^*$ with $\|\bgamma\| < B < \infty$) --- using a non-OLS pilot estimator $\hat \bbeta$ of $\bbeta^*$ in \eqref{gaussian_model} --- in a high-dimensional setting. It is worthy to mention that the majority of previous works on debiasing focus exclusively on debiasing $\ell_1$-penalized regression. There are some exceptions such as \cite{bellec2019second}, but their setting is substantially different from the present work. 

The algorithm proposed in this paper is capable of debiasing any estimator $\hat \bbeta$ which can be used to produce the following quantities:
\begin{itemize}
    \item An estimator $\vb$ of a vector sufficiently close to $\bbeta^*$ (or ideally $\bbeta^*$ itself) in the $\ell_2$ sense.
    \item A convex set $K$ such that $\vb, \bbeta^* \in K$ (here $K$ may be given or may be constructed from $\hat \bbeta$).
    \item $\vb$ is a boundary point in $K$ such that the tangent cone of $K$ at $\vb$ is sufficiently small.
\end{itemize}
We will make use of sample splitting to produce $\vb$ and $K$ from $\hat \bbeta$ on one half of the sample, and estimate a projection direction used in the debiasing on the other half. For more detailed information on our abstract procedure refer to Section \ref{algo:sec}.

Finally we mention that our debiasing procedure does not require prior knowledge of the inverse population covariance matrix --- $\bSigma^{-1}$ --- which is known to make inference easier \citep{javanmard2018debiasing, bellec2019second}.

%%%%%%%%%%%
\subsection{Paper Organization}
The paper is structured as follows. Section \ref{algo:sec} describes our abstract debiasing procedure and shows how the program from the second step can be solved with subgradient descent. Section \ref{asymp:sec} proves the main theorem of the paper and provides a confidence interval construction. Section \ref{cvs_ls:sec} is dedicated to convex constrained least squares, where we formally describe how one can solve step 1 of our abstract debiasing procedure in such a setting. In Section \ref{minkowski:section} we show how our algorithm can be used to debias a Minkowski gauge selector with unknown covariance. Section \ref{minkowski:gauge:regularization:sec} considers Minkowski gauge regularization: it shows a new result for this type of optimization program when the design is Gaussian, and shows how one can debias the regularized group LASSO. Section \ref{slope_sqrtslope:sec} discusses applications to SLOPE and square-root SLOPE. Section \ref{subGaussian:noise:section} contains an extension to non-Gaussian noise. Section \ref{simulations:sec} illustrates our results with some numerical studies and finally in Section \ref{discussion:section} we give a brief discussion. All technical proofs are deferred to the supplement.

%%%%%%%%%%%%%%%%%%%%
%% Sec 2: Algo
%%%%%%%%%%%%%%%%%%%%
\section{The Debiasing Algorithm}
\label{algo:sec}
In this section we propose an optimization-based Algorithm \ref{algo_unknowncov} as a general procedure to debias an individual coordinate, as well as any contrast of $\bbeta^*$ using a non-OLS estimator $\hat\bbeta$. Then in Section \ref{lambda_pick:subsec} and \ref{subgrad:subsec} we provide details for how to solve the optimization problem in step 2 of the proposed Algorithm \ref{algo_unknowncov}.

%%%%%%%%%%%%%%%%%%
\subsection{The Debiasing Algorithm}
For simplicity of the presentation, and without loss of generality we will assume that we are given $2n$ samples from model \eqref{gaussian_model}. If the actual number of samples is odd we can simply drop one sample. We randomly split the data set $(\Xb, \bY)$ where $\Xb = (\bX_1, \ldots, \bX_{2n})\T$, $\bY = (Y_1, \ldots, Y_{2n})\T$ into two equally-sized partitions $(\overline\Xb, \overline\bY)$ and $(\tilde\Xb, \tilde\bY)$.  The first half $(\overline\Xb, \overline\bY)$ is used to obtain an estimator $\hat\bbeta$ of the true coefficient $\bbeta^*$, and then is used to obtain $\vb$ and $K$. The second half $(\tilde\Xb, \tilde\bY)$ is used to construct the debiased $\hat\bbeta_d$ based on $\vb$ and $K$.

Step 1 of Algorithm \ref{algo_unknowncov} uses the first half of the data to construct a vector $\vb$ which is close to $\bbeta^*$ in $\ell_2$-distance, and a convex set $K$ which has a small tangent cone at $\vb$. In all of our examples to follow, such a construction uses a pilot estimator $\hat \bbeta$ which can be a constrained or reguralized estimator. We therefore view our procedure as a procedure for debiasing the pilot vector $\hat \bbeta$, but in principle one may bypass estimating $\hat \bbeta$ and may use the first half of the data to directly find $\vb$ and $K$ obeying the desired properties. 

Next we solve an optimization program (see step 2 of Algorithm \ref{algo_unknowncov}) to get an auxiliary vector $\hat\bmeta$ which is used in the final debiasing formula as a proxy to the $j$-th row of $\bSigma^{-1}$. In fact, as implied by Theorem \ref{debiase_formula_applicable_unkowncov}, any feasible point of the optimization program in step 2 would successfully produce an asymptotically normal debiased estimator. In other words, the limiting distribution of $\sqrt{n}(\hat\bbeta_d^{(j)} - \bbeta^{*(j)})$ would be a zero-mean Gaussian random variable, but its variance might be large. To achieve a small variance for the limiting distribution, we pick the objective function in the optimization of step 2 to minimize such a variance, which is inspired by \cite[Algorithm 1]{javanmard2014confidence}. The following Algorithm \ref{algo_unknowncov} summarizes our debiasing procedure.
\begin{algorithm}%[H]
\caption{Debiasing the $j$\textsuperscript{th} Coordinate of A Non-Ordinary Least Squares Estimator}
\label{algo_unknowncov}
\begin{algorithmic}
\STATE \textbf{Input:} Two equal size partitions of the data $(\overline\Xb, \overline\bY)$ and $(\tilde\Xb, \tilde\bY)$.%, %$\hat\bbeta$ obtained using $(\overline\Xb, \overline\bY)$.\\

\STATE \textbf{Initialize:} Empirical Gram matrix of the second partition $\hat\bSigma=\frac{1}{n}\tilde\Xb\T\tilde\Xb$.

\begin{enumerate}
    \item \label{step_1}
    Using the first data split find a convex set $K$ and a vector $\vb$ on the boundary of $K$, such that: $\vb, \bbeta^* \in K$ with high probability, and $\overline w(\cT_K(\vb)\cap \mathbb{S}^{p-1})\|\vb-\bbeta^*\|=o_p(1)$.
    \item \label{step_2}
    The debiased $j$\textsuperscript{th} coefficient $\hat\bbeta_{d}^{(j)} \leftarrow \eb^{(j)\top}\vb + n^{-1}\hat\bmeta\T \tilde\Xb\T (\tilde\bY - \tilde\Xb\vb)$, where $\hat\bmeta$ is computed by \\
    \begin{align}\hat\bmeta \leftarrow \argmin_{\bmeta}\, \|\hat\bSigma^{\frac{1}{2}}\bmeta\| \mbox{   subject to  } \displaystyle\sup_{\ub \in \cT_K(\vb) \cap \mathbb{S}^{p-1}} |(\bmeta\T \hat \bSigma - \eb^{(j)\top}) \ub|  \leq \frac{\rho\overline w(\cT_K(\vb)\cap\mathbb{S}^{p-1})}{\sqrt{n}},\label{opt_step3}
    \end{align}
    for some sufficiently large tuning parameter $\rho>0$.
\end{enumerate}
\end{algorithmic}
\end{algorithm}

\begin{remark}
Several remarks regarding Algorithm \ref{algo_unknowncov} are in order. First we comment on step 1. One may wonder how to construct a set $K$ and vector $\vb$ with the desired properties, and if that is even possible. While it is hard to answer this without having a concrete example at hand, we will give a couple of comments. The set $K$ may be naturally given to the practitioner --- for example it may be the constraint set if the practitioner is solving convex constrained least squares. On the other hand, a set $K$ could be constructed via the vector $\hat \bbeta$. If $\hat \bbeta$ for instance is known to satisfy $\|\hat \bbeta - \bbeta^*\| \leq b(n,p, \bbeta^*)$ for some explicitly quantifiable upper bound $b(n,p,\bbeta^*)$ one may start the construction of $K$ based on the Euclidean ball around $\hat \bbeta$ with radius $b(n,p,\bbeta^*)$ (for more details on approach this we refer to Section \ref{slope_sqrtslope:sec} where we build a convex set $K$ for the SLOPE and square-root SLOPE estimators). The vector $\vb$ on the other hand should be selected to respect the geometry of $K$ and will likely have to possess additional properties (e.g. sparsity or other adequate restrictions which make the tangent cone at it small). We provide a detailed process of finding $\vb$ for each type of estimator $\hat\bbeta$ in our examples; see Sections \ref{cvs_ls:sec} --- \ref{slope_sqrtslope:sec}. 
 
We now comment on the condition $\overline w(\cT_K(\vb)\cap \mathbb{S}^{p-1})\|\vb-\bbeta^*\|=o_p(1)$ required in step 1. Intuitively, we need $\|\vb-\bbeta^*\|$ to be small because in the final step the debiased estimator $\hat\bbeta_d$ is constructed from $\vb$. A small upper bound on the Gaussian complexity of the tangent cone $\overline w(T_K(\vb)\cap\mathbb{S}^{p-1})$ is needed to guarantee fast convergence rate of the debiased estimator $\hat\bbeta_d$, and fast computation of the optimization in step 2. 

Finally we comment on step 2. Step 2 of our abstract procedure is reminiscent of previous ideas on debiasing which attempt to estimate the inverse covariance (aka precision) matrix along a direction of interest. We stress on the fact that our proposal is distinct from previous works however, and even in the ``classical'' example of LASSO will produce a distinct projection direction $\hat \bmeta$. In addition, we mention that if one is interested in performing inference on general bounded contrasts of $\bbeta^*$, i.e., $\bgamma\T \bbeta^*$ for some $\|\bgamma\| \leq B$ with a finite $B$, step 2 can be readily modified by changing $\hat \bmeta$ to
\begin{align*}\hat\bmeta \leftarrow \argmin_{\bmeta}\, \|\hat\bSigma^{\frac{1}{2}}\bmeta\| \mbox{   subject to  } \displaystyle\sup_{\ub \in \cT_K(\vb) \cap \mathbb{S}^{p-1}} |(\bmeta\T \hat \bSigma - \bgamma\T) \ub|  \leq  \frac{\rho\overline w(\cT_K(\vb)\cap\mathbb{S}^{p-1})}{\sqrt{n}}.
\end{align*}
For simplicity of presentation we stick to our formulation with $\eb^{(j)}$ but all of our proofs and results can be easily modified to the more general setting described above by changing $\eb^{(j)}$ to $\bgamma$.
\end{remark}

In the next two subsections, we address two questions regarding the optimization \eqref{opt_step3} of step 2 of Algorithm \ref{algo_unknowncov}. The first question is whether the constraint in \eqref{opt_step3} is empty. In Section \ref{lambda_pick:subsec} we will show that \eqref{opt_step3} is guaranteed to have a feasible point with high probability, and furthermore the interior of such a constraint is not empty if $\rho$ is sufficiently large.

In addition, the above optimization \eqref{opt_step3} can be solved by subgradient descent. An explicit formula of the subgradient is complicated by the unconventional constraint, which makes the program in step 2 a semi-infinite program. See \cite{hettich1993semi} for details about semi-infinite programming. Section \ref{subgrad:subsec} gives out the explicit formula of the subgradient, and proves the convergence of such a subgradient descent method.

%%%%%%%%%%%%%%%%%%
\subsection{Studying the Constraint Set of Step 2}
\label{lambda_pick:subsec}
We begin by showing that $\bmeta = \bSigma^{-1}\eb^{(j)}$ is a feasible point of the optimization \eqref{opt_step3}. In fact, the right hand side of the constraint in \eqref{opt_step3} --- $\frac{\rho\overline w(\cT_K(\vb)\cap\mathbb{S}^{p-1})}{\sqrt{n}}$ --- is inspired by analyzing the magnitude of $\sup_{\ub \in \cT_K(\vb) \cap\mathbb{S}^{p-1}} |(\bmeta\T \hat \bSigma - \eb^{(j)\top}) \ub|$ when evaluated at $\bmeta = \bSigma^{-1}\eb^{(j)}$. The intuition is that $\hat\bmeta$ is a proxy of $\bSigma^{-1}\eb^{(j)}$. This idea is of course standard and central in all previous debiasing works, but the challenge in our setting is to analyze the empirical process $\sup_{\ub \in \cT_K(\vb) \cap\mathbb{S}^{p-1}} |(\bmeta\T \hat \bSigma - \eb^{(j)\top}) \ub|$ at $\bmeta=\bSigma^{-1}\eb^{(j)}$. To this end we will use a powerful result due to \cite{mendelson2016upper}. Before we state our result we first formally define sub-Gaussian random vectors.
\begin{definition}\label{sub:gaussian:def} A vector $\bZ \in \RR^p$ is called sub-Gaussian if there exists a constant $C \in \RR_+$ such that for any unit vector $\wb \in \mathbb{S}^{p-1}$ and any $\lambda \in \RR$, $\EE \exp(\lambda (\bZ - \EE \bZ)\T\wb) \leq \exp(\lambda^2 C)$.
\end{definition}

\begin{lemma}
Suppose that $\Xb=(\bX_1,...,\bX_n)\T$ where every observation $\bX_i$ is a zero-mean sub-Gaussian random variable with covariance matrix $\bSigma$. Furthermore, suppose the eigenvalues of $\bSigma$ are bounded from above and below, in the sense that there exist absolute constants $0< c < C$ such that $c < \lambda_{\min}(\bSigma) \leq \lambda_{\max}(\bSigma) < C$. Let $\hat\bSigma = \frac{1}{n}\Xb\T\Xb$ be the empirical Gram matrix. Suppose that the upper bound $\overline w(\cT_K(\vb) \cap \mathbb{S}^{p-1})$ is chosen so that $\overline w(\cT_K(\vb) \cap \mathbb{S}^{p-1}) \rightarrow \infty$ as $n\rightarrow\infty$. Then for $\bmeta\T = \eb^{(j)\top}\bSigma^{-1}$, with probability converging to one (see the proof for a precise expression) we have
\begin{align*}
    \sup_{\ub \in \cT_K(\vb) \cap\mathbb{S}^{p-1}} |(\bmeta\T \hat \bSigma - \eb^{(j)\top}) \ub| \lesssim  \frac{\overline w(\cT_K(\vb)\cap\mathbb{S}^{p-1})}{\sqrt{n}}.
\end{align*}
\label{feasible_point}
\end{lemma}

\begin{remark} Lemma \ref{feasible_point} requires that $\overline w(\cT_K(\vb) \cap \mathbb{S}^{p-1}) \rightarrow \infty$. Since $\vb$ is random, it is convenient to assume this holds for all $\vb$. If one knows an upper bound on $w(\cT_K(\vb) \cap \mathbb{S}^{p-1}) \leq u(\cT_K(\vb) \cap \mathbb{S}^{p-1})$ for all vectors $\vb$, obtaining a diverging upper bound is simple: just take $\overline w(\cT_K(\vb) \cap \mathbb{S}^{p-1}) = u(\cT_K(\vb) \cap \mathbb{S}^{p-1}) \vee a_n$ for any slowly diverging sequence $a_n$. For future reference we will always assume that $\overline w(\cT_K(\vb) \cap \mathbb{S}^{p-1})$ is constructed in such a way, and we do not explicitly mention the term ``$\vee a_n$'' later on. In addition we will implicitly be assuming that we have $\overline w(\cT_K(\vb) \cap \mathbb{S}^{p-1}) \rightarrow \infty$, and we will omit stating this assumption sometimes. It may be confusing why this assumption is required since intuitively one would want to obtain as tight bound to $w(\cT_K(\vb) \cap \mathbb{S}^{p-1})$ as possible. The reason for this is because the ``converging to one'' probability needs this condition. In most practical applications (as can be seen by our examples below) the quantity $w(\cT_K(\vb) \cap \mathbb{S}^{p-1})$ will naturally be diverging so there is no need to add additional terms $\vee a_n$. However, in cases where this is needed $a_n$ should be chosen as small as possible, such as $a_n = \log \log n$ e.g.  

Note that in the result of Lemma \ref{feasible_point}, $\cT_K(\vb) \cap \mathbb{S}^{p-1}$ can be substituted by a general compact set in $\mathbb{R}^p$ since the proof of Lemma \ref{feasible_point} does not rely on the the fact that $\cT_K(\vb)$ is a cone. Here we stated the lemma with $\cT_K(\vb) \cap  \mathbb{S}^{p-1}$ because this is the only set of interest for us. Also, the result of Lemma \ref{feasible_point} still holds if $\eb^{(j)}$ is replaced by any other unit norm vector, which supports the generalization of Algorithm \ref{algo_unknowncov} to debias a linear combination of coordinates. See also Remark \ref{debias_linear_comb}.
\end{remark}

The following Corollary proves that the constraint of \eqref{opt_step3} has a non-empty interior. It is a sufficient condition for the convergence of the subgradient descent in the next section. 
\begin{corollary}[\textbf{Non-empty Interior of the Constraint}]
Under the same assumptions of Lemma \ref{feasible_point} the set 
\begin{align*}
Q = \bigg\{\bmeta:\, \sup_{\ub \in \cT_K(\vb) \cap\mathbb{S}^{p-1}} |(\bmeta\T \hat \bSigma - \eb^{(j)\top}) \ub| \leq \frac{\rho\overline w(\cT_K(\vb)\cap\mathbb{S}^{p-1})}{\sqrt{n}} \bigg\},
\end{align*}
has a non-empty interior with high probability for sufficiently large $\rho$.
\label{nonempty_int_Q}
\end{corollary}

%%%%%%%%%%%%
\subsection{Solving the Optimization Problem \texorpdfstring{\eqref{opt_step3}}{} by Subgradient Descent}
\label{subgrad:subsec}
We will now explain how to solve the optimization program \eqref{opt_step3} by subgradient descent for constrained optimization. We implicitly assume in this section that the projection $\Pi_{\cT_K(\vb)}$ can be computed in a reasonable time. This may not always hold in practice due to the fact that both the set $K$ and estimator $\vb$ are random variables and depend on the first sample split. However we note that in all of our examples to be considered (see Sections \ref{cvs_ls:sec} --- \ref{slope_sqrtslope:sec}) this projection is indeed feasible and can be computed fast. In addition finding a projection on a convex set is always a convex optimization problem, which can be solved in principle. Define
\begin{align}\label{psi:eta:Q:notation}
    \psi(\bmeta) = \sup_{\ub \in \cT_K(\vb) \cap\mathbb{S}^{p-1}} |(\bmeta\T \hat \bSigma - \eb^{(j)\top}) \ub| - \frac{\rho\overline w(\cT_K(\vb)\cap\mathbb{S}^{p-1})}{\sqrt{n}}.
\end{align}
The constraint in \eqref{opt_step3} can be written as $Q = \{\bmeta:\, \psi(\bmeta)\leq0\}$. According to \cite[Section 7]{boyd2003subgradient}, the subgradient descent moves towards the optima by generating a sequence $\{\bmeta_n\}$ as
\begin{align}
\label{sg_step}
    \bmeta_{n+1} = \bmeta_n - h_{n}\gb_n,
\end{align}
where $h_n$ is the step size, and $\gb_n$ is the gradient of the objective function $f(\bmeta)=\|\hat\bSigma^{\frac{1}{2}}\bmeta\|$ if $\bmeta_n \in Q$; otherwise is a subgradient of the constraint function $\psi(\bmeta)$ if $\bmeta_n \notin Q$. Put
\begin{align*}
    \phi_0(\bmeta) = \frac{\Pi_{\cT_K(\vb)}(\hat\bSigma\bmeta - \eb^{(j)})}{\|\Pi_{\cT_K(\vb)}(\hat\bSigma\bmeta - \eb^{(j)})\|}, \quad \phi_1(\bmeta) = \frac{\Pi_{-\cT_K(\vb)}(\hat\bSigma\bmeta - \eb^{(j)})}{\|\Pi_{-\cT_K(\vb)}(\hat\bSigma\bmeta - \eb^{(j)})\|},\footnotemark
\end{align*}
\footnotetext{We understand $0/0 = 0$ in case when the projection is the zero vector.}
where $-\cT_K(\vb) = \{-\bbeta : \bbeta \in \cT_K(\vb)\}$. Lemma \ref{sg_psi} below, shows that the explicit form of $\gb_n$ is given by: 
\begin{align}
\label{sg_expression}
    \gb_n = \begin{cases}
    \hat\bSigma \bmeta_n/\|\hat\bSigma^{\frac{1}{2}}\bmeta_n\| & \text{, if } \bmeta_n \in Q\\
    \hat\bSigma\phi_{\one\{(\bmeta_n\T \hat \bSigma - \eb^{(j)\top}) (\phi_0(\bmeta_n) - \phi_1(\bmeta_n)) < 0\}}(\bmeta_n)
    & \text{, if } \bmeta_n \notin Q.
    \end{cases}
\end{align}
It is clear that the first expression in \eqref{sg_expression} for $\bmeta_n \in Q$ is the gradient of the objective function $f(\bmeta)=\|\hat\bSigma^{\frac{1}{2}}\bmeta\|$ at $\bmeta_n$ when $\bmeta_n \neq 0$. If $\bmeta_n$ turns out to be $0$, $\bg_n$ can be taken as $\hat \bSigma^{1/2}\wb$ for any unit vector $\wb$. However, if $\bmeta_n = 0$ is a feasible point, it is necessarily an optimal value so that the algorithm should terminate. In Lemma \ref{sg_psi} we show that the second expression in \eqref{sg_expression} is a subgradient of $\psi(\bmeta)$ at $\bmeta_n$ when $\bmeta_n \notin Q$.
\begin{lemma}
\label{sg_psi}
For $\bmeta_n \notin Q$, the expression of $\gb_n$ at \eqref{sg_expression} is a subgradient of $\psi(\bmeta)$ at $\bmeta_n$.
\end{lemma}
We observe that if one can compute $\Pi_{\cT_K(\vb)}$, one can clearly compute 
\begin{align}\label{neg:cone:calc}
    \Pi_{-\cT_K(\vb)}(\xb) = -\argmin_{\wb \in \cT_K(\vb)}\|\wb - (-\xb)\| = -\Pi_{\cT_K(\vb)}(-\xb),
\end{align}
We provide Algorithm \ref{algo_step3} as a summary of solving \eqref{opt_step3}, assuming $\Pi_{\cT_K(\vb)}$ is computable in a reasonable time. In Sections \ref{cvs_ls:sec} --- \ref{slope_sqrtslope:sec} we will see that such a projection $\Pi_{\cT_K(\vb)}$ can be obtained efficiently for some specific convex cones with a simple structure.

\begin{algorithm}%[H]
\caption{Solve the Optimization \eqref{opt_step3} in Step \ref{step_2} of Algorithm \ref{algo_unknowncov}}
\label{algo_step3}
\begin{algorithmic}
\STATE \textbf{Input:} The convex set $K$, the vector $\vb$ from step 2, empirical Gram matrix of the second partition $\hat\bSigma=\frac{1}{n}\tilde\Xb\T\tilde\Xb$.\\
\STATE \textbf{Initialize:} $\bmeta_1$\\% $\bmeta_{out}$. Make sure that $\bmeta_{out}$ is feasible.\\
\STATE Run until some convergence criteria is satisfied:\\
\STATE\hspace*{\algorithmicindent} Compute $P_{+} \leftarrow\Pi_{\cT_{K}(\vb)}(\hat\bSigma\bmeta_{n} - \eb^{(j)})$, $P_{-} \leftarrow\Pi_{-\cT_{K}(\vb)}(\hat\bSigma\bmeta_{n} - \eb^{(j)})$.\\
\STATE\hspace*{\algorithmicindent} \algorithmicif{ $ \max\{\|P_+\|, \|P_-\|\}\leq \frac{\rho\overline w(\cT_K(\vb)\cap\mathbb{S}^{p-1})}{\sqrt{n}}$}\\
\STATE\hspace*{\algorithmicindent}\hspace*{\algorithmicindent} \algorithmicif{\,\,\ $\|\hat\bSigma^{\frac{1}{2}}\bmeta_{n}\| \leq \|\hat\bSigma^{\frac{1}{2}}\bmeta_{out}\|$:} $\bmeta_{out}\leftarrow \bmeta_{n}$\\
\STATE\hspace*{\algorithmicindent}\hspace*{\algorithmicindent} $\bmeta_{n+1} \leftarrow \bmeta_{n} - h_n\frac{\hat\bSigma \bmeta_{n}}{\|\hat\bSigma^{\frac{1}{2}} \bmeta_{n}\|}$\\
% Else
\STATE\hspace*{\algorithmicindent} \algorithmicelse:\\
\STATE\hspace*{\algorithmicindent}\hspace*{\algorithmicindent} $\phi_0(\bmeta_{n}) \leftarrow P_+ \,/\, \|P_+\|$\\
\STATE\hspace*{\algorithmicindent}\hspace*{\algorithmicindent} $\phi_1(\bmeta_{n}) \leftarrow P_- \,/\, \|P_-\|$.\\
\STATE\hspace*{\algorithmicindent}\hspace*{\algorithmicindent} $\bmeta_{n+1} \leftarrow \bmeta_{n} - h_n \hat\bSigma\phi_{\one\{(\bmeta_{n}\T \hat \bSigma - \eb^{(j)\top}) (\phi_0(\bmeta_{n-1}) - \phi_1(\bmeta_{n-1})) < 0\}}(\bmeta_{n})$\\
\STATE $\hat\bmeta \leftarrow \bmeta_{out}$.
\end{algorithmic}
\end{algorithm} 
We note that the condition 
\begin{align*}
    \max\{\|P_+\|, \|P_-\|\}\leq \frac{\rho\overline w(\cT_K(\vb)\cap\mathbb{S}^{p-1})}{\sqrt{n}},
\end{align*}
used in Algorithm \ref{algo_step3} is equivalent to checking feasibility, i.e., checking
\begin{align*}
    \psi(\bmeta_n) \leq 0,%\frac{\rho\overline w(\cT_K(\vb)\cap\mathbb{S}^{p-1})}{\sqrt{n}},
\end{align*}
since $\langle\Pi_{\cT_{K}(\vb)}(\hat\bSigma\bmeta_{n} - \eb^{(j)}),\hat\bSigma\bmeta_{n} - \eb^{(j)}\rangle  = \|\Pi_{\cT_{K}(\vb)}(\hat\bSigma\bmeta_{n} - \eb^{(j)})\|^2$ as can be seen from Lemma \ref{project_argsup} in the supplementary material. In practice, one would like to pick $\rho$ sufficiently large so that there exists a feasible point, yet not overly large since it is important for the theory $\rho$ to remain bounded. 

Let $\bmeta^* = \argmin_{\bmeta\in Q} \|\hat\bSigma^{\frac{1}{2}}\bmeta\|$ be the constrained minima of \eqref{opt_step3}. It is proved in Lemma \ref{cvg_sbgrad} that there exists a subsequence of $\{\bmeta_n\}$ in \eqref{sg_step} converging to $\bmeta^*$, and it takes $n=O(1/\epsilon^2)$ iterations to get an $\epsilon$-suboptimal solution, i.e. $\|\hat \bSigma^{\frac{1}{2}}\bmeta_n\|-\|\hat \bSigma^{\frac{1}{2}}\bmeta^*\|\leq\epsilon$. Therefore the subgradient descent is an appropriate method for solving program \eqref{opt_step3}. As we mentioned earlier the constraint in the optimization program \eqref{opt_step3} is unconventional since the $\sup$ can be regarded as infinite number of constraints. Such programs are called semi-infinite programs. The proof of Lemma \ref{cvg_sbgrad} is inspired by \cite[Section 7]{boyd2003subgradient} which is suitable for unconventional constraints. For completeness we also mention that  \cite{polyak1967general} was the first to prove the convergence of subgradient descent with rather general constraints.
\begin{lemma}[Convergence of subgradient descent]
For any bounded starting point $\bmeta_1$, one can construct a sequence $\{\bmeta_n\}$ by \eqref{sg_step}, \eqref{sg_expression}. As detailed in Algorithm \ref{algo_step3} at every step of the iteration, we record the best candidate found so far as
\begin{align*}
    \bmeta_n^{best} = \argmin_{\bmeta_i} \big\{\|\hat\bSigma^{\frac{1}{2}}\bmeta_i\|\,\big|\,\bmeta_i\in Q,\,i\in[n]\big\}.
\end{align*}
Let $\bmeta^*$ achieve the minima of \eqref{opt_step3} and $h_n$ be the step size of the subgradient descent. Suppose we run Algorithm \ref{algo_step3} for $k$ iterations. Then for some absolute constants $C_1, C_2$, 
\begin{align*}
    \epsilon: = \|\hat\bSigma^{\frac{1}{2}}\bmeta_k^{best}\| - \|\hat\bSigma^{\frac{1}{2}}\bmeta^*\| \lesssim \frac{C_1^2 + C_2^2\sum_{n=1}^k h_n^2}{\sum_{n=1}^k h_n}.
\end{align*}
\label{cvg_sbgrad}
\end{lemma}
For $h_n$ satisfying $\sum_{n=0}^{+\infty}h_n = +\infty$ and $\sum_{n=0}^{+\infty}h_n^2 = o(\sum_{n=0}^{+\infty}h_n)$, we have $\epsilon\rightarrow0$ so that $\lim_{n\rightarrow\infty}\|\hat\bSigma^{\frac{1}{2}}\bmeta_n^{best}\| = \|\hat\bSigma^{\frac{1}{2}}\bmeta^*\|$, which implies the convergence of the subgradient descent. Moreover, different choices of the step size $h_n$ give different convergence rates. For example, if $h_n=1/\sqrt{n}$, the convergence rate is nearly quadratic as $k=O(\log^2{k}/\epsilon^2)$; if $h_n=h \asymp 1/\sqrt{k}$ is a fixed small constant, the exact quadratic convergence rate $k=O(1/\epsilon^2)$ is achieved (although the algorithm does not converge to the target if ran for infinitely many iterations in this case).

%%%%%%%%%%%%%%
%% Sec 3
%%%%%%%%%%%%%
\section{Asymptotic Distribution and Confidence Interval of the Debiased Estimator}
\label{asymp:sec}
In this section we derive the limiting distribution of the debiased estimator obtained by Algorithm \ref{algo_unknowncov}. We then construct a confidence interval using a consistent estimator of $\sigma$ --- the standard deviation of the noise $\varepsilon$. The following Theorem \ref{debiase_formula_applicable_unkowncov} shows that Algorithm \ref{algo_unknowncov} successfully debiases the $j$-th coordinate of an estimator of $\bbeta^*$ given model \eqref{gaussian_model}, when the population covariance matrix $\bSigma$ has bounded spectrum. 
\begin{theorem}
\label{debiase_formula_applicable_unkowncov}
Consider a linear model in \eqref{gaussian_model} with Gaussian errors $\varepsilon_i \sim N(0,\sigma^2)$. Suppose the eigenvalues of $\bSigma$ are bounded from both above and below.
Then, under the assumptions of Lemma \ref{feasible_point} and $\overline w(\cT_K(\vb)\cap \mathbb{S}^{p-1})\|\vb-\bbeta^*\|=o_p(1)$, the debiased $j$\textsuperscript{th} coefficient $\hat\bbeta_{d}^{(j)}$ obtained by Algorithm \ref{algo_unknowncov} is conditionally asymptotically normal with mean equal to $\bbeta^{*(j)}$. In particular, if $Z_j = \frac{1}{\sqrt{n}}\hat\bmeta\T \tilde\Xb\T\bvarepsilon$, we have
{\footnotesize
\begin{align*}
    \sqrt{n}(\hat\bbeta_d^{(j)} - \bbeta^{*(j)}) = Z_j + \Delta_j,\quad
    Z_j|\overline\Xb, \overline\bY, \tilde\Xb \sim N(0, \sigma^2\hat\bmeta\T\hat\bSigma\hat\bmeta),\quad
    \Delta_j = \sqrt{n}(\hat \bmeta\T \hat \bSigma - \eb^{(j)\top})(\bbeta^* - \vb),
\end{align*}
}
and $\Delta_j = o_p(1)$.
\end{theorem}

\begin{remark}
\label{debias_linear_comb}
We will reiterate that our debiasing procedure works for a linear combination of coordinates (i.e. a contrast). It is not hard to see from the proof of Lemma \ref{feasible_point} and Theorem \ref{debiase_formula_applicable_unkowncov} that if we replace $\eb^{(j)}$ by any bounded in Euclidean norm vector, the same results will also hold. In terms of implementation, to debias a contrast, one simply needs to replace $\eb^{(j)}$ by the relevant vector with bounded norm in step \ref{step_2}.
\end{remark}

\begin{remark}
For simplicity of exposition the above theorem assumes that the errors are Gaussian. Our procedure also works with non-Gaussian errors using a modification similar in spirit to the one proposed in \cite[Section 4]{javanmard2014confidence}. Details will be given in Section \ref{algo_unknowncov:subgaussian:noise}.
\end{remark}

\begin{remark} In this remark we explain the validity of our procedure if one is interested in testing multiple coordinates simultaneously. Suppose $S \subset [p]$ is a set of coordinates of interest. Let $\hat \Mb$ represent the matrix whose rows are vectors $\hat \bmeta_j\T$ for $j \in S$ which are obtained via solving \eqref{opt_step3} for $j \in S$. We then have
{\footnotesize
\begin{align*}
\sqrt{n}(\hat \bbeta_d^{S} - \bbeta^{*S}) = \bZ + \bDelta,\quad
    \bZ|\overline\Xb, \overline\bY, \tilde\Xb \sim N(0, \sigma^2\hat\Mb\hat\bSigma\hat \Mb\T),\quad
    \bDelta = \sqrt{n}(\hat \Mb \hat \bSigma - \Eb^{S})(\bbeta^* - \vb),
\end{align*}
}where $\|\bDelta\|_{\infty} = o_p(1)$, $\Eb^S$ collects all vectors $\eb^{(j)\top}$ in its rows, and super-indexing $\bbeta^{*S}$ and $\hat \bbeta^S_d$ by the set $S$ collects all coordinates belonging to the set $S$, and $\hat\bbeta_{d}^{(j)} = \eb^{(j)\top}\vb + n^{-1}\hat\bmeta_j\T \tilde\Xb\T (\tilde\bY - \tilde\Xb\vb)$ for $j \in S$. Here we require $C'|S|\log^4 n/n + |S| \exp(-C''_{\bSigma}\overline w(\cT_K(\vb) \cap \mathbb{S}^{p-1})) = o(1)$, where $C', C''_{\bSigma}$ are absolute constants ($C''_{\bSigma}$ may depend on $\bSigma$'s spectrum). This latter condition is needed since we want to ensure that all points $\eb^{(j)\top}\bSigma^{-1}$ will be feasible points in the program \eqref{opt_step3} with high probability (see also the proof of Lemma \ref{feasible_point}). 
\end{remark}

%%%%%%%%%%%%%%
\subsection{Confidence Intervals}
Based on Theorem \ref{debiase_formula_applicable_unkowncov}, a $(1-\alpha)$-level confidence interval of $\bbeta^{*(j)}$ can be constructed as 
\begin{align}
\label{ci}
\bigg( \hat\bbeta_d^{(j)}-z_{\frac{\alpha}{2}}\sigma\frac{\|\hat\bSigma^{\frac{1}{2}}\hat\bmeta\|}{\sqrt{n}}, \hat\bbeta_d^{(j)}+z_{\frac{\alpha}{2}}\sigma\frac{\|\hat\bSigma^{\frac{1}{2}}\hat\bmeta\|}{\sqrt{n}} \bigg).
\end{align}
Usually the variance of the noise $\sigma$ is unknown. Thus the need for consistent estimation of $\sigma$ arises. In order to estimate $\sigma$ we assume there exists an estimator $\hat \bbeta$ which does well in terms of mean squared prediction error (see Theorem \ref{sigma_hat_rate} for the precise assumption on $\hat \bbeta$). We use only the first half of the data to estimate $\sigma$ with $\hat\sigma = \sqrt{\frac{1}{n}\sum_{i\in[n]} (Y_i - \bX_i\T \hat \bbeta)^2}$. Alternatively, for this step one could estimate $\hat \sigma$ using the entire data set, since we do not need sample splitting when we estimate $\sigma$ (we only need a consistent estimator). The following Theorem \ref{sigma_hat_rate} proves the consistency of such an estimator of $\sigma$. Theorem \ref{sigma_hat_rate} does not require the noise to be Gaussian, and even sub-Gaussian. It only assumes the existence of a $6$-th moment.
\begin{theorem}
\label{sigma_hat_rate}
Let $\hat\sigma = \sqrt{\frac{1}{n}\sum_{i \in [n]} (Y_i - \bX_i\T \hat \bbeta)^2}$. Suppose $\EE \varepsilon^6 < +\infty$, and that the eigenvalues of $\bSigma$ are bounded from above and below. Let $\hat\bbeta$ be an estimator of $\bbeta^*$ such that with probability converging to $1$ we have $\|\Xb(\hat \bbeta - \bbeta^*)\|\lesssim \sigma\delta$ for some $\delta=o(\sqrt{n})$. Then with probability converging to $1 - e^{-\delta^2/2}$,  we have
\begin{align*}
    |\hat\sigma^2-\sigma^2| \lesssim \frac{\bigg(\sqrt{\Var(\varepsilon_i^2)}\vee\sigma^2\bigg)\,\delta}{\sqrt{n}}.
\end{align*}
\end{theorem}
In the above since $\delta$ can be taken such that $\delta \rightarrow \infty$ as $n \rightarrow \infty$ (as long as $\delta = o(\sqrt{n})$), the result shows that $\hat \sigma$ is consistent. Note that the assumption $\|\Xb(\hat \bbeta - \bbeta^*)\|\lesssim\sigma\delta$ is achieved by many estimators. For example, \cite[Lemma A.1]{neykov2019gaussian} implies that convex constrained least squares estimators satisfy this condition assuming certain rate conditions on the growth of the Gaussian complexity of the tangent cone with respect to the sample size; \cite[Corollary 6.2]{bellec2018slope} and \cite[Corollary 6.2]{derumigny2018improved} imply that it holds for SLOPE and square-root SLOPE assuming that the vector is sufficiently sparse with respect to the sample size. The explicit order of $\delta$ for those cases can be found in Lemma \ref{sigma_hat_cvs_ls} and Lemma \ref{sigma_hat_slope} when we consider applying our general procedure to some special cases. In the case when $\hat \sigma$ is consistent, it follows by Slutsky's theorem that $\sigma$ in the confidence interval in \eqref{ci} can be substituted with $\hat\sigma$:

\begin{align}
\label{ci:sigma:est}
\bigg( \hat\bbeta_d^{(j)}-z_{\frac{\alpha}{2}}\hat \sigma\frac{\|\hat\bSigma^{\frac{1}{2}}\hat\bmeta\|}{\sqrt{n}}, \hat\bbeta_d^{(j)}+z_{\frac{\alpha}{2}}\hat\sigma\frac{\|\hat\bSigma^{\frac{1}{2}}\hat\bmeta\|}{\sqrt{n}} \bigg).
\end{align}

In the following Section \ref{cvs_ls:sec} and Section \ref{slope_sqrtslope:sec}, we discuss in details how to implement the debiasing procedure Algorithm \ref{algo_unknowncov} for some commonly used estimators including monotone regression, positive monotone regression, LASSO, SLOPE and square-root SLOPE. More concretely, the next section, Section \ref{cvs_ls:sec} is dedicated to convex constrained least squares, while Section \ref{slope_sqrtslope:sec} discusses an application to SLOPE and square-root SLOPE.

%%%%%%%%%
%% Sec CVX LS
%%%%%%%%%
\section{Convex Constrained Least Squares}%\label{section:convex:constrained:least:squares}
\label{cvs_ls:sec}
In this section we are interested in the estimator \eqref{constrained_ls} which we mentioned in the introduction section. Clearly this estimator is a form of constrained least squares, where the practitioner has knowledge that the true coefficient $\bbeta^*$  belongs to a convex set $K$. Assuming that least squares is a reasonable criteria to estimate $\bbeta^*$, the practitioner further imposes a restriction that $\hat \bbeta \in K$. Similarly to how LASSO biases the coefficients by shrinking them towards zero, imposing a constraint on $\hat \bbeta$ also biases the coefficients and standard inference methods do not work even in the low-dimensional setting. This motivates us to debias individual coordinates or contrasts of the estimator $\hat \bbeta$. In this section, we will assume that $\bX_i \sim N(0,\bSigma)$. The sole reason why we require this, is that there are known estimation and in-sample prediction guarantees for the performance of $\hat \bbeta$ given in \cite{neykov2019gaussian} which require the same condition. We do anticipate that at least some of those results may be generalized to broader distributional settings, as suggested by the works of \cite{genzel2020generic, li2015geometric}, but this is out of the scope of the present paper.

Since a set $K$ with the property $\bbeta^* \in K$ is given, it is natural to try and use that knowledge in our abstract debiasing procedure. In particular, we will use $K$ as the convex set required in step 1 and step 2 of Algorithm \ref{algo_unknowncov}. It remains to construct a vector $\vb \in K$ which obeys the requirements of step 1. We now provide such a construction. We claim that the solution of the following optimization program
\begin{align}
\label{step2_cvs_ls}
    \vb := \argmin_{\wb \in K} \|\hat \bbeta - \wb\| + \frac{\overline w(\cT_K(\wb)\cap\mathbb{S}^{p-1})}{\sqrt{n}},
\end{align} 
would satisfy the properties required of $\vb$. We now give a high level intuition why such $\vb$ is worth considering. Recall that the condition $\overline w(\cT_K(\vb)\cap \mathbb{S}^{p-1})\|\vb-\bbeta^*\|=o_p(1)$ in step 1 of Algorithm \ref{algo_unknowncov}. This condition will be met if both $\overline w(\cT_K(\vb)\cap \mathbb{S}^{p-1})$ and $\|\vb-\bbeta^*\|$ are ``small''. Suppose there exists a vector $\vb'$ such that $\|\vb' - \bbeta^*\|$ is small, and in addition $\vb'$ has a ``small'' tangent cone, in the sense that $\frac{\overline w(\cT_K(\vb')\cap\mathbb{S}^{p-1})}{\sqrt{n}}$ is small. By the definition of $\vb$ it follows that 
\begin{align*} 
    \|\hat \bbeta - \vb\| + \frac{\overline w(\cT_K(\vb)\cap\mathbb{S}^{p-1})}{\sqrt{n}} \leq  \|\hat \bbeta - \vb'\| + \frac{\overline w(\cT_K(\vb')\cap\mathbb{S}^{p-1})}{\sqrt{n}}.
\end{align*} 
Therefore both terms $ \|\hat \bbeta - \vb\| $ and $\frac{\overline w(\cT_K(\vb)\cap\mathbb{S}^{p-1})}{\sqrt{n}}$ are ``small''.  By the triangle inequality $\|\vb - \bbeta^*\| \leq \|\hat \bbeta - \vb\| +  \|\hat \bbeta - \bbeta^*\|$. Finally we know by a result of \citep[see Corollary 2.7]{neykov2019gaussian} that $\|\hat \bbeta - \bbeta^*\|$ is ``small''. This implies that $\|\vb-\bbeta^*\|$ is ``small''. Theorem \ref{main_rst_cvs_ls} makes the above intuition precise and proves why the solution of program \eqref{step2_cvs_ls} satisfies the condition needed in step \ref{step_1}. 
\begin{theorem}
\label{main_rst_cvs_ls}
Consider the same setting as Theorem \ref{debiase_formula_applicable_unkowncov}, and further assume that $\bX_i\sim N(0, \bSigma)$. Suppose there exists $\vb'\in K$ such that $\|\vb'-\bbeta^*\|^2=o(1/\sqrt{n})$, and the tangent cone of $K$ at $\vb'$ has a simple structure such that $\overline w^2(\cT_K(\vb')\cap\mathbb{S}^{p-1})=o(\sqrt{n})$ and $\overline w^2(\cT_K(\vb')\cap\mathbb{S}^{p-1}) \rightarrow \infty$. Then for $\hat\bbeta$ being the constrained least squares estimator obtained via \eqref{constrained_ls}, the solution $\vb$ of \eqref{step2_cvs_ls} satisfies the condition needed in step 1 of Algorithm \ref{algo_unknowncov} with probability converging to $1$ asymptotically.
\end{theorem}

\begin{remark}
Some comments are in order. The existence of a vector $\vb'$ which is close to $\bbeta^*$, with a sufficiently small tangent cone is natural. If $\vb' = \bbeta^*$, this condition requires that $\bbeta^*$ has a simple structure; otherwise when $\vb' \neq \bbeta^*$ it does not require that $\bbeta^*$ has a simple structure, as long as it is close enough to a vector $\vb'$ with a simple structure. This enables consistent estimation of $\bbeta^*$ in high-dimensional settings. As an example (for a case when $\vb' = \bbeta^*$) consider the set $K = \{\bbeta: \|\bbeta\|_1 \leq \|\bbeta^*\|_1\}$ which is the LASSO constraint. Requiring that $\bbeta^*$ has a cone with small Gaussian complexity is equivalent to imposing a sparsity assumption on $\bbeta^*$. 

In addition, notice that the vector $\vb'$ in Theorem \ref{main_rst_cvs_ls} is not necessarily the same as the vector $\vb$ found by \eqref{step2_cvs_ls}. However, it may be useful to think that the vector $\vb$ is attempting to estimate $\vb'$ (although this intuition too is not necessarily precise). The existence of $\vb'$ guarantees that we can find a ``useful'' $\vb$ by \eqref{step2_cvs_ls} in step \ref{step_1}. After we find the desired $\vb$, one can compute the auxiliary vector $\hat\bmeta$ in step \ref{step_2} based on $\vb$ and $K$, and then use $\hat\bmeta$ to construct the debiased estimator $\hat\bbeta_d$ and the confidence interval as \eqref{ci} or \eqref{ci:sigma:est}. 
\end{remark}

Of course, in practice, in order to construct the confidence interval \eqref{ci:sigma:est} we need to estimate $\sigma$. As discussed in Lemma \ref{sigma_hat_cvs_ls} below, consistent estimation of $\sigma$ is possible in the convex constrained least squares case. 
\begin{lemma}
\label{sigma_hat_cvs_ls}
Consider the same setting as Theorem \ref{main_rst_cvs_ls} where $\hat\bbeta$ is a convex constrained least squares estimator. Then Theorem \ref{sigma_hat_rate}, applies with
\begin{align*}
    \delta \asymp \frac{\sqrt{n}}{\sigma}\|\vb'-\bbeta^*\| + \overline w(\cT_K(\vb')\cap\mathbb{S}^{p-1}),
\end{align*}
where $\delta = o(\sqrt{n})$ as required.
\end{lemma}

Our debiasing algorithm does not require the population covariance matrix $\bSigma$ to be known as long as it has bounded spectrum. Can one do better if one is given knowledge of $\bSigma$? It is known \citep{javanmard2018debiasing} that with prior knowledge of $\bSigma$, the LASSO estimator $\hat \bbeta$ can be debiased with the following formula:
\begin{align}
    \hat\bbeta_d = \hat\bbeta + n^{-1}\bSigma^{-1} \tilde\Xb\T (\tilde\bY - \tilde\Xb\hat \bbeta).
    \label{debiase_formula_split:main:text}
\end{align}
What is more, \cite{javanmard2018debiasing} show that when the design is Gaussian the requirement for the debiasing procedure to work with known $\bSigma$ is much weaker compared to the requirement with unknown $\bSigma$. See also \cite{bellec2019biasing} for a sharpened version of this result. In fact \cite{javanmard2018debiasing} also show that the same debiased estimator works without sample splitting under more stringent assumptions, but this is out of the scope of the present paper. Lemma \ref{debiase_formula_applicable} will show that the debiasing formula in \eqref{debiase_formula_split:main:text} also works for any convex constrained least squares estimator under proper conditions. Afterwards we will compare the conditions needed to successfully debias a convex constrained least squares estimator $\hat\bbeta$ for the known and unknown $\bSigma$ cases. Similarly to the LASSO case, without the knowledge of $\bSigma^{-1}$, we impose more stringent assumptions on the structure of tangent cones of the parameter space $K$.
\begin{lemma}
\label{debiase_formula_applicable}
Consider a linear model as in \eqref{gaussian_model} with Gaussian errors $\varepsilon_i \sim N(0,\sigma^2)$. Further assume that $\bX_i\sim N(0, \bSigma)$. Let $\{a_n\}_{n = 1}^{\infty}$ be any slowly diverging sequence with $n$, and let $\vb'\in K$, be a vector such that $\|\vb'-\bbeta^*\|\,a_n = o(1)$, $\overline w(\cT_K(\vb')\cap\mathbb{S}^{p-1})a_n=o(\sqrt{n})$ and $\overline w(\cT_K(\vb')\cap\mathbb{S}^{p-1}) \rightarrow \infty$. Let $\hat\bbeta$ be a convex constrained least squares estimator obtained by \eqref{constrained_ls} on the first half of the data. The debiased $j$\textsuperscript{th} coefficient $\hat\bbeta_{d}^{(j)}$ obtained by \eqref{debiase_formula_split:main:text} is conditionally asymptotically normal with mean equal to $\bbeta^{*(j)}$. In particular, let $Z = \frac{1}{\sqrt{n}}\bSigma^{-1} \tilde\Xb\T\bvarepsilon$, and $\hat\bSigma = \frac{1}{n}\tilde\Xb\T\tilde\Xb$ be the empirical Gram matrix of the second half, we have
\begin{align}\label{asymp:normality:convex:least:squares}
    \sqrt{n}(\hat\bbeta_d^{(j)} - \bbeta^{*(j)}) = & Z^{(j)} + \Delta^{(j)},\quad
    Z^{(j)}|\tilde\Xb \sim N(0, \sigma^2\eb^{(j)\top}\bSigma^{-1}\hat\bSigma\bSigma^{-1}\eb^{(j)}),\nonumber\\
    ~~~~ & \Delta^{(j)} = \sqrt{n}(\eb^{(j)\top}\bSigma^{-1} \hat\bSigma - \eb^{(j)\top})(\bbeta^* - \hat\bbeta),
\end{align}
and $\Delta^{(j)} = o_p(1)$.
\end{lemma}

Suppose $\bbeta^* = \vb'$ is $s$-sparse and $K = \{\bbeta : \|\bbeta\|_1 \leq \|\bbeta^*\|_1\}$. The condition $\overline w(\cT_K(\vb')\cap\mathbb{S}^{p-1})a_n= \overline w(\cT_K(\vb')\cap\mathbb{S}^{p-1})a_n = o(\sqrt{n})$ in Lemma \ref{debiase_formula_applicable} is in fact a condition on the sparsity $s$. The Gaussian complexity of the tangent cone $\cT_K(\bbeta^*)$ can be evaluated in terms of the sparsity $s$ as \citep[Proposition 3.10]{chandrasekaran2012convex}
\begin{align}
    \overline w(\cT_K(\bbeta^*)\cap\mathbb{S}^{p-1}) = O\bigg(\sqrt{s\log\frac{ep}{s}}\bigg).
\label{lasso_gw}
\end{align}
Thus if $s$ doesn't scale with $n, p$ we have $s=o(n/ (a_n^2\log p) )$. If one selects $a_n=\sqrt{\log p}$, the condition in Lemma \ref{debiase_formula_applicable} becomes $s=o(n/ (\log p)^2)$ (assuming $p \rightarrow \infty$ as $n \rightarrow \infty$), which matches the condition needed in debiasing the regularized LASSO for the known covariance case \citep{javanmard2018debiasing}. Assuming $a_n = \sqrt{\log p}$ is convenient since in this case by tracking the proof of Lemma \ref{debiase_formula_applicable} and applying the union bound one may claim that \eqref{asymp:normality:convex:least:squares} holds for all $j \in [p]$, which is precisely the setting of \citep{javanmard2018debiasing}.  

The condition $\overline w^2(\cT_K(\vb')\cap\mathbb{S}^{p-1})=o(\sqrt{n})$ needed in Theorem \ref{main_rst_cvs_ls} is more stringent than the condition $\overline w(\cT_K(\vb')\cap\mathbb{S}^{p-1})a_n=o(\sqrt{n})$ in Lemma \ref{debiase_formula_applicable}, which can be viewed as the price we pay for having an unknown covariance. On an important note, presently we do not have corresponding lower bounds showing that these conditions are also necessary.
We may say however that in the case when $K = \{\bbeta : \|\bbeta\|_1 \leq \|\bbeta^*\|_1\}$ the condition $\overline w^2(\cT_K(\vb')\cap\mathbb{S}^{p-1})=o(\sqrt{n})$ reduces to a condition on the sparsity parameter $s$ by \eqref{lasso_gw}. The equivalent condition in terms of $s$ is $s \log ep/s = o(\sqrt{n})$ which matches the assumption needed in debiasing the regularized LASSO for the unknown covariance case \citep{cai2017confidence, javanmard2018debiasing}. Finally we would like to add that we do show lower bounds on the confidence interval lengths in Appendix \ref{lower:bounds:appendix} of the supplement, and further one additional example on non-negative least squares can be found in Appendix \ref{non:negative:least:squares} of the supplement.

%%%%%%%%%%%%%%%%%%%%%%%%%
%%%%%%%%%%%%%%
\subsection{Monotone Cone Regression}
\label{case_mnt}
Consider the case where the true coefficient $\bbeta^*$ is in a monotone cone parameter space $M^p$ in $\RR^p$ defined as
%comprised of $l$ constant pieces
$$M^p = \{(\beta_1,\ldots,\beta_p)\T\in\RR^p: \beta_1 \leq \beta_2 \leq \ldots \leq \beta_p\}.$$
Notice that $M^p$ is convex. Moreover, the set of monotone vectors with $l$ constant pieces is defined as \citep{gao2017minimax}
\begin{align*}
    M_l^p = \big\{ & (\beta_1,\ldots,\beta_p)\T\in\RR^p: \text{there exist }\{a_j\}^l_{j=0}\text{ and }\{u_j\}^l_{j=0}\text{ such that} \\
    & 0 = a_0 \leq a_1 \leq \ldots \leq a_l = p,\\
    & u_1 \leq u_2 \leq \ldots \leq u_l,\text{ and }\beta_i = u_j\text{ for all }i\in(a_{j-1}, a_j]
    \big\}.
\end{align*}
Given the prior knowledge $\bbeta^*\in M^p$, the constrained least squares estimator $\hat\bbeta$ in \eqref{constrained_ls} can be solved by incorporating isotonic regression in projected gradient descent. 

To find the desired vector $\vb$ in step \ref{step_1}, we solve \eqref{step2_cvs_ls} with $\overline w(\cT_{M^p}(\vb)\cap\mathbb{S}^{p-1}) = \sqrt{l\log(ep/l)}$. The latter is a legitimate upper bound on the Gaussian complexity of the tangent cone, as the result in \cite[(1.19), (1.22), Proposition 3.1]{bellec2018sharp} shows that for a monotone cone $M^p\in\RR^p$, the complexity of the tangent cone at any vector $\vb$ comprised of $l$ constant pieces has an explicit upper bound $w(\cT_{M^p}(\vb)\cap\mathbb{S}^{p-1}) \leq \sqrt{l\log(ep/l)}$. Thus the optimization problem \eqref{step2_cvs_ls} can be simplified to
\begin{align}
\label{mnt_step2}
    \argmin_{\vb\in M_l^p}\|\hat \bbeta - \vb\| + \sqrt{\frac{l}{n}\log\frac{ep}{l}}.
\end{align}
For a fixed $l$, the term $\sqrt{(l/n)\log(ep/l)}$ is constant for all $\vb\in M_l^p$. Thus in each $M_l^p$, the solution of $\argmin_{\vb\in M^p_l}\|\hat \bbeta - \vb\| + \sqrt{(l/n)\log(ep/l)}$ should minimize $\|\hat\bbeta - \vb\|$, which is exactly the projection of $\hat\bbeta$ to $M_l^p$, denoted as $\Pi_{M_l^p}(\hat\bbeta)$. Let $p'$ be the number of constant pieces in $\hat\bbeta$, where $p'\leq p$. 
The optimization problem \eqref{mnt_step2} can be converted to an optimization problem over finitely many candidates. Define
$$\hat l = \argmin_{l\in[1,p']}\|\hat \bbeta - \Pi_{M^p_{l}}(\hat\bbeta)\| + \sqrt{\frac{l}{n}\log\frac{ep}{l}}.$$
Since there is no point in looking for values of $l > p'$ as this will only increase the loss function (compared to when $l = p'$), the desired $\vb$ in \eqref{mnt_step2} is exactly $\Pi_{M^p_{\hat l}}(\hat\bbeta)$. There is an efficient projection algorithm of $\hat\bbeta$ to $M^p_l$ as proposed by \cite[Algorithm 1]{gao2017minimax} which takes $O(p'^3)$ time to compute all projections for $l\in[1,p']$.

Once $\vb$ is obtained, we solve the optimization program \eqref{opt_step3} using subgradient descent as in Algorithm \ref{algo_step3}. The final piece of the puzzle is to show how to calculate the projections $\Pi_{\cT_{M^p}(\vb)}(\cdot)$ and $\Pi_{-\cT_{M^p}(\vb)}(\cdot)$. We compute them by decomposing $\cT_{M^p}(\vb)$. Since $\vb$ is $\hat l$ piece-wise monotone, the tangent cone of $M^p$ at $\vb$ can be decomposed as \citep[Proposition 3.1]{bellec2018sharp}
\begin{align*}
    \cT_{M^p}(\vb) = M^{p_1} \times M^{p_2} \times \ldots \times M^{p_{\hat l}},
\end{align*}
where each $p_i$ is the length of each constant piece of $\vb$, and $p_1+\ldots+p_{\hat l}=p$. Thus for any vector $\ub = (u_1, u_2, \ldots, u_p)\T \in\RR^p$, the projection of $\ub$ to $\cT_{M^p}(\vb)$ is \citep[Equation B.2]{amelunxen2014living}
\begin{align}
    \Pi_{\cT_{M^p}(\vb)}(\ub) = \bigg(&\Pi_{M^{p_1}}\big( (u_1, \ldots, u_{p_1}) \big)\T, \Pi_{M^{p_2}}\big( (u_{p_1+1}, \ldots, u_{p_1+p_2}) \big)\T ,\nonumber\\ &\ldots , \Pi_{M^{p_{\hat l}}}\big( (u_{p_1+\ldots + p_{\hat{l}-1}+1}, \ldots, u_p) \big)\T\bigg)\T,
\label{proj_mnt_tgcone}
\end{align}
noting that projections into a monotone cone, as in \eqref{proj_mnt_tgcone} can be efficiently implemented via the PAVA algorithm for isotonic regression \cite[see e.g.]{robertson1988order}. Once we have computed $\hat\bmeta$, we can debias $\hat\bbeta$ using the formula in step \ref{step_2}. The entire procedure to get a debiased estimation $\hat\beta_d^{(j)}$ for monotone cone regression is summarized in Algorithm \ref{algo_mnt}.

\begin{remark}\label{remark:monotone:cone}
We remark that thanks to Theorem \ref{main_rst_cvs_ls}, $\bbeta^*$ need not be piecewise constant. In fact, by Lemma 2 of \cite{bellec2015sharp} we know that any vector $\bbeta^* \in M^p$, can be approximated within $\|\bbeta^* - \vb'\| \leq \frac{\bbeta^{*(p)} - \bbeta^{*(1)}}{2k}$ by a vector $\vb'$ consisting of at most $k$ constant pieces. 

So long as $\bbeta^{*(p)} - \bbeta^{*(1)}$ is bounded, it suffices that $p$ is such that we can select $k \gg n^{1/4}$ with $k \log p/k = o(\sqrt{n})$, and the regression with signal $\bbeta^*$ can be debiased.
\end{remark}
\begin{algorithm}[ht]
\caption{Debias the $j$\textsuperscript{th} Coefficient for Monotone Cone Regression}
\label{algo_mnt}
\begin{algorithmic}
\STATE\textbf{Input:} Two equal size partitions $(\overline\Xb, \overline\bY)$ and $(\tilde\Xb, \tilde\bY)$, $\hat\bbeta$ obtained by projected gradient descent with isotonic regression.\\
\STATE\textbf{Initialize:} Empirical Gram matrix of the second partition $\hat\bSigma=\frac{1}{n}\tilde\Xb\T\tilde\Xb$.
\begin{enumerate}
    \item Solve $\hat l \leftarrow \argmin_{l\in[1,p']}\|\hat \bbeta - \Pi_{M^p_{l}}(\hat\bbeta)\| + \sqrt{\frac{l}{n}\log\frac{ep}{l}}$.\\
    $\vb \leftarrow \Pi_{M^p_{\hat l}}(\hat\bbeta)$.
    \item Run Algorithm \ref{algo_step3}. Compute $\Pi_{\cT_{M^p}(\vb)}(\cdot)$ by isotonic regression (PAVA) with \eqref{proj_mnt_tgcone}. For $\Pi_{-\cT_{M^p}(\vb)}(\cdot)$ use \eqref{neg:cone:calc}. %the negation of an increasing monotone cone is a decreasing monotone cone.\\
    The debiased $j$\textsuperscript{th} coefficient equals $\hat\bbeta_{d}^{(j)} \leftarrow \vb^{(j)} + n^{-1}\hat\bmeta\T \tilde\Xb\T (\tilde\bY - \tilde\Xb\vb)$.
\end{enumerate}
\end{algorithmic}
\end{algorithm}

%%%%%%%%%%
\subsection{Positive Monotone Cone Regression}
\label{case_pos_mnt}
Based on the analysis in Section \ref{case_mnt} for the monotone cone $M^p$, we can analogously develop the debiasing technique when the true coefficient is inside of a positive monotone cone defined as
$$M^{p+} =  \{(\beta_1,\ldots,\beta_p)\T\in\RR^p: 0 \leq \beta_1 \leq \beta_2 \leq \ldots \leq \beta_p\}.$$
The algorithm to debias the $j$\textsuperscript{th} coefficient in positive monotone cone regression is the same as Algorithm \ref{algo_mnt} except for some minor modifications.
Specifically, $\hat\bbeta$ can also be obtained by projected gradient descent, but such a projection onto a positive monotone cone is done by fitting an isotonic regression followed by assigning zeros to all the negative coordinates \citep{nemeth2012project}. The procedure of finding $\vb$ in step \ref{step_1} is the same as the monotone cone case. This is so since $\hat\bbeta$ is always positive and the algorithm in \cite{gao2017minimax} computes the projections of $\hat\bbeta$ onto $M_l^p$ by further averaging itself, all the projections automatically belong to the positive monotone cone.
For step \ref{step_2}, we need to project a vector $\ub = (u_1,\ldots,u_p)\T\in\RR^p$ onto $\cT_{M^{p+}}(\vb)$---the tangent cone of the positive monotone cone $M^{p+}$ at $\vb$. By Proposition \ref{decomp_tan_pos}, $\cT_{M^{p+}}(\vb)$ can be decomposed into Cartesian products of a positive monotone cone and several other monotone cones. Thus the projection onto $\cT_{M^{p+}}(\vb)$ can be computed as a Cartesian product of the projection onto every component.

\begin{proposition}
\label{decomp_tan_pos}
Suppose $\vb\in M^{p+}$ has $l$ constant pieces, and the length of each constant piece is $p_i$ for $i\in[l]$. If the first constant piece consists of zeros, the tangent cone of $M^{p+}$ at $\vb$ can be decomposed as
$$\cT_{M^{p+}}(\vb) = M^{p_1+} \times M^{p_2} \times \ldots \times M^{p_{l}},$$
otherwise it is
$$\cT_{M^{p+}}(\vb) = M^{p_1} \times M^{p_2} \times \ldots \times M^{p_{l}}.$$
\end{proposition}

\begin{remark}
Similarly to the monotone cone case, the $\bbeta^*$ vector need not be piecewise constant. See Remark \ref{remark:monotone:cone}.
\end{remark}

%%%%%%%%%
\subsection{Constrained Group LASSO}
\label{lasso:sec}
The next example is a constrained group LASSO problem, which includes constrained LASSO as a special case. Importantly, this section will serve as a building block to our algorithm which debiases SLOPE and square-root SLOPE (see Section \ref{slope_sqrtslope:sec} below). Let us assume that the support of the vector $\bbeta^*$ is split into disjoint fixed groups $G \in \cG$, where the cardinality $|\cG| = M$, where each group $G \in \cG$ satisfies $|G| \leq B$ for some constant $B$. Define the norm $\|\bbeta\|_{2,1} = \sum_{G \in \cG} \|\bbeta_G\|$, where $\bbeta_G$ is the restriction of $\bbeta$ on the group $G$. Suppose that the convex set is $K = \{\bbeta: \|\bbeta\|_{2,1} \leq \|\bbeta^*\|_{2,1}\}$. Here we assume a prior knowledge of $\|\bbeta^*\|_{2,1}$. This is a common assumption in works analyzing the constrained version of LASSO e.g. \citep[see e.g.]{thrampoulidis2014simple, wainwright2019high}. In Section \ref{debias:group:lasso:regularied} we debias a regularized version of the group LASSO which does not need this requirement. Clearly, when the groups $G$ are singletons, the constrained group LASSO reduces to the constrained LASSO, so this is a more general example.

We now follow the outline of Algorithm \ref{algo_unknowncov} to debias the constrained group LASSO estimator. In step \ref{step_1}, in order to solve \eqref{step2_cvs_ls}, we use 
\begin{align*}
    \overline w(\cT_K(\vb)\cap\mathbb{S}^{p-1}) = \sqrt{s B + s (\sqrt{2\log(M - s)} + \sqrt{B})^2},
\end{align*}
\citep[see][Theorem 3.1]{rao2011tight}, where $s$ is the number of non-zero groups in $\vb$. 
Let $\vb_s$ be the projection of $\hat\bbeta$ onto the set of $s$-group sparse vectors with $2,1$-norm equal to $\|\bbeta^*\|_{2.1}$. The optimization \eqref{step2_cvs_ls} reduces to an optimization with finite candidates
\begin{align*}
    \hat s \leftarrow \argmin_{s\in[1, \|\hat\beta\|_0]}\|\hat \bbeta - \vb_s\| + \sqrt{\frac{s B + s (\sqrt{2\log(M - s)} + \sqrt{B})^2}{n}},
\end{align*}
and we find the output of step \ref{step_1} by choosing $\vb=\vb_{\hat s}$. 
According to Lemma \ref{beta_proj_onto_s}, the computation of the projection $\vb_s$ has a complexity $O(s)$ (after the groups of $\hat \bbeta$ have been ordered by magnitude), by greedily taking the largest $s$ groups of $\hat \bbeta$ and distributing the remaining of the 2,1-norm equally across the $s$-groups.  
\begin{lemma}
\label{beta_proj_onto_s}
Let $S$ be the set of indices of the $s$ largest in magnitude (in $\ell_2$ norm) groups of $\hat \bbeta$, and $\Lambda\geq0$ be a constant. Let $\vb_s$ be the projection of $\hat\bbeta$ onto the set T = $\{\bbeta:\, \|\bbeta\|_{2,1} = \Lambda \text{ and } \|\bbeta\|_0 = s\}$. Then $\vb_s$ satisfies
\begin{align*}
    \vb_{s,G} = \begin{cases}
    0, & \text{ if } G\notin S\\
    \hat \bbeta_{G}+\frac{\bbeta_{G}}{\|\bbeta_{G}\|}\,\frac{\Lambda - \sum_{G'\in S}\|\hat \bbeta_{G'}\|}{s}, & \text{ if } G\in S.
    \end{cases}
\end{align*}
\end{lemma}

In the above, ties in ordering the of the groups of $\hat \bbeta$ in magnitude can be broken arbitrarily. Once we obtain the vector $\vb$ in step \ref{step_1}, the projection onto the tangent cone $\cT_K(\vb)$ needed in step \ref{step_2} can be done efficiently by first finding the projection onto its polar cone --- the normal cone at $\vb$ with respect to the set $K$
\begin{align*}
    \cN_K(\vb) = \{ \gb:\, \langle \gb, \vb'-\vb\rangle\leq 0\,\,,\,\vb'\in K\}.
\end{align*}
Then the projection $\Pi_{\cT_K(\vb)}(\zb)=\zb-\Pi_{\cN_K(\vb)}(\zb)$ by applying Moreau's decomposition \citep{moreau1962decomposition}. Let $S$ be the set of non-zero coordinates of $\vb$. For the set $K$ equal to the $2,1$ ball with radius $\|\bbeta^*\|_{2,1}$, the normal cone has an explicit form. Using the result in the beginning of Section 3 of \cite{rao2011tight} we know that the normal cone for the group lasso is given by
\begin{align*}
    \cN_{K}(\vb) = \bigg\{\zb \in \RR^p: \zb_G = t \frac{{\vb_{G}}}{\|\vb_G\|} \forall \mbox{ active }G , \|\zb_G\| \leq t, \forall \mbox{ inactive } G \bigg\}.
\end{align*}
Based on the expression above, the projection of a vector $\zb$ onto the normal cone $\cN_K(\vb)$ can be converted to a one-dimensional convex optimization program with an auxiliary parameter $t$ which can be solved by golden section search, e.g. \citep{kiefer1953sequential}
\begin{align}\label{proj_tan1}
    \hat t = \argmin_{t \in [0, \max_{G \in \cG} \|\zb_G\|]}\sum_{G \mbox{ \tiny active}} \bigg\|\zb_G - t \frac{\vb_G}{\|\vb_G\|}\bigg\|^2 + \sum_{G \mbox{ \tiny inactive}} (\|\zb_G\| - t)_+^2,
\end{align}
where $(x)_+ = \max\{x, 0\}$. The search interval of $t$ has an upper bound $\max_{G \in \cG} \|\zb_G\|$ since the objective function will have a larger value for all $t>\max_{G \in \cG} \|\zb_G\|$ compared with $t=\max_{G \in \cG} \|\zb_G\|$. Once $\hat t$ is obtained the projection onto $\cN_K(\vb)$ is
\begin{align}
    \Pi_{\mathcal{N}_{K}(\vb)}(\zb) & = \begin{cases}
    \hat t\frac{\vb_G}{\|\vb_G\|}, & \text{ if } G\in S\\
    \frac{\zb_G}{\|\zb_G\|}(\hat{t}\wedge \|\zb_G\|), & \text{ if } G\notin S,
    \end{cases}
    \label{proj_tan2}
\end{align}
where $S$ denotes the set of all active groups. We remark that golden section search can get arbitrarily close to the optimal value, which is good from a computational standpoint. If one would like to obtain the exact solution (which is desirable for theoretical purposes), one can order all $\|\zb_G\|$ values and look for $t$ in between them. Each problem is a constrained quadratic polynomial so it is easy to optimize. This approach will solve \eqref{proj_tan1} precisely. 

A summary of the debiasing procedure specific for the constrained group LASSO estimator is given in Algorithm \ref{algo_lasso}. 

\begin{algorithm}[ht]
\caption{Debias the $j$\textsuperscript{th} Coefficient for constrained group LASSO}
\label{algo_lasso}
\begin{algorithmic}
\STATE\textbf{Input:} Two equal size partitions $(\overline\Xb, \overline\bY)$ and $(\tilde\Xb, \tilde\bY)$; $\hat\bbeta$ obtained by solving a constrained group LASSO problem. $K = \{\bbeta: \|\bbeta\|_{2,1} \leq \|\bbeta^*\|_{2,1}\}$.\\
\STATE\textbf{Initialize:} Empirical Gram matrix of the second partition $\hat\bSigma=\frac{1}{n}\tilde\Xb\T\tilde\Xb$.
\begin{enumerate}
    \item Solve $\hat s \leftarrow \argmin_{s\in[1,\|\hat\bbeta\|_0]}\|\hat \bbeta - \vb_s\| + C\sqrt{\frac{s B + s (\sqrt{2\log(M - s)} + \sqrt{B})^2}{n}}$. For each $s$, the projection $\vb_s$ is computed according to Lemma \ref{beta_proj_onto_s}.\\
    $\vb \leftarrow \vb_{\hat s}$;
    
    \item Run Algorithm \ref{algo_step3}. Compute $\Pi_{\cN_{K}(\vb)}(\cdot)$ by \eqref{proj_tan1}, \eqref{proj_tan2}, and apply Moreau's decomposition to get $\Pi_{\cT_K(\vb)}(\cdot)$. For $\Pi_{-\cT_K(\vb)}(\cdot)$ use \eqref{neg:cone:calc}. %note that $-\cT_K(\vb)=\cT_K(-\vb)$ since $-K = K$.\\
    The debiased $j$\textsuperscript{th} coefficient $\hat\bbeta_{d}^{(j)} \leftarrow \vb^{(j)} + n^{-1}\hat\bmeta\T \tilde\Xb\T (\tilde\bY - \tilde\Xb\vb)$.
\end{enumerate}
\end{algorithmic}
\end{algorithm}

%%%%%%%%%
%% Sec 3
%%%%%%%%%

\section{Minkowski Gauge Selectors}\label{minkowski:section}

In this section we first revisit a result of \cite{cai2016geometric} on Minkowski gauge minimization in a linear model. We will then show that our algorithm can be used to answer an open question left in the discussion section of \cite{cai2016geometric}, namely how can one handle inference in the unknown covariance case. We will first begin with reviewing general notation and definitions. 
% present a new result regarding regularized least squares with a Minkowski gauge of a convex set. We then, show how one can construct the set $K$ and the vector $\vb$ in this setting. This section is most closely related to the work of \cite{cai2016geometric}, although our approach is very different. In addition, one important difference is that we tackle the unknown covariance case, whereas \cite{cai2016geometric} assume the covariance is identity (with an extension given for a known covariance). This solves an open problem which \cite{cai2016geometric} left in their discussion section, namely how to handle inference in the unknown covariance case. 

Let $K' \subset \RR^p$  be a given fixed symmetric convex body such that $\mathbf{0} \in K'$ is an interior point. Without loss of generality we will suppose that $K'$ is a subset of the  $\ell_\infty$ unit ball. Consider solving the following program
\begin{align*}
    \min \rho_{K'}(\bbeta), \mbox{ s.t. } \rho_{K'}^*\bigg(n^{-1}\sum_{i \in [n]} [Y_i \bX_i - \bX_i \bX_i\T \bbeta]\bigg) \leq \lambda,
\end{align*}
where $\rho_{K'}(\bbeta)$ is the Minkowski gauge of $K'$, i.e., 
\begin{align*}
    \rho_{K'}(\bbeta) = \inf\{r \in \RR : r > 0, \bbeta \in r K'\},
\end{align*}
and $\rho_{K'}^*$ is defined as $\rho_{K'}^*(\vb) = \sup_{\ab: \rho_{K'}(\ab) \leq 1} \vb \T \ab$, and $\lambda > 0$ is a tuning parameter.
%{\color{red} Here we may need an additional constraint see Cai (2.13). -- It's not necessary it seems, at least for our examples...}

Here the true model is \eqref{gaussian_model}, where we remind the reader that we are assuming $\varepsilon_i$ are independent of $\bX_i$. Let us define the tangent cone $\overline \cT_{K'}(\bbeta) := \{t(\hb - \bbeta): t\geq 0, \rho_{K'}(\hb) \leq  \rho_{K'}(\bbeta)\}$. In terms of the previous notation we have the identity $\overline \cT_{K'}(\bbeta) = \cT_{K}(\bbeta)$ for $K = \rho_{K'}(\bbeta)K'$. We need additionally to define the so called ``asphericity ratio'' (see (2.11) of \cite{cai2016geometric}):
\begin{align*}
    \gamma_{K'}(\bbeta^*) = \sup\bigg\{\frac{\rho_{K'}(\hb)}{\|\hb\|} : \hb \in \overline \cT_{K'}(\bbeta^*) \bigg\}.
\end{align*}
Following \cite{cai2016geometric} we will assume further that $\bX_i \sim N(0, \bSigma)$, where we will additionally assume that $\bSigma$ has bounded spectrum (both from above and below). 
% Define the quantity
% \begin{align*}
%     C(\bbeta^*) := \sup_{\bbeta \neq \bbeta^*} \frac{\rho_{K'}(\bbeta^*) - \rho_{K'}( \bbeta)}{\|\bbeta^* - \bbeta\|}.
% \end{align*}

We now restate part of Corollary 1 of \cite{cai2016geometric}.
\begin{proposition}[Corollary 1 \cite{cai2016geometric}]\label{thm:minkowski:gauge} Set $\lambda \asymp \frac{w(\Xb K')}{n}$ and suppose that $n \gtrsim \|\bSigma^{\frac{1}{2}}\|_{op}\times \allowbreak \|\bSigma^{-\frac{1}{2}}\|_{op}(w(\cT_{K'}(\bbeta^*) \cap \mathbb{S}^{p-1}) + 1)$. In addition let $\varepsilon_i \sim N(0,\sigma^2)$. Then, with high probability (converging to $1$ asymptotically)
% . Suppose $w(\cT_{K'}(\bbeta^*) \cap \mathbb{S}^{p-1}) = o(\sqrt{n})$ and $w(\cT_{K'}(\bbeta^*) \cap \mathbb{S}^{p-1}) \rightarrow \infty$. Then with probability converging to $1$ we have
\begin{align*}
    \|\hat \bbeta - \bbeta^*\|\lesssim \sigma \frac{\gamma_{K'}(\bbeta^*)w(\Xb K')}{n},
\end{align*}
where, as usual, $\lesssim$ and $\gtrsim$ omit absolute constants.
\end{proposition}

\begin{remark}The original statement of \cite{cai2016geometric} imposes the condition $n \gtrsim w(\bSigma^{\frac{1}{2}}\cT_{K'}(\bbeta^*) \cap \mathbb{B}_2^{p})$ instead of $n \gtrsim \|\bSigma^{\frac{1}{2}}\|_{op} \|\bSigma^{-\frac{1}{2}}\|_{op}(w(\cT_{K'}(\bbeta^*) \cap \mathbb{S}^{p-1}) + 1)$. However,  by Remark 1.7 of \cite{plan2016generalized} we have $  w(\bSigma^{\frac{1}{2}}\cT_{K'}(\bbeta^*) \cap \mathbb{B}_2^{p})\leq \|\bSigma^{\frac{1}{2}}\|_{op} \|\bSigma^{-\frac{1}{2}}\|_{op}  w(\cT_{K'}(\bbeta^*) \cap \mathbb{B}_2^{p})$ while by the proof of Proposition 10.2 of \cite{amelunxen2014living} we know ${w}(\cT_K(\bbeta^*) \cap \mathbb{S}^{p-1})^2 \leq {w}(\cT_K(\bbeta^*) \cap \mathbb{B}_2^p)^2 \leq {w}(\cT_K(\bbeta^*) \cap \mathbb{S}^{p-1})^2 + 1$ which shows that the condition we required implies that of \cite{cai2016geometric}. 
\end{remark}

We now suggest a way to find the set $K$ and a vector $\vb$ for the program above. We will assume a known upper bound on the asphericity ratio of $\bbeta^*$. Concretely, suppose $\bar \gamma_{K'}(\bbeta)$ is a known and ideally easily calculable upper bound on $\gamma_{K'}(\bbeta)$. Such bound exists often times as the examples in \cite{cai2016geometric} show. We assume we have a known upper bound on $\bar \gamma_{K'}(\bbeta^*) \leq \bar s$ such that $\bar s^2 \overline w(\Xb K') \overline w(K')  = o(n)$. We propose to solve the following optimization
\begin{align}\label{find:v:minkowski}
    \max \rho_{K'}(\vb), \mbox{ s.t. }\nonumber\\
    \|\vb - \hat \bbeta\| \lesssim \bigg(\frac{\bar s \overline w(\Xb K') }{n}\bigg) \bigg(\frac{n}{\bar s^2 \overline w(\Xb K') \overline w(K')}\bigg)^{\gamma}\\
    \bar \gamma_{K'}(\vb) \leq \bar s.\nonumber
\end{align}
for some small $1 > \gamma > 0$. We have the following result
\begin{lemma}\label{find:v:minkowski:lemma}
Suppose $\bar s^2 \overline w(\Xb K') \overline w(K')  = o(n)$. Let the set $K = \rho_{K'}(\vb)K'$. Then under the assumptions of Proposition \ref{thm:minkowski:gauge}, $\bbeta^* \in K$, and in fact $\bbeta^*$ is a feasible point of program \eqref{find:v:minkowski} for large enough $n$. Furthermore, the solution to program \eqref{find:v:minkowski} satisfies
\begin{align*}
    \|\vb - \bbeta^*\| \overline w(\overline \cT_{K'}(\vb) \cap \mathbb{S}^{p-1}) = o_p(1),
\end{align*}
where $ \overline w(\overline \cT_{K'}(\vb) \cap \mathbb{S}^{p-1}) = \bar \gamma_{K'}(\vb) w(K')$ is an upper bound of $w(\overline \cT_{K'}(\vb) \cap \mathbb{S}^{p-1})$.
\end{lemma}

    % Hence in situations where $\frac{w(\Xb K')}{\sqrt{n}} \gtrsim w(K')$, we would expect that 
    % \begin{align*}
    %     \frac{\gamma_{K'}(\bbeta^*)w(\Xb K')}{n} \lesssim \frac{w(\overline \cT_{K'}(\bbeta^*) \cap \mathbb{S}^{p-1})}{\sqrt{n}}
    % \end{align*}
    
We will illustrate the abstract procedures above with an example on trace regression. 

\subsection{Trace Regression}

In this subsection we consider trace regression, where $\bbeta^* \in \RR^{p_1 \times p_2}$ is a low-rank matrix. Here we observe
\begin{align*}
    Y_i = \bX_i\T \operatorname{vec}(\bbeta^*) + \varepsilon_i = \tr(\Xb_i\T \bbeta^*) + \varepsilon_i,
\end{align*}
where $\Xb_i \in \RR^{p_1 \times p_2}$ is a matrix form of $\bX_i \in \RR^{p_1p_2}$. Here the set $K' = \{\bbeta \in \RR^{p_1 \times p_2}: \|\bbeta\|_* \leq 1\}$, where $\|\cdot\|_*$ denotes the nuclear norm (i.e. the sum of its singular values). By a calculation in Section 3.5.2 of \cite{cai2016geometric}, we know that $\gamma_{K'}(\bbeta^*) \leq 2 \sqrt{2r}$, where $r$ is the rank of $\bbeta^*$. In addition, by a calculation in the same section, \cite{cai2016geometric} have argued that $w(\Xb K')/\sqrt{n} \leq \sqrt{p_1 + p_2}$. The extreme points of the set $K'$ are rank one matrices $\ub \vb\T $ where $\|\ub\| = \|\vb\| = 1$. It follows that for a standard Gaussian matrix $\Gb \in \RR^{p_1 \times p_2}$ the maximizer of $\tr(\Gb \ub \vb\T) = \sigma_{\max}(\Gb)$ is the largest singular value, which is known to be of the order of $\sqrt{p_1} + \sqrt{p_2} \leq \sqrt{2(p_1 + p_2)}$ in expectation \cite[see Theorem 5.32]{Vershynin2012Introduction}. Hence $w(K') \leq \sqrt{2(p_1 + p_2)}$. Thus, to clarify, according to the condition $\bar s^2 \overline w(\Xb K') \overline w(K')  = o(n)$ in Lemma \ref{find:v:minkowski:lemma} we require that $(p_1 + p_2)r = o(\sqrt{n})$ in order for us to be able to debias $\bbeta^*$.

We will now focus on how to solve the optimization program given in the end of the last section in this example. Suppose we know an upper bound $\bar r$ of the rank of $\bbeta^*$. Then the final constraint can be substituted with $\rank(\vb) \leq \bar r$. Now suppose, that the SVD of $\hat \bbeta$ is given by $\hat \bbeta = \hat \Ub \hat \bLambda \hat \Vb\T$, where $\hat \bLambda \in \RR^{\hat r \times \hat r}$, $\hat \Ub \in \RR^{p_1 \times \hat r}$ and $\hat \Vb \in \RR^{p_2 \times \hat r}$. We have the following 
\begin{lemma}\label{lemma:find:v:trace:reg}
By expanding the orthogonal bases from the SVD of $\hat \bbeta$, write $\hat \bbeta = \tilde \Ub \tilde \bLambda \tilde \Vb\T$, where $\tilde \bLambda \in \RR^{\min(p_1, p_2) \times \min(p_1, p_2)}$, $\tilde \Ub \in \RR^{p_1 \times \min(p_1, p_2)}$ and $\tilde \Vb \in \RR^{p_2 \times \min(p_1, p_2)}$, and $\tilde \bLambda$ has the same positive entries as $\hat \bLambda$ and additional zeros, while $\tilde \Ub\T \tilde \Ub = \tilde \Vb\T \tilde \Vb= \mathbb{I}$. Assume that the diagonal entries of $\tilde \bLambda$ are ordered in a decreasing manner. The solution to program \eqref{find:v:minkowski} is given by:
\begin{align*}
    \vb = \sum_{i = 1}^{\bar r} \bigg(\tilde \lambda_i + \sqrt{\frac{U - \sum_{i > \bar r} \tilde \lambda_i^2}{\bar r}}  \bigg)\tilde\ub_i \tilde \vb_i\T,
\end{align*}
where $U = \bigg(\frac{2\sqrt{2\bar r} \sqrt{p_1 + p_2}}{\sqrt{n}}\bigg) \bigg(\frac{\sqrt{n}}{(2\sqrt{2\bar r})^2 \sqrt{2} (p_1 + p_2)}\bigg)^{\gamma}$.
\end{lemma}
\begin{remark}
The reader may notice that Lemma \ref{lemma:find:v:trace:reg} is implicitly assuming that $U > \sum_{i > \bar r} \tilde\lambda_i^2$. From Lemma \ref{find:v:minkowski:lemma} we know that for sufficiently large $n$, $\bbeta^*$ is a feasible point of program \eqref{find:v:minkowski}, and $\operatorname{rank}(\bbeta^*) \leq \overline r$. By the proof of Lemma \ref{lemma:find:v:trace:reg} it is simple to see that when $\bbeta^*$ is feasible, it follows that $U \geq \sum_{i > r} \tilde \lambda_i^2 \geq \sum_{i > \overline r} \tilde \lambda_i^2$, where we denoted the rank of $\bbeta^*$ with $r$.
\end{remark}

Next on the agenda is to find the projection onto $\overline \cT_{K'}(\vb)$. According to \cite{chandrasekaran2012convex}, the normal cone to the tangent cone $\overline \cT_{K'}(\vb)$ is given by the following formula
\begin{align*}
    \overline \cN_{K'}(\vb) = \{t \Ub \Vb\T + \Wb: \Wb\T \Ub = 0, \Wb \Vb = 0, \sigma_{\max}(\Wb) \leq t, t \geq 0\},
\end{align*}
where $\vb = \Ub \bLambda \Vb\T$ is the SVD of $\vb$, and $\sigma_{\max}(\Wb)$ denotes the largest singular value of $\Wb$. Let $\Ab \in \RR^{p_1\times p_2}$  be a matrix that we wish to project on $\overline \cT_{K'}(\vb)$. By Moreau's decomposition we know that if we project $\Ab$ onto $\overline \cN_{K'}(\vb)$ and the resulting matrix is $\Bb$ then the projection onto $\overline \cT_{K'}(\vb)$ is given by $\Ab - \Bb$. To this end we need to solve the following problem
\begin{align*}
    \min \|\Ab - \Bb\|_F^2, \mbox{s.t. } \Bb \in \overline  \cN_{K'}(\vb).
\end{align*}

Let $\Pb_{\Vb^\perp}$ and $\Pb_{\Ub^\perp}$ be the projections on the perp column spaces of $\Vb \in \RR^{p_2 \times r}$ and $\Ub \in \RR^{p_1 \times r}$. We have the following result.

\begin{proposition}\label{trace:reg:projection:proposition} Solve the following optimization problem by PAVA
\begin{align*}
    (\hat t, \hat \sigma_1,\ldots, \hat \sigma_{\min(p_1,p_2) - r}) & = \argmin -2 t \tr(\Vb\T \Ab\T \Ub) - 2\sum_{i = 1}^{\min(p_1,p_2) - r}\bar\lambda_i \sigma_i + t^2 r + \sum_{i = 1}^{\min(p_1,p_2) - r} \sigma_i^2,\\
    & ~~~~~ \mbox{s.t. } t \geq \sigma_1 \geq \sigma_2 \geq \ldots \geq \sigma_{\min(p_1,p_2) - r},
\end{align*}
where $\bar \lambda_i$ form a decreasing sequence and are given by 
\begin{align*}
    \Pb_{\Ub^{\perp}}\Ab \Pb_{\Vb^{\perp}} = \sum_{i = 1}^{\min(p_1, p_2) - r} \bar \lambda_i \bar \ub_i \bar \vb_i\T.
\end{align*}
Then $\Bb$ is given by
\begin{align*}
    \Bb = \hat t \Ub \Vb\T + \sum_{i = 1}^{\min(p_1, p_2) - r} \hat \sigma_i \bar \ub_i \bar \vb_i\T.
\end{align*}
\end{proposition}

A summary of the debiasing procedure specific for trace regression is given in Algorithm \ref{algo_trace}. 

% \begin{algorithm}[H]
% \caption{Debias the $j$\textsuperscript{th} coefficient for trace regression.}
% \label{algo_trace}
% \hspace*{\algorithmicindent} \textbf{Input:} Two equal size partitions $(\overline\Xb, \overline\bY)$ and $(\tilde\Xb, \tilde\bY)$; $\hat\bbeta$ obtained by solving a Minkowski gauge minimization for trace regression. $K = \rho_{K'}(\vb) K'$ where $K'$ is the set of matrices with nuclear norm $\leq 1$.\\
% \hspace*{\algorithmicindent} \textbf{Initialize:} Empirical Gram matrix of the second partition $\hat\bSigma=\frac{1}{n}\tilde\Xb\T\tilde\Xb$.

% %\begin{algorithmic}
% \begin{enumerate}
%     \item Use Lemma \ref{lemma:find:v:trace:reg} to output $\vb$, and the set $K = \rho_{K'}(\vb) K'$.
    
%     \item Run Algorithm \ref{algo_step3}. Compute $\Pi_{\overline \cN_{K'}(\vb)}(\cdot) = \Pi_{ \cN_{K}(\vb)}(\cdot)$ by Proposition \ref{trace:reg:projection:proposition}, and apply Moreau's decomposition to get $\Pi_{\overline \cT_{K'}(\vb)}(\cdot) = \Pi_{ \cT_{K}(\vb)}(\cdot)$. For $\Pi_{-\overline \cT_{K'}(\vb)}(\cdot) = \Pi_{- \cT_{K}(\vb)}(\cdot)$ use \eqref{neg:cone:calc}. %note that $-\cT_K(\vb)=\cT_K(-\vb)$ since $-K = K$.\\
%     The debiased $j$\textsuperscript{th} coefficient $\hat\bbeta_{d}^{(j)} \leftarrow \operatorname{vec}(\vb)^{(j)} + n^{-1}\hat\bmeta\T \tilde\Xb\T (\tilde\bY - \tilde\Xb\operatorname{vec}(\vb))$.
% \end{enumerate}
% %\end{algorithmic}
% \end{algorithm}
\begin{algorithm}[ht]
\caption{Debias the $j$\textsuperscript{th} coefficient for trace regression.}
\label{algo_trace}
\begin{algorithmic}
\STATE\textbf{Input:} Two equal size partitions $(\overline\Xb, \overline\bY)$ and $(\tilde\Xb, \tilde\bY)$ (here these partitions are $n \times p_1p_2$ matrices); $\hat\bbeta$ obtained by solving a Minkowski gauge minimization for trace regression. $K = \rho_{K'}(\vb) K'$ where $K'$ is the set of matrices with nuclear norm $\leq 1$.\\
\STATE\textbf{Initialize:} Empirical Gram matrix of the second partition $\hat\bSigma=\frac{1}{n}\tilde\Xb\T\tilde\Xb$.
\begin{enumerate}
    \item Use Lemma \ref{lemma:find:v:trace:reg} to output $\vb$, and the set $K = \rho_{K'}(\vb) K'$.
    
    \item Run Algorithm \ref{algo_step3}. Compute $\Pi_{\overline \cN_{K'}(\vb)}(\cdot) = \Pi_{ \cN_{K}(\vb)}(\cdot)$ by Proposition \ref{trace:reg:projection:proposition}, and apply Moreau's decomposition to get $\Pi_{\overline \cT_{K'}(\vb)}(\cdot) = \Pi_{ \cT_{K}(\vb)}(\cdot)$. For $\Pi_{-\overline \cT_{K'}(\vb)}(\cdot) = \Pi_{- \cT_{K}(\vb)}(\cdot)$ use \eqref{neg:cone:calc}. %note that $-\cT_K(\vb)=\cT_K(-\vb)$ since $-K = K$.\\
    The debiased $j$\textsuperscript{th} coefficient $\hat\bbeta_{d}^{(j)} \leftarrow \operatorname{vec}(\vb)^{(j)} + n^{-1}\hat\bmeta\T \tilde\Xb\T (\tilde\bY - \tilde\Xb\operatorname{vec}(\vb))$. (here $\hat \bmeta$ is in vectorized form).
\end{enumerate}
\end{algorithmic}
\end{algorithm}

Finally, we should mention that works on debiasing trace regression are scarce in the literature. We would like to single out a couple of papers \citep{xia2021statistical, carpentier2018adaptive} which handle this problem. The methods proposed in these works differ significantly with our algorithm above. The work of \cite{carpentier2018adaptive} requires that the covariates satisfy the restricted isometry property which is arguably more stringent than our assumption on the covariance of the data, while \cite{xia2021statistical} require a certain condition on the SVD of $\bbeta^*$ which we do not need.

\section{Minkowski Gauge Regularization}\label{minkowski:gauge:regularization:sec}

In this section we present a new result regarding regularized least squares with a Minkowski gauge of a convex set. We then, show how one can construct the set $K$ and the vector $\vb$ in a specific instance of this setting --- the (regularized) group LASSO. Similarly to the previous section let $K'$ be a convex body containing $\mathbf{0}$ as an interior point, and $\rho_{K'}(\bbeta)$ be its associated Minkowski gauge. If one has a linear model, as in \eqref{gaussian_model}, the set $K'$ may incorporate some prior knowledge of the parameter $\bbeta^*$ and a natural procedure which has been previously proposed \cite[see equation (22) of ][e.g.]{chandrasekaran2012convex} is
\begin{align*}
\hat \bbeta_{\lambda} := \argmin_{\bbeta} n^{-1}\sum_{i \in [n]} (Y_i - \bX_i\T \bbeta)^2 + \lambda \rho_{K'}(\bbeta),
\end{align*}
for some tuning parameter $\lambda \geq 0$. One special example of this is the regularized version of LASSO e.g., where the set $K'$ is the cross-polytope. We now have the following result. %{\color{red} Need Gaussian assumptions etc.}
%{\color{red} change all occurrances of $\vb$ with $\wb$ in the proof too.}
\begin{theorem}\label{funamental:theorem:minkowski:reg} Let us observe $n$ i.i.d. samples from model \eqref{gaussian_model}, where $\bX_i \sim N(0,\bSigma)$, where $\bSigma$ has a bounded spectrum in the sense that there exist absolute constants $c,C > 0$ such that $c \leq \lambda_{\min}(\bSigma) \leq \lambda_{\max}(\bSigma) \leq C$. For any $\wb \in \RR^p$ there exists a $\lambda^*_{\wb}$ such that for any $\lambda \geq \lambda^*_{\wb}$: $\rho_{K'}(\hat \bbeta_{\lambda}) \leq \rho_{K'}(\wb)$. Suppose now that $\wb$ is such that $\overline w(\overline \cT_{K'}(\wb) \cap \mathbb{S}^{p-1}) \rightarrow \infty$ and $\overline w(\overline \cT_{K'}(\wb) \cap \mathbb{S}^{p-1}) = o(\sqrt{n})$. For any value of $\lambda \geq \lambda_{\wb}^*$ with high probability (converging to $1$ asymptotically) we have 
\begin{align*}
\|\hat \bbeta_{\lambda} - \bbeta^*\| \lesssim \|\bbeta^* - \wb\| + \frac{\sigma \overline w(\overline \cT_{K'}(\wb) \cap \mathbb{S}^{p-1})}{\sqrt{n}} + \lambda s(\wb),
\end{align*}
where 
\begin{align*}
s(\wb) = \sup_{\wb' : \rho_{K'}(\wb') \leq \rho_{K'}(\wb)} \frac{\rho_{K'}(\wb) - \rho_{K'}(\wb')}{\|\wb - \wb'\|}.
\end{align*}
In addition, define the quantity
%\begin{align*}
%s(\wb, r) = \sup_{\wb' : \rho_{K'}(\wb') \leq \rho_{K'}(\wb), \|\wb - \wb'\| = r} \frac{\rho_{K'}(\wb) - \rho_{K'}(\wb')}{\|\wb - \wb'\|}.
%\end{align*}
\begin{align*}
s(\wb, r) = \sup_{\wb' : \rho_{K'}(\wb') \leq \rho_{K'}(\wb), \|\wb - \wb'\| \leq r} \rho_{K'}(\wb) - \rho_{K'}(\wb').
\end{align*}
Then with high probability $\lambda_{\wb}^* \lesssim  \frac{ \big(\|\bbeta^* - \wb\| + \frac{\sigma \overline w(\overline \cT_{K'}(\wb) \cap \mathbb{S}^{p-1})}{\sqrt{n}}\big)^2}{s(\wb, r')}$, where $r' = c_0  \big(\|\bbeta^* - \wb\| + \frac{\sigma \overline w(\overline \cT_{K'}(\wb) \cap \mathbb{S}^{p-1})}{\sqrt{n}}\big)$ for some small constant $c_0 > 0$.
\end{theorem}

%{\color{red} Remark that $s(\vb)$ is less than or equal to the asphericity ratio.}
Before we proceed with an example of how to use Theorem \ref{funamental:theorem:minkowski:reg} we end this section with a quick remark about the quantity $s(\wb)$. Suppose $K'$ is balanced so that the Minkowski gauge is a semi-norm. In such cases it is simple to see by the definition of $s(\wb)$ that $s(\wb) \leq \gamma_{K'}(\wb)$, where $\gamma_{K'}(\wb)$ is the asphericity ratio as defined in Section \ref{minkowski:section}. %We will now illustrate the usage of Theorem \ref{funamental:theorem:minkowski:reg} with an example of regularized group LASSO.

\subsection{Example: Debiasing the regularized Group LASSO}\label{debias:group:lasso:regularied}

To illustrate the usefulness of Theorem \ref{funamental:theorem:minkowski:reg} we will now consider an example --- a regularized group LASSO estimator. Consider the set 
\begin{align}\label{group:LASSO:set}
K' =\{\bbeta: \|\bbeta\|_{2,1} \leq 1\},
\end{align} where the $\|\cdot\|_{2,1}$ is as defined in Section \ref{lasso:sec}. We will also be using the same notation as in Section \ref{lasso:sec}. We have the following lemma.

\begin{lemma}\label{lemma:group:lasso:svr} For the group LASSO we have $s(\wb) \leq \sqrt{s}$, where $s$ denotes the number of active groups in $\wb$. Furthermore, suppose that $\min_{G \in S} \|\wb_G\| \geq h$, where $S$ denotes the set of active groups. Then $s(\wb, r) \geq s h$, for $r = \sqrt{s h^2}$. %{\color{red} what is $c$?}
\end{lemma}

For brevity set 
\begin{align*}
H(s) := \sqrt{\frac{ B + (\sqrt{2\log(M - s)} + \sqrt{B})^2}{n}},
\end{align*}
where we remind the reader that $B$ denotes the maximum number of elements within a group. Now observe that for any vector $\bbeta^*$, by zeroing out the smallest in $\|\cdot\|$ norm groups of coefficients, one can find a vector $\bar \wb$ such that $\|\bbeta^* - \bar \wb\| + \sigma\frac{\overline w(\overline \cT_{K'}(\bar \wb) \cap \mathbb{S}^{p-1})}{\sqrt{n}} = \|\bbeta^* - \bar \wb\| +\sigma \sqrt{s}H(s)$ is minimal, where $s$ denotes the number of active groups in $\bar \wb$. Now we can construct $\wb$ from $\bar \wb$ by adding signal strength to each non-zero group in $\bar \wb$ in the following way $\wb_G = \bar \wb_G + \frac{\bar \wb_G}{\|\bar \wb_G\|}c\sigma H(s)$ for all $G \in S$ where $S$ denotes the acrive groups of $\bar \wb$. Clearly, by the triangle inequality, such $\wb$ satisfies the property $\|\bbeta^* -  \wb\| + \sigma \sqrt{s}H(s) \asymp \|\bbeta^* -  \bar \wb\| + \sigma \sqrt{s}H(s)$. By Lemma \ref{lemma:group:lasso:svr} with $h = c_0(\|\bbeta^* -  \wb\|/\sqrt{s} + \sigma H(s))$, we now know that $s(\wb) \leq \sqrt{s}$, and $s(\wb,r) \geq c_0(\|\bbeta^* -  \wb\|\sqrt{s}  + s \sigma H(s))$, and therefore $\lambda_{\wb}^* \lesssim \|\bbeta^* -  \wb\|/\sqrt{s}+ \sigma H(s)$, while the rate for the group LASSO estimators will become %{\color{red} this is quite unclear here}
\begin{align*}
\|\hat \bbeta - \bbeta^*\| \lesssim  \|\bbeta^* - \wb\| + \sigma\sqrt{s} H(s) + \lambda \sqrt{s},
\end{align*}
for $\lambda \geq \lambda_{\wb}^*$. This guarantee for the group LASSO appears to be novel in the literature (albeit it is valid in the Gaussian design model only) --- the novelty being that $s$ here is not the group sparsity of $\bbeta^*$ but of the vector $\wb$. We note that when $\lambda$ is selected of the (optimal) order of $\|\bbeta^* -  \wb\|/\sqrt{s} + \sigma H(s)$, this bound is at least as good as $\frac{\overline w(\overline \cT_{K'}(\bbeta^*)\cap \mathbb{S}^{p-1})}{\sqrt{n}} \asymp\sigma \sqrt{s_{\bbeta^*}} H(s_{\bbeta^*})$ where $s_{\bbeta^*}$ is the number of active groups in $\bbeta^*$ (this is so by the construction of $\wb$). Finally, we would like to remark that the group LASSO example cannot be tackled using a Minkowski gauge selector as in \cite{cai2016geometric} as their bounds using the asphericity ratio and the Gaussian width of the set $\Xb K'$ lead to sub-optimal guarantees. 

We now turn our attention to how one can construct the set $K$ and the vector $\vb$, required for the debiasing algorithm. The set $K$ will be selected so that it is proportional to the set $K'$ from the Minkowski gauge. If one has knowledge on the true coefficient $\rho_{K'}(\bbeta^*)$ the set $K$ can simply be selected as $\rho_{K'}(\bbeta^*) K'$. Then $\vb$ can be constructed as in Section \ref{lasso:sec}. We now present an alternative route which only assumes an upper bound on the number of active groups in $\bbeta^*$. Suppose $s_{\bbeta^*} \leq \bar s$, where $\bar s$ is a known integer.  We propose to solve the following program
\begin{align}\label{find:v:reg:group:LASSO}
\max_{\vb} \|\vb\|_{2,1}, \mbox{ s.t. } \nonumber \\
\|\hat \bbeta - \vb\| \leq C \sqrt{\frac{\bar s \log M + \bar s B}{n}},\\
 \mbox{ at most $\bar s$ groups of $\vb$ are active}.\nonumber
\end{align}
Clearly, if the constant $C$ is big enough (assuming $\sigma$ is a constant which does not scale with $n$), $\bbeta^*$ will be a feasible point and therefore it will follow that $\|\vb\|_{2,1} \geq \|\bbeta^*\|_{2,1}$ which is sufficient so that $K = \|\vb\|_{2,1}K'$ contains the point $\bbeta^*$. In addition the requirement $\overline w(\overline \cT_{K'}(\vb) \cap \mathbb{S}^{p-1}) \|\vb - \bbeta^*\| \lesssim  \sqrt{\bar s \log( M-s) + \bar s B} C \sqrt{\frac{\bar s \log M  + \bar s B}{n}} = o(1)$ assuming that $C \frac{\bar s \log M  + \bar s B}{\sqrt{n}} = o(1)$. Hence one way to ensure $C$ is big enough is to select $C = (\frac{\sqrt{n}}{\bar s \log M  + \bar s B})^{\gamma}$ for some $\gamma < 1$ assuming $\frac{\bar s \log M  + \bar s B}{\sqrt{n}} = o(1)$. 

We now give the algorithm to solve the above optimization program which may seem combinatorial at first glance. Sort $\|\hat \bbeta_{G}\|_2$ for $G \in \cG$ in decreasing order, and take the top $\bar s$ groups, breaking ties arbitrarily. Then $\vb$ is given by 
\begin{align}\label{vb:formula}
\vb_{G_{\#i}} = \hat \bbeta_{G_{\#i}} + \frac{\hat \bbeta_{G_{\#i}}}{\|\hat \bbeta_{G_{\#i}}\|}\frac{\sqrt{C^2\frac{\bar s \log M + \bar s B}{n} - \sum_{i = \bar s+1}^M \|\bbeta_{G_{\#i}}\|^2}}{\sqrt{\bar s}},
\end{align}
where $\hat \bbeta_{G_{\#i}}$ denotes the $i$-th largest in $\|\cdot\|$ group, and where $i$ ranges in the set $[\bar s]$. Since the proof is almost identical to the one of Lemma \ref{slope_unknownl1_step2_sol} to follow it is omitted. The final thing is how to project onto the set $\overline \cT_{K'}(\vb) = \cT_{K}(\vb)$ for $K = \|\vb\|_{2,1}K'$. However, note that this algorithm was already given in Section \ref{lasso:sec}. We summarize the debiasing for the regularized group LASSO in Algorithm \ref{algo_lasso_regularized} below.

% \begin{algorithm}[H]
% \caption{Debias the $j$\textsuperscript{th} Coefficient for regularized group LASSO}
% \label{algo_lasso_regularized}
% \hspace*{\algorithmicindent} \textbf{Input:} Two equal size partitions $(\overline\Xb, \overline\bY)$ and $(\tilde\Xb, \tilde\bY)$; $\hat\bbeta$ obtained by solving a constrained group LASSO problem. %$K = \{\bbeta: \|\bbeta\|_1 \leq \|\bbeta^*\|_1\}$.\\

% \hspace*{\algorithmicindent} \textbf{Initialize:} Empirical Gram matrix of the second partition $\hat\bSigma=\frac{1}{n}\tilde\Xb\T\tilde\Xb$.

% %\begin{algorithmic}
% \begin{enumerate}
%     \item Solve optimization program \eqref{find:v:reg:group:LASSO} with \eqref{vb:formula}. The set $K = \|\vb\|_{2,1}K'$ where $K'$ is given by \eqref{group:LASSO:set}.
    
%     \item Run Algorithm \ref{algo_step3}. Compute $\Pi_{\overline \cN_{K'}(\vb)}(\cdot)$ by \eqref{proj_tan1}, \eqref{proj_tan2}, and apply Moreau's decomposition to get $\Pi_{\overline \cT_{K'}(\vb)}(\cdot)$. For $\Pi_{-\overline \cT_{K'}(\vb)}(\cdot)$ use \eqref{neg:cone:calc}. %note that $-\cT_K(\vb)=\cT_K(-\vb)$ since $-K = K$.\\
%     The debiased $j$\textsuperscript{th} coefficient $\hat\bbeta_{d}^{(j)} \leftarrow \vb^{(j)} + n^{-1}\hat\bmeta\T \tilde\Xb\T (\tilde\bY - \tilde\Xb\vb)$.
% \end{enumerate}
% %\end{algorithmic}
% \end{algorithm}
\begin{algorithm}[ht]
\caption{Debias the $j$\textsuperscript{th} Coefficient for regularized group LASSO}
\label{algo_lasso_regularized}
\begin{algorithmic}
\STATE\textbf{Input:}  Two equal size partitions $(\overline\Xb, \overline\bY)$ and $(\tilde\Xb, \tilde\bY)$; $\hat\bbeta$ obtained by solving a constrained group LASSO problem. 
\STATE\textbf{Initialize:} Empirical Gram matrix of the second partition $\hat\bSigma=\frac{1}{n}\tilde\Xb\T\tilde\Xb$.
\begin{enumerate}
    \item Solve optimization program \eqref{find:v:reg:group:LASSO} with \eqref{vb:formula}. The set $K = \|\vb\|_{2,1}K'$ where $K'$ is given by \eqref{group:LASSO:set}.
    
    \item Run Algorithm \ref{algo_step3}. Compute $\Pi_{\overline \cN_{K'}(\vb)}(\cdot)$ by \eqref{proj_tan1}, \eqref{proj_tan2}, and apply Moreau's decomposition to get $\Pi_{\overline \cT_{K'}(\vb)}(\cdot)$. For $\Pi_{-\overline \cT_{K'}(\vb)}(\cdot)$ use \eqref{neg:cone:calc}. %note that $-\cT_K(\vb)=\cT_K(-\vb)$ since $-K = K$.\\
    The debiased $j$\textsuperscript{th} coefficient $\hat\bbeta_{d}^{(j)} \leftarrow \vb^{(j)} + n^{-1}\hat\bmeta\T \tilde\Xb\T (\tilde\bY - \tilde\Xb\vb)$.
\end{enumerate}
\end{algorithmic}
\end{algorithm}

We conclude this section by mentioning that there are very few works debiasing the group LASSO estimator that we are aware of \citep{mitra2016benefit,bellec2019second, cai2022sparse}. All of these papers have very different strategies: \cite{mitra2016benefit} debias the scaled group LASSO, \cite{bellec2019second} work only in the regime $p/n \leq \gamma$ for some fixed constant $\gamma \in (0,\infty)$, while \cite{cai2022sparse} debias the sparse group lasso which has both $\|\cdot\|_1$ and $\|\cdot\|_{2,1}$ penalties included.

Our final remark in this section is that it is also not too hard to see that Theorem \ref{funamental:theorem:minkowski:reg} leads to rate optimal bounds for the case of trace regression which can then be debiased with our algorithm. For brevity (and since we already saw an example of this in the preceding section) we will omit these calculations.

\section{SLOPE and Square-Root SLOPE}
\label{slope_sqrtslope:sec}
In this section we show how our debiasing scheme can be used in SLOPE and square-root SLOPE estimators. It is worthwhile to also mention that even though this section is dedicated to SLOPE and square-root SLOPE, the same debiasing procedure also works for some other types of estimators whose error rate $\|\hat\bbeta-\bbeta^*\|$ is tractable. Examples include the LASSO penalized version \citep{tibshirani1996regression}, MCP \citep{zhang2010nearly}, SCAD penalized estimator \citep{fan2001variable}, elastic net \citep{zou2005regularization} etc.

SLOPE was first proposed by \cite{bogdan2015slope} as
\begin{align}
    \hat\bbeta = \argmin_{\bbeta\in\mathbb{R}^p}\frac{1}{n}\|\overline\bY-\overline{\Xb}\bbeta\|^2 + \lambda_1|\bbeta_{\#1}| + \lambda_2|\bbeta_{\#2}| + \ldots + \lambda_p|\bbeta_{\#p}|,
\label{slope_opt}
\end{align}
where $\lambda_1\geq\lambda_2\geq\ldots\geq\lambda_p$, and $|\bbeta_{\#1}|\geq|\bbeta_{\#2}|\geq\ldots\geq|\bbeta_{\#p}|$ are the entries of $\bbeta$ sorted in a decreasing order in terms of their absolute value. Let $A\geq2(4+\sqrt{2})$ be a constant. According to \cite[Corollary 6.2]{bellec2018slope}, if one picks
\begin{align}\label{lambdaj:def}
    \lambda_i=A\sigma\sqrt{\frac{\log(2p/i)}{n}}, i\in[p],
\end{align}
the SLOPE estimator achieves the optimal error rate:
\begin{align}
    \|\hat \bbeta - \bbeta^*\| \leq \overline{C}\sigma \sqrt{\frac{s \log (2ep/s)}{n}},
\label{slope_est_error}
\end{align}
where $\overline{C}>0$ is a constant and $s$ is the number of non-zero coordinates in $\bbeta^*$.

The square-root SLOPE \citep{stucky2017sharp} is introduced to alleviate the restriction of knowing $\sigma$ while still achieving the optimal rate \eqref{slope_est_error}. It estimates $\sigma$ and $\bbeta$ simultaneously:
\begin{align}
    (\hat{\bbeta}, \hat{\sigma}) \in \argmin_{\bbeta\in\mathbb{R}^p, \sigma>0} \sigma + \frac{1}{n\sigma}\|\overline\bY-\overline{\Xb}\bbeta\|^2 + \lambda_1|\bbeta_{\#1}| + \lambda_2|\bbeta_{\#2}| + \ldots + \lambda_p|\bbeta_{\#p}|.
\label{sqrt_slope_opt}
\end{align}
Let $A'\geq4(4+\sqrt{2})$ be a constant. \cite[Theorem 6.1]{derumigny2018improved} shows that if the constraint parameters are picked as
\begin{align}\label{lambdaj:square:root:slope:def}
    \lambda_i=A'\sqrt{\frac{\log(2p/i)}{n}},\,\,i\in[p],
\end{align}
the square-root SLOPE will achieve the optimal rate \eqref{slope_est_error}.

We will now suggest two ways to debias both the SLOPE and square-root SLOPE estimator $\hat \bbeta$. The first assumes knowledge on $\|\bbeta^*\|_1$, while the second assumes knowledge of an upper bound on sparsity $\|\bbeta^*\|_0\leq s^u$.

First, suppose that we know $\|\bbeta^*\|_1$ and $\bbeta^*$ is $s$-sparse, but $s$ is not necessarily known. Then the approaches of both SLOPE and square-root SLOPE are identical to how we debias the constrained LASSO problem in Section \ref{lasso:sec}, since the convex set $K=\{\bbeta: \|\bbeta\|_1 \leq \|\bbeta^*\|_1\}$ can be used in the same manner as in the constrained LASSO case. In step \ref{step_1} we find a $\vb=\vb_s$ such that $\|\hat \bbeta - \vb_s\| + \sqrt{s\log(ep/s)/n}$ is the smallest given that $\vb_s$ is $s$-sparse and $\|\vb_s\|_1 = \|\bbeta^*\|_1$. Next we solve step \ref{step_2} with such a vector $\vb$ and a convex set $K$.

Second, we consider the case when $\|\bbeta^*\|_1$ is unknown, but an upper bound on sparsity $\|\bbeta^*\|_0\leq s^u$ is available. In this case we do not have a prior knowledge of the convex parameter space $K$ in which $\bbeta^*$ belongs to. Instead we will construct $K$ and the vector $\vb$ required in step \ref{step_1} ``from scratch''. To find a vector $\vb$ which satisfies the condition in step \ref{step_1}, we propose to solve the following optimization program
\begin{align}
    \argmax \|\vb\|_1, \mbox{ s.t. } \|\vb - \hat \bbeta\| \leq C \sqrt{\frac{s^u \log(2ep/s^u)}{n}} \text{ and } \|\vb\|_0 \leq s^u,
\label{slope_unknownl1_step2}
\end{align}
for a sufficiently large constant $C$. Since the function $s \mapsto s\log(2ep/s)$ is increasing in $s$, $\bbeta^*$ is a feasible point when $C$ is sufficiently large. Theorem \ref{main_rst_slope} proves that the solution $\vb$ of the above optimization program \eqref{slope_unknownl1_step2} satisfies the condition in step \ref{step_1} with the set 
\begin{align}
\label{slope_k}
    K=\{\bbeta: \|\bbeta\|_1 \leq \|\vb\|_1\}.
\end{align}
Notice that since $\bbeta^*$ is a feasible point of \eqref{slope_unknownl1_step2} with a proper choice of $C$, $\vb$ also satisfies $\|\vb\|_1 \geq \|\bbeta^*\|_1$ which implies $\bbeta^* \in K$. In order for us to state our next result we need to give a definition from \citep[Page 10]{bellec2018slope}.

\begin{definition}[Weighted Restricted Eigenvalue (WRE) condition]
For a design matrix $\Xb \in \RR^{n \times p}$ satisfying $\|\Xb \eb^{(j)}\| \leq \sqrt{n}$ for all $j \in [p]$ define
\begin{align*}
    \vartheta(s, c_0) = \min_{\bdelta\in\{\bdelta:\sum_{j=1}^p\lambda_j|\delta_{\#j}|\leq (1+c_0)\|\bdelta\|\sqrt{\sum_{j=1}^s\lambda_j^2}\}, \bdelta\neq\mathbf{0}}\,\,\frac{1}{\sqrt{n}} \frac{\|\Xb\bdelta\|}{\|\bdelta\|},
% \label{def_vartheta}
\end{align*}
where $\lambda_j$ are given in \eqref{lambdaj:def} (or equivalently in \eqref{lambdaj:square:root:slope:def}). A design matrix $\Xb$ as above is said to satisfy WRE if $\vartheta(s, c_0) > 0$.
 %\citep[Page 10]{bellec2018slope}
%and 
%\begin{align}
%\label{def_vartheta_star}
%    \vartheta^* = \min_{c_0>0, s\in[1, s^u]} \vartheta(s, c_0)
%\end{align}
\end{definition}
The next theorem will condition on the event that $\overline \Xb$ (the design matrix from the first split of the data) satisfies the WRE for $s^u$ and $c_0 = 3$ for SLOPE and $c_0 = 20$ for square-root SLOPE.% with $\vartheta(s^u, 3) = \vartheta^*$.

\begin{theorem}
\label{main_rst_slope}
Consider the same setting as Theorem \ref{debiase_formula_applicable_unkowncov}. Suppose $\|\bbeta^*\|_0\leq s^u$. Condition on the event that the matrix $\overline \Xb$ satisfies the WRE with $\vartheta^* := \vartheta(s^u, 3)$ for SLOPE and $\vartheta^* := \vartheta(s^u, 20)$ for square-root SLOPE. With a proper choice of $C \gtrsim \frac{\sigma}{\vartheta^*}$ satisfying $C \frac{s^u\log (ep/s^u)}{\sqrt{n}}=o(1)$, for $\hat\bbeta$ as a SLOPE estimator obtained via \eqref{slope_opt} or a square-root SLOPE estimator obtained via \eqref{sqrt_slope_opt}, the solution $\vb$ of \eqref{slope_unknownl1_step2}, and the set $K = \{\bbeta: \|\bbeta\|_1 \leq \|\vb\|_1\}$ satisfy the condition needed in step 1 of Algorithm \ref{algo_unknowncov}.
\end{theorem}
\begin{remark}\label{slope:remark:important}
We now comment on the condition that $\overline \Xb$ satisfies the WRE with $ \vartheta(s^u, c_0)$ for $c_0 = 3$ or $c_0 = 20$. By Theorem 8.3 of \cite{bellec2018slope} we know that for a large class of data generating mechanisms (including Gaussian and bounded mean-zero $\bX_i$ for $i \in [n]$) if $\bSigma$ has bounded from below by $\kappa > 0$ eigenvalue, and in addition $\max_{i} \bSigma_{ii} \leq \frac{1}{2}$ then if $n \gtrsim \frac{(1 + c_0)^2}{\kappa^2} s^u\log(2ep/s^u)$ the matrix $\overline{\Xb}$ will satisfy WRE with $s^u$ and $c_0$ with $\vartheta(s^u, c_0) = \kappa/\sqrt{2}$ with high probability. It follows that when $\sigma$ is fixed, $C \gtrsim \frac{\sqrt{2}\sigma}{\kappa}$ suffices to meet the requirements in Theorem \ref{main_rst_slope}. This is surely satisfied if one picks $C \gg 1$. Below we give an example of such a choice for $C$.
\end{remark}

% \begin{theorem}
% \label{main_rst_slope}
% Consider the same setting as Theorem \ref{debiase_formula_applicable_unkowncov}. Suppose $\|\bbeta^*\|_0\leq s^u$ and $s^u=o(\sqrt{n}/\log (p/s^u))$. With a proper choice of $C$, for $\hat\bbeta$ as a SLOPE estimator obtained via \eqref{slope_opt} or a square-root SLOPE estimator obtained via \eqref{sqrt_slope_opt}, the solution $\vb$ of \eqref{slope_unknownl1_step2} satisfies the condition needed in step \ref{step_1}.
% \end{theorem}

%As for the choice of $C$ in \eqref{slope_unknownl1_step2}, it should be no smaller than the constant in \eqref{slope_est_error}. At the same time, $C$ shouldn't be too large so that the resulting $\vb$ in \eqref{slope_unknownl1_step2} still satisfies $\overline w(\cT_K(\vb)\cap \mathbb{S}^{p-1})\|\vb-\bbeta^*\|=o_p(1)$ as required in step \ref{step_1}. 
From the proof of Theorem \ref{main_rst_slope} it becomes evident that in principle, we can select any small enough $C > \overline{C}\sigma$ in \eqref{slope_est_error} since that will ensure that $\bbeta^*$ is a feasible point in \eqref{slope_unknownl1_step2}. One might directly analyze an upper bound on $\overline{C}$ according to the high probability upper bounds on $\|\hat \bbeta - \bbeta^*\|$ given in \cite[Corollary 6.2]{bellec2018slope} and \cite[Corollary 6.2]{derumigny2018improved}. However, such an upper bound on $\overline{C}$ requires finding weighted restricted eigenvalues and may not be easily computable. %solving a sparse eigenvalue problem, which is combinatorial in general.
Here we suggest an alternative way to obtain a slightly larger $C$ for the debiasing purpose. This is possible under the assumptions of Remark \ref{slope:remark:important}. We claim that $C$ can be picked as
\begin{align}
    C \sim \Big( \frac{\sqrt{n}}{s^u\log(ep/s^u)} \Big)^{\gamma}\text{ where }0< \gamma < 1 \text{ is a small number.}
\label{slope_c_pick}
\end{align}
In this way, if $s^u=o(\sqrt{n}/\log(ep/s^u))$, the order of $C$ in \eqref{slope_c_pick} is slightly larger than the constant in \eqref{slope_est_error} which is $O(1)$ (assuming $\sigma = O(1)$) under the assumptions of Remark \ref{slope:remark:important}. Thus $\bbeta^*$ is guaranteed to be a feasible point of \eqref{slope_unknownl1_step2}. At the same time, $C$ is only moderately large so that $\overline w(\cT_K(\vb)\cap \mathbb{S}^{p-1})\|\vb-\bbeta^*\|=o_p(1)$ still holds under the same assumptions as in Theorem \ref{main_rst_slope}. This is because in the proof of Theorem \ref{main_rst_slope} we establish that with high probability
\begin{align*}
    \overline w(\cT_K(\vb)\cap \mathbb{S}^{p-1})\|\vb-\bbeta^*\| & \lesssim C\frac{s^u\log(ep/s^u)}{\sqrt{n}} \sim \bigg(\frac{s^u\log(ep/s^u)}{\sqrt{n}}\bigg)^{1-\gamma} = o(1).
\end{align*}
After picking a proper $C$, there are no obstacles to compute a $\vb$ in step \ref{step_1} since the optimization program \eqref{slope_unknownl1_step2} actually has an analytical solution as shown in Lemma \ref{slope_unknownl1_step2_sol}. 
\begin{lemma}
\label{slope_unknownl1_step2_sol}
The solution of \eqref{slope_unknownl1_step2} is
\begin{align*}
    \vb_{\#i} = \begin{cases}
    \hat \bbeta_{\#i} + \sign(\hat \bbeta_{\#i}) c, & i = 1,\ldots, s^u\\
    0, & \text{otherwise}.
    \end{cases}
\end{align*}
where $c=\sqrt{\bigg(C^2\frac{s^u \log 2ep/s^u}{n} - \sum_{i = s^u + 1}^p \hat \bbeta_{\#i}^2\bigg)/s^u}$, and ties in $\hat \bbeta_{\# i}$ are broken arbitrarliy, and with a slight abuse of notation we assign the same index for $\vb_{\#i}$ in $\vb$ as $\hat \bbeta_{\#i}$ has in $\hat \bbeta$.
\end{lemma}
Notice that $C$ should be selected so that we are able to compute $c$ as a positive real number, hence it should satisfy $C \geq \sqrt{\frac{n\sum_{i = s^u + 1}^p \hat \bbeta_{\#i}^2}{s^u \log (2ep/s^u)}}$.
%\footnote{Observe that this does not imply that $C$ is ``too large''. From \eqref{slope_opt} we know that $\sqrt{\|\hat \bbeta_{S_*} - \bbeta^*_{S_*}\|^2 + \|\hat \bbeta_{S_*^c}\|^2} \leq \overline{C}\sigma \sqrt{\frac{s \log (2ep/s)}{n}}$, where $S_*$ denotes the support of $\bbeta^*$. Next since $s^u \geq s$ it follows that $\sqrt{\sum_{i = s^u + 1}^p \hat \bbeta_{\#i}^2} \leq \|\hat \bbeta_{S_*^c}\|$, which shows that if $C > \overline{C}\sigma$ the condition will be met.}
%
Observe that this does not imply that $C$ is ``too large''. From \eqref{slope_opt} we know that $\sqrt{\|\hat \bbeta_{S_*} - \bbeta^*_{S_*}\|^2 + \|\hat \bbeta_{S_*^c}\|^2} \leq \overline{C}\sigma \sqrt{\frac{s \log (2ep/s)}{n}}$, where $S_*$ denotes the support of $\bbeta^*$. Next since $s^u \geq s$ it follows that $\sqrt{\sum_{i = s^u + 1}^p \hat \bbeta_{\#i}^2} \leq \|\hat \bbeta_{S_*^c}\|$, which shows that if $C > \overline{C}\sigma$ the condition will be met.
After one finds $\vb$ in step \ref{step_1}, one can compute the auxiliary vector $\hat\bmeta$ in step \ref{step_2} based on $\vb$ and $K=\{\bbeta: \|\bbeta\|_1 \leq \|\vb\|_1\}$, and then use $\hat\bmeta$ to construct the debiased estimator $\hat\bbeta_d$ and the confidence interval as \eqref{ci}. When constructing the confidence interval, we estimate $\sigma$ via $\hat\sigma = \sqrt{n^{-1}\sum_{i\in[n]} ( Y_i -  \bX_i\T \hat \bbeta)^2}$ on the first sample split. The following Lemma \ref{sigma_hat_slope} coupled with Theorem \ref{sigma_hat_rate} together show that we are able to get a consistent estimator of $\sigma$.
\begin{lemma}
\label{sigma_hat_slope}
Consider the same setting as Theorem \ref{main_rst_slope} where $\hat\bbeta$ is a SLOPE or square-root SLOPE estimator. Then under the conditions of Remark \ref{slope:remark:important}, Theorem \ref{sigma_hat_rate} applies with
\begin{align*}
    \delta \asymp \sigma \kappa^{-1}\sqrt{s^u \log(2ep/s^u)}.
\end{align*}
\end{lemma}

Lemma \ref{sigma_hat_slope} establishes that it is possible to consistently estimate $\sigma$, and therefore we can construct confidence intervals as in \eqref{ci:sigma:est}. We end this section with two remarks regarding the choice of $s^u$ and what ``classical'' debiasing methods can achieve in the SLOPE, or square-root SLOPE problems. 

\begin{remark}\label{slope:remark:slogp}
Since $s^u \geq s$, assuming $s^u \log (ep/s^u) = o(\sqrt{n})$ implies that $s \log (ep/s) = o(\sqrt{n})$ for the true sparsity $s$. By the work of \cite{cai2017confidence} we know the latter condition is nearly necessary in the case of sparse linear regression with unknown covariance. In fact  \cite{cai2017confidence} show that the length of the confidence interval is $ \gtrsim \max\bigg\{\frac{1}{\sqrt{n}}, \frac{s \log (ep/s)}{n}\bigg\}$. Thus if $s \log (ep/s) = O(\sqrt{n})$ interval length of the order of $\frac{1}{\sqrt{n}}$ is possible. However, in practice it is often assumed that $s \log (ep/s) = o(\sqrt{n})$ in order to achieve an exact asymptotic $(1-\alpha)$-level confidence interval. We now provide some guidance on selecting $s^u$. In principle it is difficult if not impossible to estimate an upper bound on $s$ from the data. However, in order for the debiasing to work we do need $s \log (ep/s) = o(\sqrt{n})$. If the practitioner has prior knowledge on the precise rate  $r_n := \frac{s \log (ep/s)}{\sqrt{n}}$, the practitioner can select any $s^u$ such that $\frac{s^u \log (e p/s^u)}{\sqrt{n}} = \sqrt{r_n}$, e.g. and this will work asymptotically. On the other hand, if information on $r_n$ is not available but it is known that $\frac{s \log ep/s}{\sqrt{n}} = o(1)$, the practitioner may opt for devising a slightly conservative confidence interval, by selecting $s^u$ such that $\frac{s^u \log (ep/s^u)}{\sqrt{n}} = c$ for some small constant $c$. It is not too hard to see that in such a setting, the term $|\Delta_j|$ from Theorem \ref{debiase_formula_applicable_unkowncov} will be asymptotically smaller than 
\begin{align*}
    |\Delta_j| \leq \sqrt{n}\rho c \|\bbeta^* - \vb\|,
\end{align*}
where $\rho$ is the tuning parameter from \eqref{opt_step3} of Algorithm \ref{algo_unknowncov}. Now by the triangle inequality $\|\bbeta^* - \vb\| \leq \|\bbeta^* - \hat \bbeta\| + \|\hat \bbeta - \vb\| \leq 2 C\frac{\sqrt{s^u \log 2ep/s^u}}{\sqrt{n}}$ since $\bbeta^*$ is a feasible point of \eqref{slope_unknownl1_step2}. Set $K := 2C\rho \frac{\sqrt{s^u \log 2ep/s^u}\sqrt{s^u \log ep/s^u}}{\sqrt{n}}$, where $C$ (note that any fixed constant $C$ here will do since $s^u \gg s$). is the constant from \eqref{slope_unknownl1_step2}. Therefore the confidence interval from \eqref{ci:sigma:est} widened by $\pm \frac{K}{\sqrt{n}}$ will be a valid $\frac{1}{\sqrt{n}}$-confidence interval of $\bbeta^{*(j)}$.
\end{remark}

\begin{remark}
Theorem \ref{main_rst_slope} and Remark \ref{slope:remark:slogp} point out that our debiasing algorithm works for SLOPE as long as $s=o(\sqrt{n}/\log ep/s)$. Clearly this is less stringent than the condition $s=o(\sqrt{n}/(\log ep/s)^{3/2})$. Such a condition appears necessary if one opts for applying previous debiasing algorithms and their analysis such as the one proposed by \citep[Algorithm 1]{javanmard2014confidence}. To see why the condition $s=o(\sqrt{n}/(\log e p/s)^{3/2})$ arises, the reader is referred to \cite[eq (9)]{javanmard2018debiasing} which summarizes well the standard argument for the analysis of why debiasing works. It relies on an $\ell_1-\ell_\infty$ H\"{o}lder's inequality. While the SLOPE or square-root SLOPE do not have a direct $\ell_1$ guarantee for their $\hat \bbeta$ estimates, a sub-optimal guarantee may be easily derived from \citep{bellec2018slope, derumigny2018improved}. It is simple to see that
 \begin{align*}
     \sigma \|\hat\bbeta-\bbeta^*\|_1/\sqrt{n}\lesssim\|\hat\bbeta-\bbeta^*\|_*\lesssim \sigma^2 s\log(ep/s)/n,
 \end{align*}
 where $\|\vb\|_* = \sum_{j \in [p]} \lambda_j |\vb_{\#j}|$, where $\lambda_j$ are as in \eqref{lambdaj:def}. In contrast, in the LASSO case one may bound $\|\hat \bbeta - \bbeta^*\|_1 \lesssim \sigma  s \sqrt{\frac{\log(p)}{n}}$ \citep[Section 7]{wainwright2019high}. One can see that SLOPE has an extra $\sqrt{\log (ep/s)}$ factor in the $\ell_1$-bound in comparison with LASSO, hence the extra $\sqrt{\log (ep/s)}$ factor in the condition $s=o(\sqrt{n}/(\log ep/s)^{3/2})$.
\end{remark}

In the following two subsections we give the detailed procedures about how to debias SLOPE and square-root SLOPE estimator, as specific instances of Algorithm \ref{algo_unknowncov}.

% \subsection{Debias SLOPE without upper bound on sparsity?}

% Here I will outline my idea of how we can avoid having knowledge of $s_u$. It may be wrong\ldots

% First, let us select $s_b$ (b stands for big) such that $s_b \log ep/s_b = \sqrt{n}$ (or at least approximately). This $s_b > s^*$ asymptotically. Then we may trim the $\hat \beta$ by just retaining the largest $s_b$ coefficients. We know that $$

% $n^{-1} \sum(Y_i - \bX_i\T \hat \bbeta)^2 -  \hat \sigma^2$

%%%%%%%%
\subsection{Debiasing Algorithm for SLOPE}
\label{slope:case}

We start by briefly summarizing how to solve the SLOPE $\hat\bbeta$ in \eqref{slope_opt}. The reader is encouraged to read the full details of the implementation which was first described in \cite{bogdan2015slope}. The SLOPE has a non-differentiable objective function, which can be solved by proximal gradient descent. A detailed introduction of the proximal gradient methods can be found in \cite[Chapter 2]{nesterov2003introductory}. The basic idea is: the objective function in \eqref{slope_opt} can be written as the sum of a convex differentiable function $f_1(\bbeta)=\frac{1}{n}\|\overline\bY-\overline{\Xb}\bbeta\|^2$ and a convex non-differentiable function $f_2(\bbeta)=\lambda_1|\beta_{\#1}| + \lambda_2|\beta_{\#2}| + \ldots + \lambda_p|\beta_{\#p}|$. For a convex optimization program whose objective function can be written as $f(\bbeta)=f_1(\bbeta)+f_2(\bbeta)$, where $f_1$ is differentiable but $f_2$ is not, each step of the proximal gradient method can be written as
\begin{align}
    \bbeta_{n+1} = \operatorname{prox}_{h_n}\Big( \bbeta_n - h_n\nabla f_1(\bbeta_n) \Big),
\label{slope_prox_step}
\end{align}
where $h_n$ is the step size, and $\operatorname{prox}_{h_n}(\cdot)$ is the proximal mapping defined as
\begin{align*}
    \operatorname{prox}_{h}(\xb) = \argmin_{\zb} \frac{1}{2h}\|\xb-\zb\|^2 + f_2(\zb).
\end{align*}
One can see that the proximal mapping in \eqref{slope_prox_step} forces the new candidate $\bbeta_{n+1}$ to stay close to the gradient update of $f_1$, and also makes $f_2$ small. The proximal mapping can be solved with the PAVA algorithm for isotonic regression. See \cite[Algorithm 3]{bogdan2015slope} for details. 

After solving $\hat\bbeta$, we debias it. The vector $\vb$ in step \ref{step_2} can be computed analytically by Lemma \ref{slope_unknownl1_step2_sol} with $C$ picked according to \eqref{slope_c_pick}, and $K$ is constructed as $K=\{\bbeta: \|\bbeta\|_1 \leq \|\vb\|_1\}$. Then in step \ref{step_2} we use $\vb$ and $K$ to get $\hat\bmeta$ via \eqref{opt_step3}. This can be done in the same way as in step 2 of the LASSO version Algorithm \ref{algo_lasso} since in both cases the set $K$ is an $\ell_1$ ball.
\begin{algorithm}[ht]
\caption{Debias the $j$\textsuperscript{th} Coefficient in SLOPE}
\label{algo_slope}
\begin{algorithmic}
\STATE \textbf{Input:} Two equal size partitions $(\overline\Xb, \overline\bY)$ and $(\tilde\Xb, \tilde\bY)$; $\hat\bbeta$ as a SLOPE estimator. $s^u$ upper bound on $s$, $C$ a sufficiently large tuning parameter.\\
\STATE \textbf{Initialize:} Empirical Gram matrix of the second partition $\hat\bSigma=\frac{1}{n}\tilde\Xb\T\tilde\Xb$.
\begin{enumerate}
    % \item Initialize $\bbeta_0$, pick a step size $h$.\\
    % \algorithmicrepeat\\
    % \quad\quad $\bbeta_n^{new} = \bbeta_n-\frac{h}{n}\overline\Xb\T(\overline\Xb\bbeta_n-\overline\bY)$\\
    % \quad\quad $\bbeta_{n+1}\leftarrow \sign(\bbeta_n^{new})(|\bbeta_n^{new}|-\blambda\frac{h}{n})_+$\\
    % \algorithmicuntil\,\,converge\\
    % $\hat\bbeta\leftarrow\bbeta_{n+1}$.
    \item $c\leftarrow\sqrt{\bigg(C^2\frac{s^u \log 2ep/s^u}{n} - \sum_{i = s^u + 1}^p \hat \bbeta_{\#i}^2\bigg)/s^u}$,  $\vb \leftarrow (0, \ldots, 0)$\\
    Assign $\vb_{\#i} = \hat\bbeta_{\#i} + \sign(\hat \bbeta_{\#i}) c \text{ for } i = 1,\ldots, s^u$
    
    \item Run Algorithm \ref{algo_step3}. Compute $\Pi_{\cN_{K}(\vb)}(\cdot)$ by \eqref{proj_tan1}, \eqref{proj_tan2}, and apply Moreau's decomposition to get $\Pi_{\cT_K(\vb)}(\cdot)$. For $\Pi_{-\cT_K(\vb)}(\cdot)$ use \eqref{neg:cone:calc}\\%, note that $-\cT_K(\vb)=\cT_K(-\vb)$ since $-K = K$.\\
    The debiased $j$\textsuperscript{th} coefficient $\hat\bbeta_{d}^{(j)} \leftarrow \vb^{(j)} + n^{-1}\hat\bmeta\T \tilde\Xb\T (\tilde\bY - \tilde\Xb\vb)$.
\end{enumerate}
\end{algorithmic}
\end{algorithm}

%%%%%%%
\subsection{Debiasing Algorithm for Square-Root SLOPE}
To solve the square-root SLOPE, the joint optimization \eqref{sqrt_slope_opt} can be solved by alternatively minimizing $\bbeta$ and $\sigma$: the minimization in $\bbeta$ is the same as SLOPE in \eqref{slope_opt} with parameters $\hat \sigma\lambda_1,\ldots,\hat \sigma\lambda_p$, and after that setting $\hat \sigma$ to $\hat \sigma = \|\overline\bY-\overline{\Xb}\bbeta\|/\sqrt{n}$. Details can be found in \citep[Algorithm 1]{stucky2017sharp} and \citep[Algorithm 2]{derumigny2018improved}. The debiasing algorithm for square-root SLOPE is the same as Algorithm \ref{algo_slope}.

%%%%%%%%%%%%
%% Non-Gaussian Error
%%%%%%%%%%%%
\section{Non-Gaussian Errors}\label{subGaussian:noise:section}
In this section we modify our Algorithm \ref{algo_unknowncov} to accommodate for sub-Gaussian noise. The modified procedure is presented in Algorithm \ref{algo_unknowncov:subgaussian:noise}. Algorithm \ref{algo_unknowncov:subgaussian:noise} requires an additional condition in step 1, namely $\|\vb-\bbeta^*\|\sqrt{\log n}=o_p(1)$. We view this as a fairly mild assumption, which in most relevant practical cases is dominated by the assumption $\|\vb-\bbeta^*\|\overline w(\cT_K(\vb) \cap \mathbb{S}^{p-1}) = o_p(1)$. In step 2 of Algorithm \ref{algo_unknowncov:subgaussian:noise}, we have added an additional $\ell_{\infty}$ constraint to the optimization. Observe that the modified program in step 2 is still a convex program, and can be solved by subgradient descent as before. 

\begin{algorithm}[ht]
\caption{Debias the $j$\textsuperscript{th} Coordinate of A Non-Ordinary Least Squares Estimator}
\label{algo_unknowncov:subgaussian:noise}
\begin{algorithmic}
\STATE\textbf{Input:} Two equal size partitions $(\overline\Xb, \overline\bY)$ and $(\tilde\Xb, \tilde\bY)$, $\hat\bbeta$ obtained using $(\overline\Xb, \overline\bY)$.\\
\STATE\textbf{Initialize:} Empirical Gram matrix of the second partition $\hat\bSigma=\frac{1}{n}\tilde\Xb\T\tilde\Xb$.
\begin{enumerate}
    \item 
    Using the first data split, find a convex set $K$ and a vector $\vb$, such that: $\vb, \bbeta^* \in K$ with high probability, and $ \|\vb-\bbeta^*\|\max\{\overline w(\cT_K(\vb)\cap \mathbb{S}^{p-1}), \sqrt{\log n}\}=o_p(1)$.
    \item \label{step2_non_gaussian_noise}
    The debiased $j$\textsuperscript{th} coefficient $\hat\bbeta_{d}^{(j)} \leftarrow \eb^{(j)\top}\vb + n^{-1}\hat\bmeta\T \tilde\Xb\T (\tilde\bY - \tilde\Xb\vb)$, where $\hat\bmeta$ is computed by 
    \begin{align}\hat\bmeta \leftarrow \argmin_{\bmeta}\, \|\hat\bSigma^{\frac{1}{2}}\bmeta\| \mbox{   subject to  } \label{optimization:algo6}\\
     \displaystyle\sup_{\ub \in \cT_K(\vb) \cap \mathbb{S}^{p-1}} |(\bmeta\T \hat \bSigma - \eb^{(j)\top}) \ub|  \leq \rho \frac{\overline w(\cT_K(\vb)\cap\mathbb{S}^{p-1})}{\sqrt{n}},\nonumber\\
     \|\tilde\Xb\bmeta\|_{\infty} \leq \rho' \sqrt{\log n},\nonumber%\frac
    \end{align}
    for some sufficiently large tuning parameters $\rho>0, \rho'>0$.
\end{enumerate}
\end{algorithmic}
\end{algorithm}

To show that the new optimization program has a feasible point and consequently a non-empty interior, we evaluate the constraints at the point $\bmeta = \bSigma^{-1}\eb^{(j)}$. By using a similar argument to that of \cite[p. 33]{javanmard2014confidence} we are able to show that $\|\tilde\Xb\bSigma^{-1}\eb^{(j)}\|_{\infty}\lesssim\sqrt{\log n}$, and the argument of non-empty interior is similar to how we prove Lemma \ref{nonempty_int_Q}. Details are given in Lemma \ref{non_empty_int_nonGaussian} and its proof.
\begin{lemma}
\label{non_empty_int_nonGaussian}
Suppose that $\tilde\Xb=(\tilde\bX_1,\ldots,\tilde\bX_n)\T$ where every observation $\tilde\bX_i$ is a zero-mean bounded or a zero-mean Gaussian random variable with covariance matrix $\bSigma$, and the eigenvalues of $\bSigma$ are bounded from above and below. For a sufficiently large constant $\rho'>0$, the set 
\begin{align*}
\{\bmeta:\,\,\|\tilde\Xb\bmeta\|_{\infty} \leq \rho' \sqrt{\log n}\} \cap \bigg\{\bmeta:\,\, \displaystyle\sup_{\ub \in \cT_K(\vb) \cap \mathbb{S}^{p-1}} |(\bmeta\T \hat \bSigma - \eb^{(j)\top}) \ub|  \leq \rho \frac{\overline w(\cT_K(\vb)\cap\mathbb{S}^{p-1})}{\sqrt{n}}\bigg\},
\end{align*}
has a non-empty interior.
\end{lemma}

Solving the optimization in step 2 of Algorithm \ref{algo_unknowncov:subgaussian:noise} is similar to solving the optimization in step 2 of Algorithm \ref{algo_unknowncov}, since both of them are convex programs with inequality constraints. The only difference is that the former has two constraints while the latter has only one. According to \cite[Section 7]{boyd2003subgradient}, the idea of solving optimization with multiple inequality constraints is: if the current point is feasible, subgradient descent is applied to the objective function; if the current point is not feasible, we pick any one of the violated constraints, and apply subgradient descent to it. Define
\begin{align*}
    \psi'(\bmeta) = \|\tilde\Xb\bmeta\|_{\infty}-\rho'\sqrt{\log n}.
\end{align*}
The second constraint in step 2 of Algorithm \ref{algo_unknowncov:subgaussian:noise} can be written as $Q'=\{\bmeta: \psi'(\bmeta)\leq0\}$. To this end we remind the reader of the shorthand notations $\psi(\bmeta)$ from \eqref{psi:eta:Q:notation}, and $Q = \{\bmeta:\, \psi(\bmeta)\leq0\}$. The sequence $\{\bmeta_n\}$ is generated as in \eqref{sg_step}, where $\gb_n$ is the gradient of the objective function if $\bmeta_n\in Q$ and $\bmeta_n\in Q'$; is a subgradient of $\psi(\bmeta_n)$ if $\bmeta_n\notin Q$; otherwise is a subgradient of $\psi'(\bmeta_n)$ if $\bmeta_n\in Q$ and $\bmeta_n\notin Q'$. In the following Lemma \ref{sg_psi_prime} we give the expression of a subgradient of $\psi'(\bmeta_n)$.
\begin{lemma}
\label{sg_psi_prime}
Let $i^*=\argmax\limits_{i\in[n]}|\tilde\bX_i\T\bmeta_n|$. Then $\nabla\psi'(\bmeta_n) = \sgn(\tilde\bX_{i^*}\T\bmeta_n)\tilde\bX_{i^*}$ is a subgradient of $\psi'(\bmeta_n)$.
\end{lemma}
After adding the new constraint $\psi'(\bmeta)\leq0$, Algorithm \ref{algo_step3} is modified to Algorithm \ref{algo_step3_nonGaussian_noise}. In terms of the convergence of Algorithm \ref{algo_step3_nonGaussian_noise}, it also takes $n=O(1/\epsilon^2)$ iterations to get an $\epsilon$-suboptimal solution i.e. $\|\hat \bSigma^{\frac{1}{2}}\bmeta_n\|-\|\hat \bSigma^{\frac{1}{2}}\bmeta^*\|\leq\epsilon$. The proof of Lemma \ref{cvg_sbgrad} will remain unchanged since $\psi'(\bmeta_n)$ is a Lipschitz function of $\bmeta_n$ (since with probability $1$, $\sup_{i \in [n]} \|\tilde \bX_i\| < \infty$).
\begin{algorithm}%[H]
\caption{Solve the Optimization in Step 2 of Algorithm \ref{algo_unknowncov:subgaussian:noise}}
\label{algo_step3_nonGaussian_noise}
\begin{algorithmic}
\STATE \textbf{Input:} The convex set $K$, the vector $\vb$ from step 2, empirical Gram matrix of the second partition $\hat\bSigma=\frac{1}{n}\tilde\Xb\T\tilde\Xb$.\\
\STATE \textbf{Initialize:} $\bmeta_1$\\% $\bmeta_{out}$. Make sure that $\bmeta_{out}$ is feasible.\\
\STATE Run for sufficiently long time:\\
\STATE \hspace*{\algorithmicindent} Compute $P_{+} \leftarrow\Pi_{\cT_{K}(\vb)}(\hat\bSigma\bmeta_{n} - \eb^{(j)})$, $P_{-} \leftarrow\Pi_{-\cT_{K}(\vb)}(\hat\bSigma\bmeta_{n} - \eb^{(j)})$.\\
\STATE \hspace*{\algorithmicindent} \algorithmicif{ $ \max\{\|P_+\|, \|P_-\|\}\leq \frac{\rho\overline w(\cT_K(\vb)\cap\mathbb{S}^{p-1})}{\sqrt{n}}$ \& $\|\Xb \bmeta_n\|_{\infty}\leq \rho'\sqrt{\log n}$}  \\
\STATE \hspace*{\algorithmicindent}\hspace*{\algorithmicindent} \algorithmicif{\,\,\ $\|\hat\bSigma^{\frac{1}{2}}\bmeta_{n}\| \leq \|\hat\bSigma^{\frac{1}{2}}\bmeta_{out}\|$:} $\bmeta_{out}\leftarrow \bmeta_{n}$\\
\STATE \hspace*{\algorithmicindent}\hspace*{\algorithmicindent} $\bmeta_{n+1} \leftarrow \bmeta_{n} - h_n\frac{\hat\bSigma \bmeta_{n}}{\|\hat\bSigma^{\frac{1}{2}} \bmeta_{n}\|}$\\
% Else
\STATE \hspace*{\algorithmicindent} \algorithmicelse \algorithmicif{ $ \max\{\|P_+\|, \|P_-\|\}> \frac{\rho\overline w(\cT_K(\vb)\cap\mathbb{S}^{p-1})}{\sqrt{n}}$}:\\
\STATE \hspace*{\algorithmicindent}\hspace*{\algorithmicindent} $\phi_0(\bmeta_{n}) \leftarrow P_+ \,/\, \|P_+\|$\\
\STATE \hspace*{\algorithmicindent}\hspace*{\algorithmicindent} $\phi_1(\bmeta_{n}) \leftarrow P_- \,/\, \|P_-\|$.\\
\STATE \hspace*{\algorithmicindent}\hspace*{\algorithmicindent} $\bmeta_{n+1} \leftarrow \bmeta_{n} - h_n \hat\bSigma\phi_{\one\{(\bmeta_{n}\T \hat \bSigma - \eb^{(j)\top}) (\phi_0(\bmeta_{n-1}) - \phi_1(\bmeta_{n-1})) < 0\}}(\bmeta_{n})$\\
\STATE \hspace*{\algorithmicindent} \algorithmicelse :\\
\STATE \hspace*{\algorithmicindent}\hspace*{\algorithmicindent} $\bmeta_{n+1} \leftarrow \bmeta_{n} - h_n \sign(\tilde\bX_{i^*}\T\bmeta_n)\tilde\bX_{i^*}$, where $i^*=\argmax\limits_{i\in[n]}|\tilde\bX_i\T\bmeta_n|$\\
\STATE $\hat\bmeta \leftarrow \bmeta_{out}$.
\end{algorithmic}
\end{algorithm} 

We now state a result which establishes the confidence interval for non-Gaussian noise.

\begin{theorem}
\label{non_gaussian}
Consider a linear model in \eqref{gaussian_model} and with sub-Gaussian errors $\varepsilon_i$. Suppose the eigenvalues of $\bSigma$ are bounded from both above and below. Recall that $\hat\bbeta_{d}^{(j)}$ is the debiased $j$\textsuperscript{th} coefficient obtained by Algorithm \ref{algo_unknowncov:subgaussian:noise}. Let $a_n = o(1)$ be any slowly converging to $0$ rate such that $\frac{1}{a_n}=o(\frac{n}{\log n})$, and let $c$ be sufficiently large constant satisfying $c >C'\frac{\sqrt{\log n}/\sqrt{(\|\bbeta^* - \vb\|\sqrt{\log n}) \vee a_n}}{\sqrt{n}} = o_p(1)$, where $C'$ is a universal constant. Then the confidence interval
\begin{align}
\label{ci_non_gaussian_noise}
    \bigg( \hat\bbeta_d^{(j)}-z_{\frac{\alpha}{2}}\frac{\sigma(\|\hat\bSigma^{\frac{1}{2}}\hat\bmeta\|\vee c)}{\sqrt{n}}, \hat\bbeta_d^{(j)}+z_{\frac{\alpha}{2}}\frac{\sigma(\|\hat\bSigma^{\frac{1}{2}}\hat\bmeta\|\vee c)}{\sqrt{n}} \bigg),
\end{align}
contains $\bbeta^*$ with probability at least $1-\alpha$ asymptotically.
\end{theorem}
It is worthwhile to mention that even though the length of the confidence interval \eqref{ci_non_gaussian_noise} is always of the order $O(1/\sqrt{n})$, when the quantity $\|\hat\bSigma^{\frac{1}{2}}\hat\bmeta\|$ is very small such that
\begin{align*}
    \|\hat\bSigma^{\frac{1}{2}}\hat\bmeta\| \lesssim \frac{\sqrt{\log n}/\sqrt{(\|\bbeta^* - \vb\|\sqrt{\log n}) \vee a_n}}{\sqrt{n}},
\end{align*}
as can be seen from our proof, the debiased estimator $\hat \bbeta_d^{j}$ actually converges faster than the rate $1/\sqrt{n}$. In this case the confidence interval \eqref{ci_non_gaussian_noise} is still valid, but not very efficient. And contrarily if 
\begin{align*}
    \|\hat\bSigma^{\frac{1}{2}}\hat\bmeta\| \gtrsim \frac{\sqrt{\log n}/\sqrt{(\|\bbeta^* - \vb\|\sqrt{\log n}) \vee a_n}}{\sqrt{n}},
\end{align*}
then a Central Limit Theorem applies to $\sqrt{n}(\hat\bbeta_d^{(j)} - \bbeta^{*(j)})$, and the variance would be exactly $\sigma\|\hat\bSigma^{\frac{1}{2}}\hat\bmeta\|$. Thus the confidence interval is tight when $\|\hat\bSigma^{\frac{1}{2}}\hat\bmeta\|\geq c$, and is slightly loose when $\|\hat\bSigma^{\frac{1}{2}}\hat\bmeta\|< c$ since we are using a slightly larger variance.

Finally, we can also consistently estimate $\sigma$ as in Theorem \ref{sigma_hat_rate} whose proof does not rely on the Gaussian assumption on the noise.

%%%%%%%%%%%%%%
%%%  SIMULATIONS
%%%%%%%%%%%%%
\section{Simulations}\label{simulations:sec}
Now we examine the performance of the proposed debiasing procedure for the monotone cone regression, positive monotone cone regression, non-negative least squares, constrained LASSO, SLOPE and square-root SLOPE cases. We pick a single coordinate to debias. In all the experiments of this section, the last coordinate of the signal vector is picked. 

In terms of the construction of true coefficient $\bbeta^*$, for the monotone cone case, $\bbeta^*$ consists of -1 and 1, where the first 70\% coordinates are -1, and the remaining 30\% are 1. For the positive monotone cone case, the true coefficient $\bbeta^*$ consists of 0 and 1, where the first 70\% coordinates are 0, and the remaining 30\% are 1. For the non-negative least squares case (see Appendix \ref{non:negative:least:squares} of the supplement file for details), we generate $\bbeta^*$ such that each coordinate is $\max\{N(0,3), 0\}$. For the LASSO case, $\bbeta^*$ consists of 0 and 1, where the first 99.5\% of the coordinates are 0, and the remaining 0.5\% are 1. For the SLOPE and square-root SLOPE cases, the first 99.5\% of the true $\bbeta^*$ are 0, the remaining coordinates are formed by an increasing series of integers with step size $1$ starting from $1$. 
In terms of the sample size $n$ and dimension $p$, we use $n=100, p=100$ for the monotone cone and positive monotone cone cases. Note that this conforms to our assumption that $w^2(\cT_K(\bbeta^*)\cap\mathbb{S}^{p-1}) = o(\sqrt{n})$ since the vector $\bbeta^*$ is comprised only of 2 constant pieces. For the non-negative least squares case, we pick $n=1000,\,p=50$ in order to make $w^2(\cT_K(\bbeta^*)\cap\mathbb{S}^{p-1})\asymp p$ (see Lemma \ref{width:bound:nonneg:orthant}) approximately comparable to $\sqrt{n}$. For LASSO, SLOPE and square-root SLOPE, we use $n=1000,\, p=1000$. Coupled with the small proportion of non-zero coordinates in $\bbeta^*$ this guarantees that $w^2(\cT_K(\bbeta^*)\cap\mathbb{S}^{p-1})\asymp{s\log(ep/s)}$ is smaller than $\sqrt{n}$, where $s$ denotes the sparsity of $\bbeta^*$.

The predictors $\Xb$ are drawn from a mean-zero Gaussian distribution. In order to verify the compatibility of this debiasing procedure with different types of input data, three different covariance matrices $\bSigma$ are used to generate different Gaussian distributions: an identity matrix, a random matrix with bounded eigenvalues, and a Toeplitz matrix whose $i,j$-th element is $\rho^{|i-j|}$ where $\rho\in(0, 1)$ (we use $\rho = 0.4$).

For each type of the predictor and covariance matrix $\bSigma$, we generate the data $\Xb$, $\bY$, $\bbeta^*$, we obtain the original estimator $\hat\bbeta$, and perform Algorithm \ref{algo_unknowncov} to debias the last coordinate. The experiment is repeated $100$ times. According to Theorem \ref{debiase_formula_applicable_unkowncov}, for any coordinate $j$, the debiased estimator $\hat\bbeta_d^{(j)}$ should satisfy $\sqrt{n}(\hat\bbeta_d^{(j)}-\bbeta^{*(j)})\sim N(0, \sigma^2\hat\bmeta\T\hat\bSigma\hat\bmeta)$, which doesn't necessarily hold for the non-debiased estimator $\hat\bbeta^{(j)}$. In Figure \ref{fig:qq_identity}, we examine the distribution of $\frac{\sqrt{n}(\hat\bbeta_d^{(j)}-\bbeta^{*(j)})}{\hat\sigma\|\hat\bSigma\hat\bmeta\|}$ and $\frac{\hat\bbeta^{(j)}-\bbeta^{*(j)}}{sd(\hat\bbeta^{(j)}-\bbeta^{*(j)})}$ for $j=p$, by plotting them against the standard Gaussian distribution in a Q-Q plot. We can see from those plots that $\frac{\sqrt{n}(\hat\bbeta_d^{(j)}-\bbeta^{*(j)})}{\hat\sigma\|\hat\bSigma\hat\bmeta\|}$ appears pretty close to $N(0,1)$, which is not true for $\frac{\hat\bbeta^{(j)}-\bbeta^{*(j)}}{sd(\hat\bbeta^{(j)}-\bbeta^{*(j)})}$ in terms of both bias and variance difference. It is worth pointing out that for the SLOPE and square-root SLOPE cases, although the undebiased estimators points appear to align well on the Q-Q plot they are not centered at the correct value. Figure \ref{fig:qq_identity} only reports the results in the setting $\bSigma=\Ib$. Similar plots for the bounded eigenvalue and Toeplitz population covariance matrix settings are attached in the Supplement-A file for conciseness.

\begin{figure}%[H]
\makebox[\textwidth]{\includegraphics[width=1.3\textwidth]{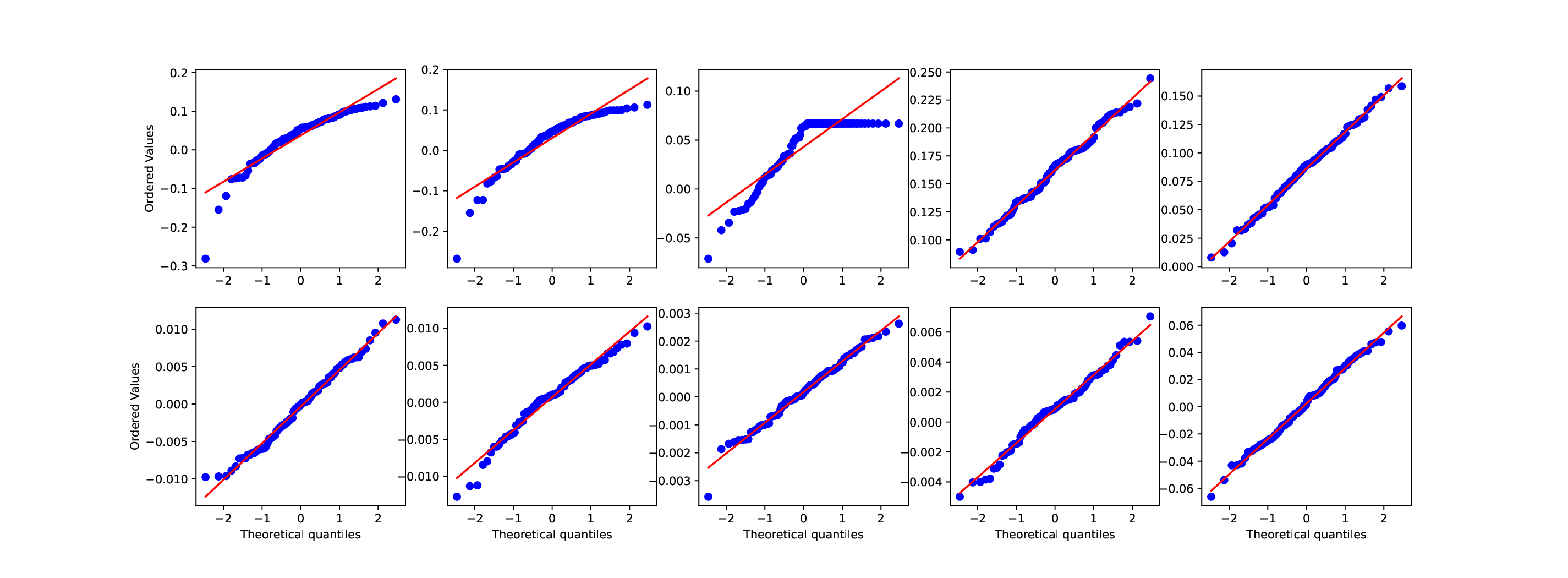}}
\caption{ The Q-Q Plot of $\frac{\sqrt{n}(\hat\bbeta_d^{(j)}-\bbeta^{*(j)})}{\hat\sigma\|\hat\bSigma\hat\bmeta\|}$ and $\frac{\hat\bbeta^{(j)}-\bbeta^{*(j)}}{sd(\hat\bbeta^{(j)}-\bbeta^{*(j)})}$ where $j=p$, against Standard Normal in the Identity Population Matrix Setting. {\footnotesize \textit{The Upper Row: $\frac{\hat\bbeta^{(j)}-\bbeta^{*(j)}}{sd(\hat\bbeta^{(j)}-\bbeta^{*(j)})}$ the scaled Difference between the Undebiased Estimator and the True Coefficient; the Lower Row: $\frac{\sqrt{n}(\hat\bbeta_d^{(j)}-\bbeta^{*(j)})}{\hat\sigma\|\hat\bSigma\hat\bmeta\|}$ the scaled Difference between the Debiased Estimator and the True Coefficient. From Left to Right: Monotone Cone Regression, Positive Monotone Cone Regression, Non-negative Regression, LASSO, SLOPE, Square-root SLOPE.} }}
\label{fig:qq_identity}
\end{figure}

%%%%%%%%%%%%
%% Future Work
%%%%%%%%%%%%
\section{Future Work}\label{discussion:section}

In this paper we proposed a novel abstract procedure for debiasing linear regressions. Our method is able to perform inference for some constrained and regularized problems for which inferential tools were not previously available. 

An interesting further question to explore is whether we can prove lower bounds on confidence intervals obtained in the above way such as the work of \cite{cai2017confidence}. In other words are the conditions $w^2(\cT_K(\vb') \cap \mathbb{S}^{p-1}) = o(\sqrt{n})$ and $\|\vb' -\bbeta^*\|= o(1/\sqrt{n})$ also necessary for the unknown covariance case?

Another open question is debiasing the constrained least squares using \eqref{debiase_formula_split:main:text} in the unknown covariance case but without resorting to sample splitting. Our conjecture is that sample splitting is not required, but a proof of this fact will require carefully isolating the dependency of $\hat \bbeta$ on $\Xb$. For this purpose, it may be necessary to employ a slightly different debiasing scheme as the one undertook by \cite{bellec2019biasing}. 

Furthermore, the question of how can one solve the second optimization program if projecting on $\cT_K(\vb)$ is hard is also interesting. In particular we are curious whether it is possible to apply interior point methods.

Finally, our main procedure requires us to split the data. Inevitably, this results in a loss of efficiency. One way to correct for that is to use a cross-fitted estimator as in \cite{chernozhukov2018double,eftekhari2021inference}. It is unclear to us at the moment whether this strategy will work in our case as the influence functions of the estimators on the two samples may not be independent.

\bibliographystyle{apalike}
\bibliography{yufeibib}       % Bibliography file (usually '*.bib')

%%%%%%%%%%%%%%%%%%%%%%%%
%
% SUPPLEMENT A
%
%%%%%%%%%%%%%%%%%%%%%%%%
\newpage

% Page layouts for supplements
%-------------------------------------
\newenvironment{changemargin}[2]{%
\begin{list}{}{%
\setlength{\topsep}{0pt}%
\setlength{\leftmargin}{#1}%
\setlength{\rightmargin}{#2}%
\setlength{\listparindent}{\parindent}%
\setlength{\itemindent}{\parindent}%
\setlength{\parsep}{\parskip}%
}%
\item[]}{\end{list}}

\begin{changemargin}{-1cm}{-1.1cm}
%------------------------------------

\setcounter{page}{1}

\title{Supplement-A to ``A New Perspective on Debiasing Linear Regressions'' }
\maketitle

\renewcommand{\thesection}{\Alph{section}}

% %%%%%%%%%%%%
% %% Exp
% %%%%%%%%%%%
\section*{Additional Simulation Results}
All the code for experiments can be found in: \url{https://github.com/Pythongoras/debiascvgV2}.

\begin{figure}[H]
\makebox[\textwidth]{\includegraphics[width=1.3\textwidth]{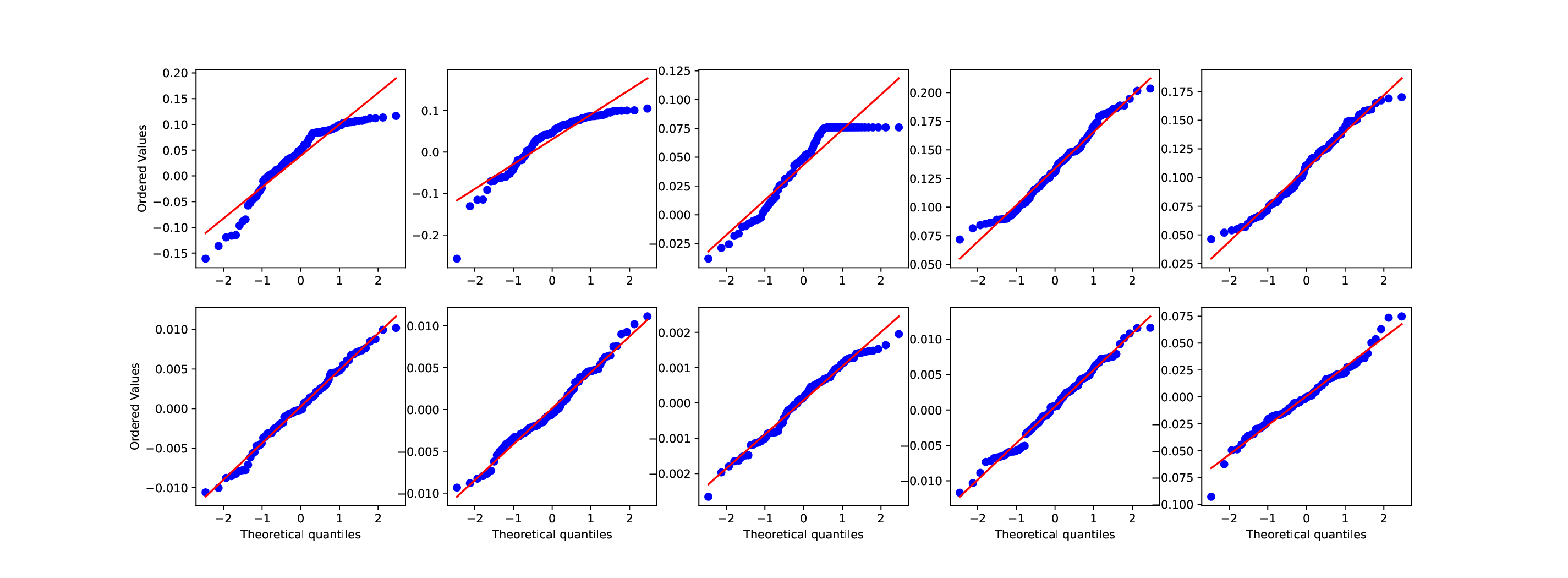}}
\caption{ The Q-Q Plot of $\frac{\sqrt{n}(\hat\bbeta_d^{(j)}-\bbeta^{*(j)})}{\hat\sigma\|\hat\bSigma\hat\bmeta\|}$ and $\frac{\hat\bbeta^{(j)}-\bbeta^{*(j)}}{sd(\hat\bbeta^{(j)}-\bbeta^{*(j)})}$ where $j=p$, against Standard Normal in the Bounded Eigenvalue Population Matrix Setting. {\footnotesize \textit{The Upper Row: $\frac{\hat\bbeta^{(j)}-\bbeta^{*(j)}}{sd(\hat\bbeta^{(j)}-\bbeta^{*(j)})}$ the scaled Difference between the Undebiased Estimator and the True Coefficient; the Lower Row: $\frac{\sqrt{n}(\hat\bbeta_d^{(j)}-\bbeta^{*(j)})}{\hat\sigma\|\hat\bSigma\hat\bmeta\|}$ the scaled Difference between the Debiased Estimator and the True Coefficient. From Left to Right: Monotone Cone Regression, Positive Monotone Cone Regression, Non-negative Regression, LASSO, SLOPE, Square-root SLOPE.} }}
\label{fig:qq_bddeig}
\end{figure}

\begin{figure}[H]
\makebox[\textwidth]{\includegraphics[width=1.3\textwidth]{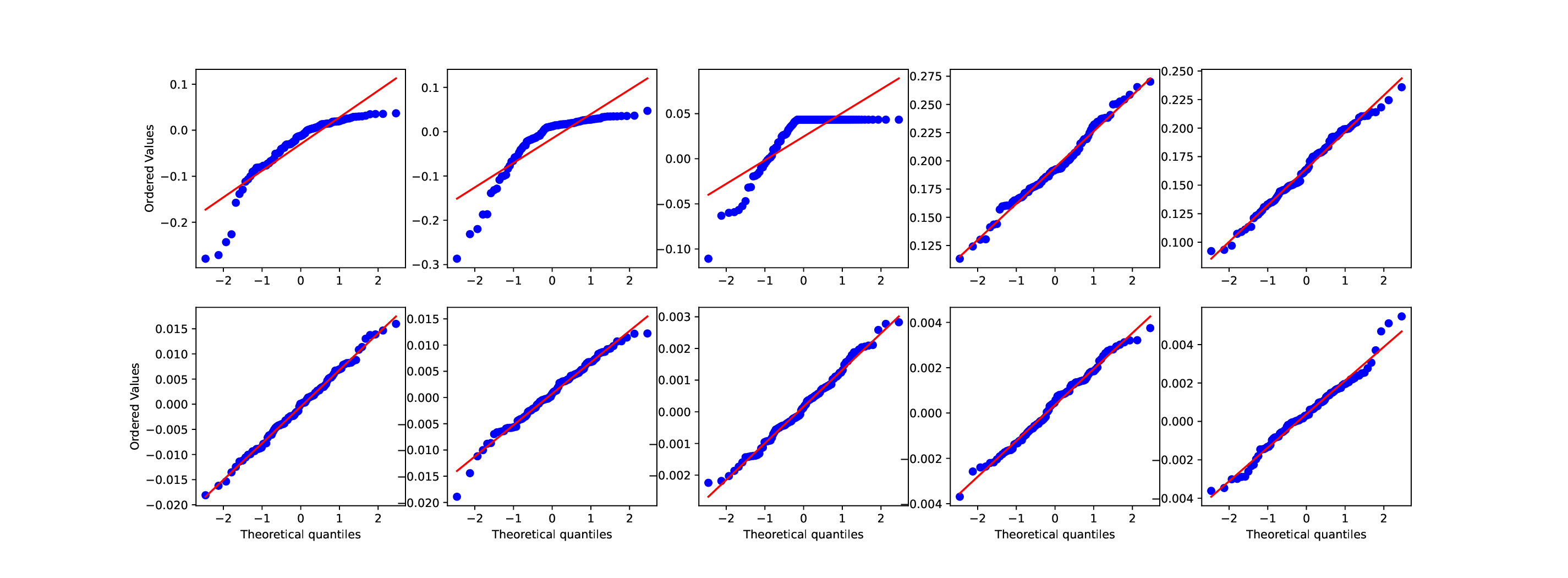}}
\caption{ The Q-Q Plot of $\frac{\sqrt{n}(\hat\bbeta_d^{(j)}-\bbeta^{*(j)})}{\hat\sigma\|\hat\bSigma\hat\bmeta\|}$ and $\frac{\hat\bbeta^{(j)}-\bbeta^{*(j)}}{sd(\hat\bbeta^{(j)}-\bbeta^{*(j)})}$ where $j=p$, against Standard Normal in the Toeplitz Population Matrix Setting. {\footnotesize \textit{The Upper Row: $\frac{\hat\bbeta^{(j)}-\bbeta^{*(j)}}{sd(\hat\bbeta^{(j)}-\bbeta^{*(j)})}$ the scaled Difference between the Undebiased Estimator and the True Coefficient; the Lower Row: $\frac{\sqrt{n}(\hat\bbeta_d^{(j)}-\bbeta^{*(j)})}{\hat\sigma\|\hat\bSigma\hat\bmeta\|}$ the scaled Difference between the Debiased Estimator and the True Coefficient. From Left to Right: Monotone Cone Regression, Positive Monotone Cone Regression, Non-negative Regression, LASSO, SLOPE, Square-root SLOPE.} }}
\label{fig:qq_ar1}
\end{figure}

%%%%%%%%%%%%%%%%%%%%%%%%
%
% SUPPLEMENT B
%
%%%%%%%%%%%%%%%%%%%%%%%%
\newpage
\title{Supplement B to ``A New Perspective on Debiasing Linear Regressions'' }
\maketitle
\setcounter{page}{1}
% %%%%%%%%%%%%
% %% Exp
% %%%%%%%%%%%
\section{Non-negative Least Squares}\label{non:negative:least:squares}

In this section we suppose that $K = \{\bbeta: \bbeta^{(i)} \geq 0 ~~ \forall i \in [p]\}$ is the non-negative orthant cone. Clearly, implementing the non-negative least squares can be done via a quadratic program, or with a projected gradient descent, where the projection onto the non-negative orthant is given by setting to $0$ any negative coefficients.

In order to implement \eqref{step2_cvs_ls} and find $\vb$ in step \ref{step_1}, we need to evaluate an upper bound on the Gaussian complexity of $\cT_K(\vb) \cap \mathbb{S}^{p-1}$ for any $\vb \in K$; see Lemma \ref{width:bound:nonneg:orthant}. 

\begin{lemma}\label{width:bound:nonneg:orthant}
If $K = \{\bbeta: \bbeta^{(i)} \geq 0 ~~ \forall i \in [p]\}$ is the non-negative orthant cone, for any $\vb\in K$ the following bound holds
\begin{align*}
    w(\cT_K(\vb) \cap \mathbb{S}^{p-1}) \leq \sqrt{p - |\{i : \vb^{(i)} = 0\}|/2}.
\end{align*}
\end{lemma}
Then as in the monotone cone case, the problem \eqref{step2_cvs_ls} boils down to an optimization over finitely many candidates. Let $\vb_s$ be the projection of $\hat \bbeta$ onto the set of non-negative vectors with exactly $s$ zero coefficients. We then need to solve
\begin{align*}
    \hat s = \argmin_{s \in [0, p]} \|\hat \bbeta - \vb_s\| + \sqrt{\frac{p - s/2}{n}},
\end{align*}
and our final solution is $\vb = \vb_{\hat s}$. What is left to show is how to obtain a vector $\vb_s$, which is discussed in Lemma \ref{obvious:lemma:nonnegative:cone}. 
\begin{lemma}\label{obvious:lemma:nonnegative:cone} Let $S$ denote the index set of the $s$ smallest in magnitude coefficients of $\hat \bbeta$. The vector $\vb_s$ is given by
\begin{align*}
    \vb_s^{(i)} = \hat \bbeta^{(i)}\mathbbm{1}(i \in S^c).
\end{align*}
In other words $\vb_s$ greedily takes the largest $p-s$ entries in $\hat \bbeta$, where ties are broken arbitrarily. 
\end{lemma}

After we obtain $\vb$ in step \ref{step_1}, we also need to write down the explicit form of the projection $\Pi_{\cT_K(\vb)}$ to solve step \ref{step_2}. Such a projection is provided in Lemma \ref{projection:positive:cone}.

\begin{lemma}\label{projection:positive:cone} We have that 
\begin{align*}
    \Pi_{\cT_K(\vb)}(\xb) = \big(\xb^{(i)}\mathbbm{1}(\vb^{(i)} \neq 0) + (\xb^{(i)})_+\mathbbm{1}(\vb^{(i)} = 0)\big)_{i \in [p]}
\end{align*}
\end{lemma}

We summarize the procedure in Algorithm \ref{algo_nnls}.

\begin{algorithm}[H]
\caption{Debias the $j$\textsuperscript{th} Coefficient for Non-negative Least Squares}
\label{algo_nnls}
\begin{algorithmic}
\STATE \textbf{Input:} Two equal size partitions $(\overline\Xb, \overline\bY)$ and $(\tilde\Xb, \tilde\bY)$, $\hat\bbeta$ obtained by projected gradient descent with isotonic regression.\\
\STATE \textbf{Initialize:} Empirical Gram matrix of the second partition $\hat\bSigma=\frac{1}{n}\tilde\Xb\T\tilde\Xb$.
\begin{enumerate}
    \item Solve $\hat s = \argmin_{s \in [0, p]} \|\hat \bbeta - \vb_s\| + \sqrt{\frac{p - s/2}{n}}$.\\
    $\vb \leftarrow \vb_{\hat s}$.
    \item Run Algorithm \ref{algo_step3}. Compute $\Pi_{\cT_{K}(\vb)}(\cdot)$ by the result of Lemma \ref{projection:positive:cone}. For $\Pi_{-\cT_{K}(\vb)}(\cdot)$ use \eqref{neg:cone:calc}. %the negation of an increasing monotone cone is a decreasing monotone cone.\\
    The debiased $j$\textsuperscript{th} coefficient equals $\hat\bbeta_{d}^{(j)} \leftarrow \vb^{(j)} + n^{-1}\hat\bmeta\T \tilde\Xb\T (\tilde\bY - \tilde\Xb\vb)$.
\end{enumerate}
\end{algorithmic}
\end{algorithm}

\subsection{Proofs}

%%%%%%%%%%%%%%
%%
%%%%%%%%%%%%%%
\begin{proof}[Proof of Lemma \ref{projection:positive:cone}]
The proof is an elementary calculation and is omitted.
\end{proof}

\begin{proof}[Proof of Lemma \ref{width:bound:nonneg:orthant}] We first note that by Cauchy-Schwartz the Gaussian complexity is upper bounded by the statistical dimension, i.e.,

\begin{align*}
    \EE_{\gb \sim N(0, \mathbf{I})} \sup_{\xb \in \cT_K(\vb) \cap \mathbb{S}^{p-1}}\langle \gb, \xb\rangle = \EE_{\gb \sim N(0, \mathbf{I})} \|\Pi_{\cT_K(\vb)}(\gb)\| \leq \sqrt{ \EE_{\gb \sim N(0, \mathbf{I})} \|\Pi_{\cT_K(\vb)}(\gb)\|^2}
\end{align*}

Now by Lemma \eqref{projection:positive:cone} the projection 
\begin{align*}
   \EE \|\Pi_{\cT_K(\vb)}(\gb)\|^2 = \EE \sum_{i: \vb^{(i)} \neq 0} \gb^{(i)2} + \EE \sum_{i: \vb^{(i)} = 0} \gb^{(i)2}_+ = p - s/2,
\end{align*}
where $s = |\{i : \vb^{(i)} = 0\}|$.
\end{proof}

%%%%%%%%%%%%
%%
%%%%%%%%%%%%\section{}
\begin{proof}[Proof of Lemma \ref{obvious:lemma:nonnegative:cone}]This statement is obvious and we omit the details. 
\end{proof}

\section{ Lower Bounds on Confidence Interval Length in Convex Constrained Least Squares}\label{lower:bounds:appendix}

We will now show that under certain conditions the $\frac{1}{\sqrt{n}}$-rate of the confidence intervals that we provide cannot be improved in a worst case sense. Of course one should not expect this is always the case. For example, if the set $K$ is a set of diameter $\ll \frac{1}{\sqrt{n}}$ the practitioner does not even need to debias their coefficients to be able to construct faster than $\frac{1}{\sqrt{n}}$ confidence intervals. In order to construct this lower bound we follow \cite{cai2017confidence} who proved lower bounds on the length of the confidence intervals under a sparse parameter space. We modify their argument and add an additional assumption in order to allow for the restriction $\bbeta^* \in K$. The lower bound is derived under the assumption that the design matrix $\bX_i\sim N(0, \bSigma)$ and the noise $\varepsilon_i\sim N(0, \sigma^2)$. %In Lemma \ref{minimax_length_ci} below we show that the length for any confidence interval on a single coordinate $\bbeta^{(j)}$ given by any algorithm that simultaneously provides valid intervals over a set of vectors $\bbeta$ at least as large as the ones our algorithm works on, cannot be more efficient than $O(\frac{1}{\sqrt{n}})$ (for more precise definition see the definition of the set $\cH$ in Lemma \ref{minimax_length_ci} below).

Before we introduce Lemma \ref{minimax_length_ci}, we need to introduce several definitions regarding the construction of confidence intervals. First we denote with $CI_{\alpha}(\eb^{(j)\top}\bbeta, \Xb, \bY)$ a $(1-\alpha)$-level confidence interval on $\eb^{(j)\top}\bbeta$ with data $(\Xb, \bY)$, and let $L(CI_{\alpha}(\eb^{(j)\top}\bbeta, \Xb, \bY))$ be its length. We write $\eb^{(j)\top}\beta$ to stress the fact that $\eb^{(j)}$ can be substituted with any vector $\bgamma$ with $\|\bgamma\| \leq B < \infty$. Then we define the set of all $(1-\alpha)$-level confidence intervals for $\eb^{(j)\top}\bbeta$ over $\bbeta\in\Theta$ as
\begin{align*}
    \mathcal{I}_{\alpha}(\Theta) &= \Big\{ CI_{\alpha}(\eb^{(j)\top}\bbeta, \Xb, \bY)=[l(\Xb,\bY),u(\Xb,\bY)]: \\
    &\inf_{\bbeta\in\Theta} \mathbb{P}_{\bbeta}\big(l(\Xb,\bY)\leq\eb^{(j)\top}\bbeta\leq u(\Xb,\bY)\big) \geq 1-\alpha \Big\}.
\end{align*}
Finally, define the worst case expected confidence interval length over $\Theta$:
\begin{align*}
    L(CI_{\alpha}(\eb^{(j)\top}\bbeta, \Xb, \bY), \Theta) := \sup_{\bbeta \in \Theta}\EE_{\bbeta} L(CI_{\alpha}(\eb^{(j)\top}\bbeta, \Xb, \bY)).
\end{align*}
The above definitions are extracted from \cite{cai2017confidence} whose work forms the basis of our Lemma \ref{minimax_length_ci}. %We would also like to mention that the our proofs extend to general bounded contrasts in place of the vector $\eb^{(j)}$ as we used above.
We need one final definition before we state the result.
\begin{definition} For a fixed upper bound $\overline w(\cT_K(\bbeta) \cap \mathbb{S}^{p-1})$ of $w(\cT_K(\bbeta) \cap \mathbb{S}^{p-1})$, and $\delta > 0$ let
\begin{align*}
    r_n := \inf_{\bbeta \in S} \frac{\overline w^2(\cT_K(\bbeta)\cap\mathbb{S}^{p-1})}{\sqrt{n}}, ~~~ S:= \bigg\{\bbeta \in K: \bbeta(+\text{ or }-)\frac{\delta\sigma\|\bSigma^{\frac{1}{2}}\|_{\operatorname{op}}^{-1}}{\sqrt{n}}\eb^{(j)}\in K\bigg\},
\end{align*}
where $r_n = \infty$ if $S = \varnothing$.  
% {\color{red} \ldots($S = \varnothing$ and $r_n = \infty$ are related? and we never mention S in the following text\ldots) -- yes, $r_n$ has $S$ in the def.}
\end{definition}
We suppress the dependence of $r_n$ on $\delta$, $K$ and $\overline w(\cT_K(\bbeta) \cap \mathbb{S}^{p-1})$ to ease the notation. In the above definition observe that the set $K$ and the dimension $p$ are also allowed to change with $n$. We have
\begin{lemma}
\label{minimax_length_ci}
Let $K\in\mathbb{R}^p$ be a convex set. For a fixed upper bound $\overline w(\cT_K(\vb) \cap \mathbb{S}^{p-1})$ of $w(\cT_K(\vb) \cap \mathbb{S}^{p-1})$ such that for all $\vb \in K$, $\overline w(\cT_K(\vb) \cap \mathbb{S}^{p-1}) \rightarrow \infty$ and a fixed $\delta > 0$, suppose that $r_n = o(1)$. For any sequence $R_n \geq 2r_n$ such that $R_n = o(1)$, define the parameter space
\begin{align}\label{H:Rn:def}
    \mathcal{H} := \mathcal{H}(R_n) & =\{\bbeta\in K:\, \|\bbeta-\vb\|^2 \leq R_n/\sqrt{n},\,\text{ for } \vb\in K \text{ and } \overline w^2(\cT_K(\vb)\cap\mathbb{S}^{p-1}) \leq R_n\sqrt{n}\}.
    %\mathcal{H}_0 & =\{\bbeta\in K:\,\,\overline w^2(\cT_K(\bbeta)\cap\mathbb{S}^{p-1})=o(\sqrt{n}) \text{ and }\bbeta(+\text{ or }-)\frac{\delta\sigma\|\bSigma^{\frac{1}{2}}\|_{\operatorname{op}}^{-1/2}}{\sqrt{n}}\eb^{(j)}\in K\}.
\end{align}
Then for any $\bbeta^*\in \mathcal{H}$ and sufficiently large $n$ we have 
% {\color{red} (why need $R_n$? can just state $\|\bbeta-\vb\|^2 o(1/\sqrt{n}),\,\text{ for } \vb\in K \text{ and } \overline w^2(\cT_K(\vb)\cap\mathbb{S}^{p-1}) =o(\sqrt{n})$)? -- you cannot just say something is $o(\sqrt{n})$ when you build parameter sets. How do I know which vectors enter and which don't?}
\begin{align*}
    \inf_{CI_{\alpha}(\eb^{(j)\top}\bbeta^*, \Xb, \bY)\in\mathcal{I}_{\alpha}( \mathcal{H})} L(CI_{\alpha}(\eb^{(j)\top}\bbeta^*, \Xb, \bY), \cH) \geq \delta\Big(1-2\alpha-\sqrt{\exp(2\delta^2)-1}\Big)\frac{\sigma\|\bSigma^{\frac{1}{2}}\|_{\operatorname{op}}^{-1}}{\sqrt{n}}.
\end{align*}
\end{lemma}
\begin{remark}
%% REMARK:
%% FOR THE MONOTONE CONE, IF ONE LOOKS AT DIRECTIONS OF THE SORT (0,0,0,\ldots,1,1,1,1,1..,1) NORMALIZED TO BE NORM 1, INSTEAD OF e_j, THEN ONE CAN ALWAYS FIND A MONOTONE VECTOR THAT IS 1/sqrt(n) CLOSE. HENCE FOR SUCH DIRECTIONS ONE CAN CLAIM SOMETHING LIKE CAI AND GUO THEOREM 3, I.E., THAT YOU DON'T NEED TO TAKE THE SUP OVER beta in Theta.

Notice that given a convex set $K$, our Algorithm \ref{algo_unknowncov} is able to perform debiasing asymptotically over the parameter space $\mathcal{H}$ according to Theorem \ref{main_rst_cvs_ls}. The result of Lemma \ref{minimax_length_ci} shows that the length of our confidence interval \eqref{ci} for a single coefficient $\bbeta^{*(j)}$ cannot be much improved asymptotically in a worst case sense, since its length times $\sqrt{n}$ is at least of the order of a constant (assuming $\bSigma$ has bounded spectrum). As mentioned earlier, we cannot expect that the sequence $r_n = o(1)$ for all convex sets $K$. But in all examples we consider in this work, $r_n = o(1)$ holds. 
% {\color{red} (Should also discuss $S$ is not empty in those cases? ) -- this is implied by $r_n = o(1)$.} 
For instance, if $K$ is a monotone cone or positive monotone cone as we will study later in Section \ref{case_mnt} and Section \ref{case_pos_mnt}, a monotone vector comprised of two constant pieces whose jump from the $(j-1)$-th coordinate to the $j$-th coordinate is greater than $\delta\sigma\|\bSigma^{1/2}\|_{\operatorname{op}}^{-1}/\sqrt{n}$ will produce $r_n \asymp \frac{2\log (ep/2)}{\sqrt{n}}$ \citep[see (1.19), (1.22), Proposition 3.1]{bellec2018sharp}. Also, if $K = \{\bbeta : \|\bbeta\|_1 \leq \|\bbeta^*\|_1\}$, there exists a $1$-sparse vector $\vb$ (with $j$-th coefficient equal to $ \|\bbeta^*\|_1$) which gives $r_n = o(1)$ whenever $1=o(\sqrt{n}/\log p)$ and $\|\bbeta^*\|_1 \geq \sigma\|\bSigma^{1/2}\|_{\operatorname{op}}^{-1}/(2\sqrt{n})$. If $K$ is the non-negative orthant cone, a vector of zeros with exception of its $j$-th coordinate being equal to $\sigma \|\bSigma^{\frac{1}{2}}\|^{-1}_{\operatorname{op}}/\sqrt{n}$ will yield $r_n \lesssim \frac{p}{\sqrt{n}}$ so when $p = o(\sqrt{n})$, $r_n = o(1)$.

%(with non-zero coefficient equal to $ \|\bbeta^*\|_1- \delta\sigma\|\bSigma^{1/2}\|_{\operatorname{op}}^{-1}/\sqrt{n}$) which gives $r_n = o(1)$ whenever $1=o(\sqrt{n}/\log p)$ and $\|\bbeta^*\|_1 \geq \sigma\|\bSigma^{1/2}\|_{\operatorname{op}}^{-1}/(2\sqrt{n})$.

% We now discuss a curious observation which may seem unrelated to Lemma \ref{minimax_length_ci} at first glance. Note that in principle $\hat\bmeta$ in step \ref{step_2} might be zero, in which case $\vb_j$ will be directly used as the debiased estimator. When $\hat \bmeta = 0$, the asymptotic statement of Theorem \ref{main_rst_cvs_ls} implies that the length of confidence interval will be smaller than $O(1/\sqrt{n})$. Therefore, by Lemma \ref{minimax_length_ci} we conclude that $\hat\bmeta=0$ can happen only when $r_n \neq o(1)$. It follows that for the cases where $K$ is a monotone cone, a positive monotone cone, or an $\ell_1$ ball (with sufficiently large radius), we will not get $\hat\bmeta=0$ in step \ref{step_2} with high probability.
\end{remark}

%In addition to the above result, we now consider the special case when $K$ is a symmetric polygon, i.e., if $\xb \in K$ it follows that $-\xb \in K$. Suppose that $K$ has a vertex representation $\{\vb_1,\ldots, \vb_k\}$ for $k \in \NN$, and let $\vb_i$ 

We end up this section with a result slightly stronger than Lemma \ref{minimax_length_ci} for the special case when $K$ is a polyhedral cone (i.e. $K = \{\xb \in \RR^n: \Ab \xb \geq 0\}$ for some matrix $\Ab$) as is the case when $K$ is the monotone or positive monotone cone or the non-negative orthant cone. It is well known that polyhedral cones are finitely-generated, i.e., there exists a $k \in \NN$ and unit norm vectors $\wb_1,\ldots, \wb_k$ such that $K = \{\sum_{i \in [k]}\alpha_i \wb_i : \alpha_i \geq 0\}$. We have the following

\begin{lemma}\label{polyhedral:cone:lower:bound} Fix a number $j \in [k]$. Let $\cH(R_n)$ be defined as in \eqref{H:Rn:def}, and set $\nu_n := 2R_n + 2 \frac{\delta^2\|\bSigma^{\frac{1}{2}}\|_{\operatorname{op}}^{-2}}{\sqrt{n}}$. Then
for any $\bbeta^*\in \cH(R_n)$ we have
\begin{align*}
    \inf_{CI_{\alpha}(\wb_j^{\top}\bbeta^*, \Xb, \bY)\in\mathcal{I}_{\alpha}( \mathcal{H}(\nu_n))} L(CI_{\alpha}(\wb_j^{\top}\bbeta^*, \Xb, \bY)) \geq \delta\Big(1-2\alpha-\sqrt{\exp(2\delta^2)-1}\Big)\frac{\sigma\|\bSigma^{\frac{1}{2}}\|_{\operatorname{op}}^{-1}}{\sqrt{n}}.
\end{align*}
\end{lemma}

In other words, if one is interested in performing inference along a generating direction of the cone, the confidence interval length has to be at least $\frac{1}{\sqrt{n}}$ for any $\bbeta^* \in \cH(R_n)$ for all algorithms which return valid $(1-\alpha)$-confidence intervals for all vectors in $\cH(\nu_n)$. Note that since $\nu_n = o(1)$ our debiasing algorithm will produce $(1-\alpha)$-level confidence intervals on $\cH(\nu_n)$ asymptotically, and therefore the length of the confidence intervals for contrasts equal to generating directions of the cone cannot be improved. Unlike Lemma \ref{minimax_length_ci}, Lemma \ref{polyhedral:cone:lower:bound} is not a worst case result since we are not taking $\sup$ over all vectors in the parameter space. We now give concrete examples of sets $K$ for which our algorithm is fully implementable. 

\subsection{Proofs}

%%%%%%%%%%
%% Lower Bound CI length
%%%%%%%%%
\begin{proof}[Proof of Lemma \ref{minimax_length_ci}]
This argument is mostly repeating an argument from \cite{cai2017confidence}. Before the proof, we need to introduce two definitions. The first is the $\chi^2$ distance between two density functions
\begin{align*}
    \chi^2(f_1, f_0) = \int \frac{(f_1(z)-f_0(z))^2}{f_0(z)}dz = \int \frac{f_1^2(z)}{f_0(z)}dz-1.
\end{align*}
The second is the total variation distance (with a scaling factor 2 in front) between two density functions
\begin{align*}
    TV(f_1,f_0) = \int|f_1(z)-f_0(z)|dz
\end{align*}
A well-known fact is that $TV(f_1,f_0)\leq\sqrt{\chi^2(f_1, f_0)}$.

Now we start the proof. Let $K\in\mathbb{R}^p$ be a convex set. The parameter space is defined as
\begin{align*}
    \mathcal{H}=\{\bbeta\in\mathbb{R}^p:\, \|\bbeta-\vb\|^2\leq \frac{R_n}{\sqrt{n}},\,\text{ for } \vb\in K \text{ and } \overline w^2(\cT_K(\vb)\cap\mathbb{S}^{p-1})=R_n\sqrt{n}\},
\end{align*}
which is the space we are able to perform inference on, asymptotically, via the debiasing procedure proposed in this paper. Suppose we want to debias the $j$-th coordinate.
\begin{enumerate}
    \item Let $\delta>0$ be a small positive constant such that $\delta\sigma\sqrt{\|\bSigma^{\frac{1}{2}}\|_{\operatorname{op}}^{-1}}=o(n^{\frac{1}{4}})$. Define
    \begin{align*}
        \mathcal{H}_0=\{\bbeta\in K:\,\,\overline w^2(\cT_K(\bbeta)\cap\mathbb{S}^{p-1})=2r_n\sqrt{n} \text{ and }\bbeta(+\text{ or }-)\frac{\delta\sigma\|\bSigma^{\frac{1}{2}}\|_{\operatorname{op}}^{-1/2}}{\sqrt{n}}\eb^{(j)}\in K\},
    \end{align*}
    and $\mathcal{H}_1=\mathcal{H}$. Since $R_n \geq 2r_n$, for and $r_n = o(1)$ for $n$ large enough it is not hard to see that $\mathcal{H}_0\subseteq\mathcal{H}$. In addition, by the definition of $r_n$, it follows that the set $\cH_0$ is not empty.
    For a given $\bbeta^*\in \mathcal{H}_0$, we find a $\bbeta'$ such that
    \begin{align*}
        \|\bbeta'-\bbeta^*\| = |\bbeta'_j-\bbeta^*_j| = \delta\frac{\sigma\|\bSigma^{\frac{1}{2}}\|_{\operatorname{op}}^{-1/2}}{\sqrt{n}}.
        % \frac{w(\cT_K(\bbeta^*)\cap\mathbb{S}^{p-1})}{\sqrt{n}}
    \end{align*}
    According to the definition of $\mathcal{H}$, we always have $\bbeta'\in\mathcal{H}$. 
    
    \item Let $f_0(\bY|\Xb)$ be the density of $\bY$ given $\Xb$ with the parameter $\bbeta^*$, and $f_1(\bY|\Xb)$ be the density of $\bY$ given $\Xb$ with the parameter $\bbeta'$. Such a conditional distribution of $\bY$ is Gaussian since the noise has a Gaussian distribution with standard error $\sigma$. It can be shown that
    \begin{align*}
        \chi^2(f_1(\bY|\Xb), f_0(\bY|\Xb)) = \exp(\frac{1}{\sigma^2}\|\Xb(\bbeta'-\bbeta^*)\|^2) - 1.
    \end{align*}
    With the fact $\bX_i\sim N(\mathbf{0}, \bSigma)$, we have
    \begin{align*}
        \chi^2(f_1(\bY, \Xb), f_0(\bY, \Xb)) & = \mathbb{E}_{\Xb}\exp(\frac{1}{\sigma^2}\|\Xb(\bbeta'-\bbeta^*)\|^2) - 1\\
        & = \prod_{i=1}^n \mathbb{E}_{\Xb}\exp\Big(\frac{1}{\sigma^2}[\bX_i\T(\bbeta'-\bbeta^*)]^2\Big) - 1\\
        & = \prod_{i=1}^n \mathbb{E}_{\Xb}\exp\Big(\frac{1}{\sigma^2}[(\bSigma^{-\frac{1}{2}}\bX_i)\T\bSigma^{\frac{1}{2}}(\bbeta'-\bbeta^*)]^2\Big) - 1.
    \end{align*}
    Since $(\bSigma^{-\frac{1}{2}}\bX_i)\T\bSigma^{\frac{1}{2}}(\bbeta'-\bbeta^*)=\|\bSigma^{\frac{1}{2}}(\bbeta'-\bbeta^*)\|z_i$ where $z_i\sim N(0, 1)$, by the moment generating function of $\chi^2$ distribution, the above equation becomes
    \begin{align*}
        \chi^2(f_1(\bY, \Xb), f_0(\bY, \Xb)) & = \Big(1- \frac{2\|\bSigma^{\frac{1}{2}}(\bbeta'-\bbeta^*)\|^2}{\sigma^2} \Big)^{-\frac{n}{2}} - 1.
    \end{align*}
    If $\frac{2\|\bSigma^{\frac{1}{2}}(\bbeta'-\bbeta^*)\|^2}{\sigma^2}<\frac{\log2}{2}$, by the inequality $\frac{1}{1-x}\leq\exp(2x)$ for $x\in[0, \frac{\log2}{2}]$, we have
    \begin{align*}
        \chi^2(f_1(\bY, \Xb), f_0(\bY, \Xb)) & \leq \exp\Big(\frac{2n\|\bSigma^{\frac{1}{2}}(\bbeta'-\bbeta^*)\|^2}{\sigma^2}\Big) - 1.
    \end{align*}
    
    \item By Lemma 1 in \cite{cai2017confidence}, for any $CI_{\alpha}(\bbeta^j, \bY, \Xb)\in\mathcal{I}_{\alpha}(\cH)$ we have%\mathcal{I}_{\alpha}(\bbeta^j, K)$ we have
    \begin{align*}
        L(CI_{\alpha}(\bbeta^j, \bY, \Xb)) & \geq \delta\frac{\sigma\|\bSigma^{\frac{1}{2}}\|_{\operatorname{op}}^{-1}}{\sqrt{n}} \Big(1-2\alpha-TV(f_1(\bY, \Xb), f_0(\bY, \Xb))\Big) \\
        & \geq \delta\frac{\sigma\|\bSigma^{\frac{1}{2}}\|_{\operatorname{op}}^{-1}}{\sqrt{n}} \Big(1-2\alpha-\sqrt{\exp(2\delta^2)-1}\Big).
    \end{align*}
    % Let $c = \delta\Big(1-2\alpha-\sqrt{\exp(2\delta^2)-1}\Big)$ finally we get
    % \begin{align*}
    %     L(CI_{\alpha}(\bbeta^j, \bY, \Xb)) \geq c\frac{\sigma\|\bSigma^{\frac{1}{2}}\|_{\operatorname{op}}^{-1/2}}{\sqrt{n}}
    % \end{align*}
\end{enumerate}
\end{proof}

\begin{proof}[Proof of Lemma \ref{polyhedral:cone:lower:bound}]
The proof is the same as that of Lemma \ref{minimax_length_ci} modulo some small changes. For any $\bbeta^* \in \cH(R_n)$ let $\bbeta = \bbeta^* + \frac{\delta \|\bSigma^{\frac{1}{2}}\|_{\operatorname{op}}^{-1}\wb_j}{\sqrt{n}}$. We now argue that $\bbeta \in \cH(\nu_n)$. It is clear that $\bbeta \in K$ by the definition of $K$. Let $\vb$ be such that $\|\bbeta^*-\vb\|^2 \leq R_n/\sqrt{n},\,\text{ for } \vb\in K \text{ and } \overline w^2(\cT_K(\vb)\cap\mathbb{S}^{p-1}) \leq R_n\sqrt{n}$. By the triangle inequality:
\begin{align*}
    \|\vb - \bbeta\| \leq \|\vb - \bbeta^*\| + \frac{\delta \|\bSigma^{\frac{1}{2}}\|_{\operatorname{op}}^{-1}}{\sqrt{n}} \leq \frac{\sqrt{R_n}}{\sqrt[4]{n}} + \frac{\delta \|\bSigma^{\frac{1}{2}}\|_{\operatorname{op}}^{-1}}{\sqrt{n}}.
\end{align*}
Squaring the inequality in the preceding display and using the elementary inequality $(a + b)^2 \leq 2 a^2 + 2 b^2$ shows that $\bbeta \in \cH(\nu_n)$. The rest of the proof is identical to that of Lemma \ref{minimax_length_ci} and we omit the details. 
\end{proof}

%%%%%%%%%%%%%%%%%%
%% Preliminary
%%%%%%%%%%%%%%%%%
\section{Preliminaries Used in the Proofs}
We present several preliminary definitions and results which are needed in the proofs of the future sections. %{\color{purple} $\|\|_{\psi_2}$ norm}
\begin{definition}\label{psi_norm} For a random variable $X \in \RR$, define its $\psi_\ell$ norm by
\begin{align*}
    \|X\|_{\psi_\ell} = \sup_{p \geq 1} p^{-1/\ell}(\EE |X|^p)^{1/p}.
\end{align*}
for $\ell \in \{1,2\}$. For a random vector $\bX \in \RR^d$ define
\begin{align*}
    \|\bX\|_{\psi_\ell} = \sup_{\vb \in \mathbb{S}^{d-1}}\|\bX\T \vb\|_{\psi_\ell}.
\end{align*}
\end{definition}

Next is Gordon's Escape Through Mesh which bounds the restricted operator norm of a Gaussian matrix over a convex set. Details can be found in \cite[Theorem~A]{gordon1988milman}.
\begin{lemma}
\label{gordon_escape}
(Gordon's Escape Through Mesh) Let $K\subset \mathbb{R}^n$ be a convex cone and $\Xb$ be an $n\times p$ standard Gaussian matrix. Then for every $t\geq0$,
\begin{align*}
    \mathbb{P}\Big\{
    \sup\limits_{\ub\in K\bigcap\mathbb{S}^{p-1}} \|\Xb\ub\| \geq \sqrt{n} + w(K\bigcap\mathbb{S}^{p-1}) + t
    \Big\} & \leq 
    e^{-\frac{t^2}{2}},\\
    \mathbb{P}\Big\{
    \inf\limits_{\ub\in K\bigcap\mathbb{S}^{p-1}} \|\Xb\ub\| \leq \sqrt{n-1} - w(K\bigcap\mathbb{S}^{p-1}) - t 
    \Big\} & \leq 
    e^{-\frac{t^2}{2}}.
\end{align*}
\end{lemma}

The next result Lemma \ref{corollary2.6} gives an upper bound of the estimation error of the convex constrained least squares, which is an analogy of Corollary 2.6 in \cite{neykov2019gaussian}. We give a proof here since the proof of Corollary 2.6 is omitted in \cite{neykov2019gaussian}. The proof is similar as the proof of \cite[Lemma 2.3]{neykov2019gaussian}. Lemma \ref{a1} is an intermediate result needed in the proof of Lemma \ref{corollary2.6}.
\begin{lemma}\cite[Lemma A.1]{neykov2019gaussian}
\label{a1}
For any $\vb\in K$ we have the following inequality
\begin{align*}
    \frac{1}{\sqrt{n}}\|\Xb(\hat\bbeta-\vb)\| \leq \frac{4}{\sqrt{n}}\|\Xb(\bbeta^*-\vb)\| + \sqrt{\Big(\frac{4}{n}\langle\Xb(\hat\bbeta-\vb),\epsilon\rangle - \frac{2}{n}\|\Xb(\hat\bbeta-\vb)\|^2 \Big)\vee 0}.
\end{align*}
\end{lemma}
\begin{lemma}
\label{corollary2.6}
For matrix $\Xb$ and vector $\bvarepsilon$, let $\bX_i\sim N(0, \bSigma)$, and $\epsilon_i$ be a zero-mean stochastic noise with finite variance $\sigma^2$. Let $K\in\mathbb{R}^p$ be a convex cone. Fix any $\bbeta^*, \hat\bbeta$ and $\vb$ in $K$. 
Suppose $1\leq w(\bSigma^{\frac{1}{2}}\cT_K(\vb)\cap\mathbb{S}^{p-1})=o(\sqrt{n})$ and $\bbeta^*\in K$. Then with probability at least $1-e^{-w(\bSigma^{\frac{1}{2}}\cT_K(\vb)\cap\mathbb{S}^{p-1})}-3e^{-\frac{\big( w(\bSigma^{\frac{1}{2}}\mathcal{T}_K(\vb)\cap\mathbb{S}^{p-1})\big)^2}{2}}-\frac{\Var(\varepsilon_i^2)}{n\sigma^4}$ we have
\begin{align*}
    \|\bSigma^{\frac{1}{2}}(\bbeta^*-\hat\bbeta)\| \lesssim \|\bSigma^{\frac{1}{2}}(\bbeta^*-\vb)\| + \frac{\sigma w(\bSigma^{\frac{1}{2}}\cT_K(\vb)\cap\mathbb{S}^{p-1})}{\sqrt{n}}.
\end{align*}
\end{lemma}

\begin{remark}\label{remark:after:important:lemma:gaussian:width}
    In the above Lemma, in the case when $\bSigma$ has bounded spectrum, one can substitute $w(\bSigma^{\frac{1}{2}}\cT_K(\vb)\cap\mathbb{S}^{p-1}) \leq \|\bSigma^{1/2}\|_{\operatorname{op}}\|\bSigma^{-1/2}\|_{\operatorname{op}}w(\cT_K(\vb)\cap\mathbb{S}^{p-1})$ (see Remark 1.7 \cite{plan2016generalized}). We may substitute $w(\bSigma^{\frac{1}{2}}\cT_K(\vb)\cap\mathbb{S}^{p-1})$ with any upper bound $\overline  w(\bSigma^{\frac{1}{2}}\cT_K(\vb)\cap\mathbb{S}^{p-1})$, and the statement (including the high-probability guarantee) continues to hold with $\overline w(\bSigma^{\frac{1}{2}}\cT_K(\vb)\cap\mathbb{S}^{p-1})$ in place of $w(\bSigma^{\frac{1}{2}}\cT_K(\vb)\cap\mathbb{S}^{p-1})$.
\end{remark}
\begin{proof}
Consider the ``empirical process'' term
\begin{align*}
    I & = \frac{2}{n}\langle\Xb(\hat\bbeta-\vb),\bvarepsilon\rangle - \frac{1}{n}\|\Xb(\hat\bbeta-\vb)\|^2.
\end{align*}
Note that the unit vector $\frac{\bSigma^{\frac{1}{2}}(\hat\bbeta-\vb)}{\|\bSigma^{\frac{1}{2}}(\hat\bbeta-\vb)\|}\in\bSigma^{\frac{1}{2}}\cT_K(\vb)\cap\mathbb{S}^{p-1}$, and $\Xb\bSigma^{-\frac{1}{2}}$ is a standard normal matrix. By Gordon's escape through mesh (Lemma \ref{gordon_escape}), with probability at least $1-e^{-\frac{t^2}{2}}$ we have
\begin{align}\label{gordon:lower:bound}
    \Big\| \Xb\bSigma^{-\frac{1}{2}}\frac{\bSigma^{\frac{1}{2}}(\hat\bbeta-\vb)}{\|\bSigma^{\frac{1}{2}}(\hat\bbeta-\vb)\|} \Big\| \geq 
    \inf\limits_{\wb\in\bSigma^{\frac{1}{2}}\cT_K(\vb)\cap\mathbb{S}^{p-1}}\|\Xb\bSigma^{-\frac{1}{2}}\wb\| \geq \sqrt{n-1}-w(\bSigma^{\frac{1}{2}}\cT_K(\vb)\cap\mathbb{S}^{p-1}) - t.
\end{align}
Then
{\footnotesize
\begin{align*}
    I = & \frac{2}{n}\langle\Xb\bSigma^{-\frac{1}{2}}\bSigma^{\frac{1}{2}}(\hat\bbeta-\vb),\bvarepsilon\rangle - \frac{1}{n}\|\Xb\bSigma^{-\frac{1}{2}}\bSigma^{\frac{1}{2}}(\hat\bbeta-\vb)\|^2 \\
    \leq & \frac{2}{n}(\sqrt{n-1}-w(\bSigma^{\frac{1}{2}}\cT_K(\vb)\cap\mathbb{S}^{p-1})-t)\|\bSigma^{\frac{1}{2}}(\hat\bbeta-\vb)\|\langle\frac{\bSigma^{\frac{1}{2}}(\hat\bbeta-\vb)}{\|\bSigma^{\frac{1}{2}}(\hat\bbeta-\vb)\|}, \frac{(\Xb\bSigma^{-\frac{1}{2}})\T\bvarepsilon}{\sqrt{n-1}-w(\bSigma^{\frac{1}{2}}\cT_K(\vb)\cap\mathbb{S}^{p-1})-t}\rangle \\
    & - \frac{1}{n}(\sqrt{n-1}-w(\bSigma^{\frac{1}{2}}\cT_K(\vb)\cap\mathbb{S}^{p-1})-t)^2\|\bSigma^{\frac{1}{2}}(\hat\bbeta-\vb)\|^2.
\end{align*}
}
Using the fact $-a^2+2ab \leq b^2$, with probability $1-e^{-\frac{t^2}{2}}$ we have
\begin{align*}
    I \leq \frac{\Big( \sup\limits_{\ub\in\bSigma^{\frac{1}{2}}\cT_K(\vb)\bigcap\mathbb{S}^{p-1}}\langle\ub, \frac{1}{\sqrt{n}}(\Xb\bSigma^{\frac{-1}{2}})\T\bvarepsilon\rangle \Big)^2}{(\sqrt{n-1}-w(\bSigma^{\frac{1}{2}}\cT_K(\vb)\cap\mathbb{S}^{p-1})-t)^2}.
\end{align*}
Note that conditioning on the error term $\bvarepsilon$, the vector $\frac{1}{\sqrt{n}}(\Xb\bSigma^{-\frac{1}{2}})\T\bvarepsilon \sim N(0, \Ib\frac{\|\bvarepsilon\|^2_2}{n})$. Let 
$$I_{up} = \sup\limits_{\ub\in\bSigma^{\frac{1}{2}}\cT_K(\vb)\bigcap\mathbb{S}^{p-1}}\langle\ub, \frac{1}{\sqrt{n}}(\Xb\bSigma^{\frac{-1}{2}})\T\bvarepsilon\rangle,$$
by a concentration inequality of Gaussian process with finite variance \cite[Theorem 5.8]{boucheron2013concentration}, we have
\begin{align*}
    \PP(I_{up}-\EE I_{up} \geq \sqrt{2t}\frac{\|\bvarepsilon\|}{\sqrt{n}})\leq e^{-t}.
\end{align*}
By the definition of Gaussian complexity $\EE I_{up} = w(\bSigma^{\frac{1}{2}}\cT_K(\vb)\cap\mathbb{S}^{p-1})\frac{\|\bvarepsilon\|}{\sqrt{n}}$ conditional on $\bvarepsilon$. Then with probability $1-e^{-t}$ we have
\begin{align*}
    I_{up} \leq (w(\bSigma^{\frac{1}{2}}\cT_K(\vb)\cap\mathbb{S}^{p-1})+\sqrt{2t})\frac{\|\bvarepsilon\|}{\sqrt{n}},
\end{align*}
thus with probability $1-e^{-t}-e^{-\frac{t^2}{2}}$ we have
\begin{align*}
    I \leq \frac{\Big( (w(\bSigma^{\frac{1}{2}}\cT_K(\vb)\cap\mathbb{S}^{p-1})+\sqrt{2t})\frac{\|\bvarepsilon\|}{\sqrt{n}} \Big)^2}{(\sqrt{n-1}-w(\bSigma^{\frac{1}{2}}\cT_K(\vb)\cap\mathbb{S}^{p-1})-t)^2}.
\end{align*}
Then by lemma \ref{a1} we have
\begin{align}
\label{bound_X(beta-v)}
    \frac{1}{\sqrt{n}}\|\Xb(\hat\bbeta-\vb)\| & \leq \frac{4}{\sqrt{n}}\|\Xb(\bbeta^*-\vb)\| + \sqrt{2I} \nonumber\\
    & \leq \frac{4}{\sqrt{n}}\|\Xb(\bbeta^*-\vb)\| + \frac{\sqrt{2}(w(\bSigma^{\frac{1}{2}}\cT_K(\vb)\cap\mathbb{S}^{p-1})+\sqrt{2t})\frac{\|\bvarepsilon\|}{\sqrt{n}}}{\sqrt{n-1}-w(\bSigma^{\frac{1}{2}}\cT_K(\vb)\cap\mathbb{S}^{p-1})-t}.
\end{align}
The terms can be rewritten as
\begin{align*}
    \|\Xb(\hat\bbeta-\vb)\| & = \Big\| \Xb\bSigma^{-\frac{1}{2}}\frac{\bSigma^{\frac{1}{2}}(\hat\bbeta-\vb)}{\|\bSigma^{\frac{1}{2}}(\hat\bbeta-\vb)\|} \Big\| \|\bSigma^{\frac{1}{2}}(\hat\bbeta-\vb)\|,\\
    \|\Xb(\bbeta^*-\vb)\| & = \Big\| \Xb\bSigma^{-\frac{1}{2}}\frac{\bSigma^{\frac{1}{2}}(\bbeta^*-\vb)}{\|\bSigma^{\frac{1}{2}}(\bbeta^*-\vb)\|} \Big\| \|\bSigma^{\frac{1}{2}}(\bbeta^*-\vb)\|.
\end{align*}
Observe that both $\bSigma^{\frac{1}{2}}(\hat\bbeta-\vb)$ and $\bSigma^{\frac{1}{2}}(\bbeta^*-\vb)$ belong to $\bSigma^{\frac{1}{2}}\cT_K(\vb)$. We can bound the terms $\Big\| \Xb\bSigma^{-\frac{1}{2}}\frac{\bSigma^{\frac{1}{2}}(\hat\bbeta-\vb)}{\|\bSigma^{\frac{1}{2}}(\hat\bbeta-\vb)\|} \Big\|$ and $\Big\| \Xb\bSigma^{-\frac{1}{2}}\frac{\bSigma^{\frac{1}{2}}(\bbeta^*-\vb)}{\|\bSigma^{\frac{1}{2}}(\bbeta^*-\vb)\|} \Big\|$ by Gordon's escape through mesh (Lemma \ref{gordon_escape}), then with probability at least$1-e^{-t}-3e^{-\frac{t^2}{2}}$ we have
\begin{align*}
    \|\bSigma^{\frac{1}{2}}(\hat\bbeta-\vb)\| \leq & \frac{4(\sqrt{n}+w(\bSigma^{\frac{1}{2}}\cT_K(\vb)\cap\mathbb{S}^{p-1})+t)}{(\sqrt{n-1}-w(\bSigma^{\frac{1}{2}}\cT_K(\vb)\cap\mathbb{S}^{p-1})-t)}\|\bSigma^{\frac{1}{2}}(\bbeta^*-\vb)\|\\
    & + \frac{\sqrt{2}(w(\bSigma^{\frac{1}{2}}\cT_K(\vb)\cap\mathbb{S}^{p-1})+\sqrt{2t})\frac{\|\bvarepsilon\|}{\sqrt{n}}}{\sqrt{n}\bigg( \sqrt{\frac{n-1}{n}}-\frac{w(\bSigma^{\frac{1}{2}}\cT_K(\vb)\cap\mathbb{S}^{p-1})+t}{\sqrt{n}} \bigg)^2}.
\end{align*}
Since $\mathbb{E}\frac{\|\bvarepsilon\|^2}{n}=\sigma^2$ and $Var\frac{\|\bvarepsilon\|^2}{n}=\frac{\Var(\bvarepsilon_i^2)}{n}$, by Chebyshev's inequality we have
\begin{align}
\label{chebyshev_sigma}
    \PP\Big( \Big|\frac{\|\bvarepsilon\|^2}{n}-\sigma^2\Big| \geq t \Big) \leq \frac{\Var(\bvarepsilon_i^2)}{nt^2}.
\end{align}
Plug in $t = \sigma^2$ to get $\frac{\|\bvarepsilon\|}{\sqrt{n}}\leq\sqrt{2\sigma}$ with probability at least $1-\frac{\Var(\bvarepsilon_i^2)}{n\sigma^4}$. And by the triangle inequality $\|\bSigma^{\frac{1}{2}}(\hat\bbeta-\bbeta^*)\| - \|\bSigma^{\frac{1}{2}}(\bbeta^*-\vb)\| \leq \|\bSigma^{\frac{1}{2}}(\hat\bbeta-\vb)\|$, we can get
\begin{align*}
    \|\bSigma^{\frac{1}{2}}(\hat\bbeta-\bbeta^*)\| \leq & \Big( \frac{4(\sqrt{n}+w(\bSigma^{\frac{1}{2}}\cT_K(\vb)\cap\mathbb{S}^{p-1})+t)}{(\sqrt{n-1}-w(\bSigma^{\frac{1}{2}}\cT_K(\vb)\cap\mathbb{S}^{p-1})-t)}+1\Big) \|\bSigma^{\frac{1}{2}}(\bbeta^*-\vb)\|\\
    & + \frac{2\sigma(w(\bSigma^{\frac{1}{2}}\cT_K(\vb)\cap\mathbb{S}^{p-1})+\sqrt{2t})}{\sqrt{n}\bigg( \sqrt{\frac{n-1}{n}}-\frac{w(\bSigma^{\frac{1}{2}}\cT_K(\vb)\cap\mathbb{S}^{p-1})+t}{\sqrt{n}} \bigg)^2},
\end{align*}
with probability at least $1-e^{-t}-3e^{-\frac{t^2}{2}}-\frac{\Var(\epsilon_i^2)}{n\sigma^4}$. Finally, given the assumption $w(\bSigma^{\frac{1}{2}}\cT_K(\vb)\cap\mathbb{S}^{p-1})=o(\sqrt{n})$, we plug in $t = w(\bSigma^{\frac{1}{2}}\cT_K(\vb)\cap\mathbb{S}^{p-1})$ to get
\begin{align*}
    \|\bSigma^{\frac{1}{2}}(\hat\bbeta-\bbeta^*)\| \lesssim \|\bSigma^{\frac{1}{2}}(\bbeta^*-\vb)\| + \frac{w(\bSigma^{\frac{1}{2}}\cT_K(\vb)\cap\mathbb{S}^{p-1})}{\sqrt{n}}\sigma,
\end{align*}
with probability at least $1-e^{-w(\bSigma^{\frac{1}{2}}\cT_K(\vb)\cap\mathbb{S}^{p-1})}-3e^{-\frac{\big( w(\bSigma^{\frac{1}{2}}\cT_K(\vb)\cap\mathbb{S}^{p-1}\big)^2}{2}}-\frac{\Var(\bvarepsilon_i^2)}{n\sigma^4}$.
\end{proof}
Next result bounds the supremum of a general covariance Gaussian process over a set $\cT_K(\vb)\cap \mathbb{S}^{p-1}$. Notice that Lemma \ref{connect_gs_comp} still holds if we replace $\cT_K(\vb)\cap \mathbb{S}^{p-1}$ by any other set in $\mathbb{S}^{p}$.
\begin{lemma}
\label{connect_gs_comp}
% Assume that $\omega(\cT_K(\vb)) \geq c$ for some absolute constant $c > 0$. 
For a convex set $K\subseteq\RR^p$, $\vb\in K$, $g\sim N(\mathbf{0}, \Ib)$, and $\bSigma\in \RR^{p\times p}$, we have
\begin{align*}
    \EE \sup_{\ub \in \cT_K(\vb) \cap\mathbb{S}^{p-1}} |g\T \bSigma^{1/2}\ub| \leq  C \|\bSigma^{1/2}\|_{\operatorname{op}}\, w(\cT_K(\vb)\cap \mathbb{S}^{p-1}),
\end{align*}
where $C\in\RR$ is a constant. 
\end{lemma}
\begin{proof}
First note that 
\begin{align*}
\EE \sup_{\ub \in \cT_K(\vb) \cap\mathbb{S}^{p-1}} |g\T \bSigma^{1/2}\ub| = \EE \sup_{\ub \in (\cT_K(\vb) \cup -\cT_K(\vb)) \cap \mathbb{S}^{p-1}} g\T \bSigma^{1/2}\ub.
\end{align*}

Now we will compare the process $X_{\ub} = g\T \bSigma^{1/2}\ub$ to the process $Y_{\ub} = \|\bSigma^{1/2}\|_{\operatorname{op}}\, g\T \ub$. We have
\begin{align*}
\EE (X_{\ub} - X_{\ub'})^2 = \EE (g\T \bSigma^{1/2}\ub - g\T \bSigma^{1/2}\ub')^2 = (\ub-\ub')\T \bSigma (\ub - \ub') \leq \|\bSigma\|_{\operatorname{op}} \|\ub - \ub'\|^2,
\end{align*}
and
\begin{align*}
\EE (Y_{\ub} - Y_{\ub'})^2 = \|\bSigma^{1/2}\|_{\operatorname{op}}^2\, \EE (g\T \ub - g\T \ub')^2 = \|\bSigma\|_{\operatorname{op}} \|\ub - \ub'\|^2 \geq \EE (X_{\ub} - X_{\ub'})^2.
\end{align*}
Hence by Sudakov-Fernique's inequality \cite[Theorem 7.2.11]{vershynin2018high}, we can claim that
\begin{align*}
 \EE \sup_{\ub \in (\cT_K(\vb) \cup -\cT_K(\vb)) \cap\mathbb{S}^{p-1}} g\T \bSigma^{1/2}\ub & \leq  \|\bSigma^{1/2}\|_{\operatorname{op}}\, \EE \sup_{\ub \in (\cT_K(\vb) \cup -\cT_K(\vb)) \cap\mathbb{S}^{p-1}} g\T \ub\\
& =  \|\bSigma^{1/2}\|_{\operatorname{op}}\, \EE \sup_{\ub \in \cT_K(\vb) \cap\mathbb{S}^{p-1}} |g\T \ub|.
\end{align*}
Notice that the Gaussian complexity $w(\cT_K(\vb)\cap\mathbb{S}^{p-1}) = \EE \sup_{\ub \in \cT_K(\vb)\cap\mathbb{S}^{p-1}} g\T \ub$ has the same order as the quantity $\EE \sup_{\ub \in \cT_K(\vb)\cap\mathbb{S}^{p-1}} |g\T \ub|$ \cite[Exercise 7.6.9]{vershynin2018high}, so we get the desired result
\begin{align*}
    \EE \sup_{\ub \in \cT_K(\vb) \cap\mathbb{S}^{p-1}} |g\T \bSigma^{1/2}\ub|
    \leq C \|\bSigma^{1/2}\|_{\operatorname{op}}\, w(\cT_K(\vb)\cap\mathbb{S}^{p-1}).
\end{align*}
\end{proof}

The next result demonstrates a property of the projection of a vector $\yb\in\RR^p$ into the intersection of a convex cone $K$ and the unit sphere $\mathbb{S}^{p-1}$.
% {\color{red} this Lemma C.7 needs to be changed over the ball otherwise its not true!}
\begin{lemma}
\label{project_argsup}
Let $K$ be a closed convex cone, and $\mathbb{B}_2^p$ be the unit ball. For any vector $\yb\in\RR^p$, we have
\begin{align*}
    \argsup_{\ub \in K\cap \mathbb{B}_2^p} \yb\T \ub = \frac{\Pi_K(\yb)}{\|\Pi_K(\yb)\|},
\end{align*}
where for the right hand side we understand $0/0 = 0$.
\end{lemma}
\begin{proof}
Arbitrarily pick $\ub\in K$. By the characterization of the projection on a closed convex set \cite[Proposition 1]{moreau1962decomposition}, 
\begin{align}
\label{cvx_proj_prop}
    (\ub - \Pi_K(\yb))(\yb - \Pi_K(\yb))\leq0.
\end{align}
Since $K$ is a convex cone, $2\Pi_K(\yb)$ and $\frac{1}{2}\Pi_K(\yb)$ are in $K$. Plug them into \eqref{cvx_proj_prop} get
\begin{align}
\label{proj_dotprod}
    \yb\T\Pi_K(\yb) = \|\Pi_K(\yb)\|^2.
\end{align}
Expand \eqref{cvx_proj_prop} and use the fact at \eqref{proj_dotprod} to get the following inequality
\begin{align*}
    \ub\T\yb \leq \ub\T\Pi_K(\yb),
\end{align*}
thus
\begin{align*}
    \sup_{\ub \in K\cap\mathbb{B}_2^p} \ub\T\yb \leq \sup_{\ub \in K\cap\mathbb{B}_2^p} \ub\T\Pi_K(\yb).
\end{align*}
By Cauchy-Schwartz inequality,
\begin{align*}
    \sup_{\ub \in K\cap\mathbb{B}_2^p} \ub\T\Pi_K(\yb) \leq \|\Pi_K(\yb)\|.
\end{align*}
Combine the above two inequalities with \eqref{proj_dotprod}, and the desired result is obtained
\begin{align*}
    \sup_{\ub \in K\cap\mathbb{B}_2^p} \ub\T\yb \leq \|\Pi_K(\yb)\| \leq \yb\T\frac{\Pi_K(\yb)}{\|\Pi_K(\yb)\|}.
\end{align*}
\end{proof}

\section{Proofs from Section \ref{algo:sec}}

%%%%%%%%%%%%%%%%%%
%%
%%%%%%%%%%%%%%%%%
%\section{Proof of Lemma \ref{feasible_point}}
\begin{proof}[Proof of Lemma \ref{feasible_point}]
We will use Corollary 1.10 of \cite{mendelson2016upper}. In his notation, it is sufficient to establish that $\xi_i = \bmeta\T \bX_i$ has a bounded $4$-th (say) moment,  that the variables $\bX_i\T\ub$ and $\bX_i\T (\ub - \ub')$ for $\ub, \ub' \in \cT_K(\vb) \cap \mathbb{S}^{p-1}$ are sub-Gaussian, and that the quantity $\sup_{\ub \in \cT_K(\vb) \cap \mathbb{S}^{p-1}} \ub\T \bSigma \ub$ is bounded. The latter is clearly true since $\bSigma$ has a bounded spectrum. For the first claim let us write
\begin{align*}
    \EE \xi_i^4 = \EE (\eb^{(j)\top}\bSigma^{-1} \bX_i)^4 < \infty,
\end{align*}
since the variable $\eb^{(j)\top}\bSigma^{-1} \bX_i$ is sub-Gaussian (this follows from the fact that $\|\eb^{(j)\top}\bSigma^{-1}\|$ is bounded and the fact that $\bX_i$ is sub-Gaussian). Finally the fact that $\bX_i\T\ub$ and $\bX_i\T (\ub - \ub')$ for $\ub, \ub' \in \cT_K(\vb) \cap \mathbb{S}^{p-1}$ is trivial (since $\bX_i$ are sub-Gaussian) and the proof is complete. For completeness we also state that 
\begin{align*}
    \sup_{\ub \in \cT_K(\vb) \cap\mathbb{S}^{p-1}} |(\bmeta\T \hat \bSigma - \eb^{(j)\top}) \ub| \lesssim  \frac{\overline w(\cT_K(\vb)\cap\mathbb{S}^{p-1})}{\sqrt{n}},
\end{align*}
holds with probability at least $1 - C'\log^4 n/n - \exp(-C''_{\bSigma} \overline w(\cT_K(\vb) \cap \mathbb{S}^{p-1}))$, according to Corollary 1.10 of \cite{mendelson2016upper}.
\end{proof}

%%%%%%%%%%%%%%%
%%
%%%%%%%%%%%%%%
%\section{Proof of Corollary \ref{nonempty_int_Q}}
\begin{proof}[Proof of Corollary \ref{nonempty_int_Q}]
By Lemma \ref{feasible_point}, we know that the vector $\bmeta\T = \eb^{(j)\top}\bSigma^{-1}$ is in $Q$ with high probability. Now the idea is to show that there exists a small $\delta>0$ such that $\BB_{\delta}(\eb^{(j)\top}\bSigma^{-1})$ is inside of $Q$ with high probability. Now let $\bx$ be a unit vector. We have
\begin{align*}
&\sup_{\ub \in \cT_K(\vb) \cap \mathbb{S}^{p-1}} |\big((\eb^{(j)\top}\bSigma^{-1} + \delta \bx\T) \Xb\T \Xb/n - \eb^{(j)\top} \big) \ub| \\
& < \frac{\rho\overline w(\cT_K(\vb)\cap\mathbb{S}^{p-1})}{\sqrt{n}} + \sup_{\bx \in \mathbb{S}^{p-1}, \ub \in \cT_K(\vb) \cap \mathbb{S}^{p-1}} \frac{\delta}{n}|\bx\T \Xb\T \Xb \ub| \\
& \leq \frac{\rho\overline w(\cT_K(\vb)\cap\mathbb{S}^{p-1})}{\sqrt{n}} + \frac{\delta}{n} \sup_{\bx \in \mathbb{S}^{p-1}}\|\Xb \bx\| \sup_{\ub \in \cT_K(\vb) \cap \mathbb{S}^{p-1}}\|\Xb\ub\|.
\end{align*}
If $\bX$ is bounded the above quantities are bounded with probability $1$ hence the conclusion follows. Next we consider the case when $\bX \sim N(0,\bSigma)$. Let $\tilde\bX$ be an $n \times p$ matrix with independent $N(0, 1)$ entries. The last two terms $\|\cdot\|$ are bounded as
\begin{align*}
    \sup_{\bx \in \mathbb{S}^{p-1}}\|\Xb \bx\| & = \sup_{\bx \in \mathbb{S}^{p-1}}\|\tilde\Xb \bSigma^{1/2} \bx\| \leq \|\tilde\Xb\|_{\operatorname{op}}\|\bSigma^{1/2}\|_{\operatorname{op}},\\
    \sup_{\ub \in \cT_K(\vb) \cap \mathbb{S}^{p-1}}\|\Xb\ub\| & = \sup_{\ub \in \cT_K(\vb) \cap \mathbb{S}^{p-1}}\|\tilde\Xb \bSigma^{1/2} \ub\| \leq \|\tilde\Xb\|_{\operatorname{op}}\|\bSigma^{1/2}\|_{\operatorname{op}}.
\end{align*}
By the tail bound of the operator norm of Gaussian matrix \cite[Corollary 7.3.3]{vershynin2018high}, $\|\tilde\Xb\|_{\operatorname{op}}$ is bounded by $\sqrt{n} + \sqrt{p}$ with high probability, so that
\begin{align*}
    \sup_{\ub \in \cT_K(\vb) \cap \mathbb{S}^{p-1}} |\big((\eb^{(j)\top}\bSigma^{-1} + \delta \bx\T) \Xb\T \Xb/n - \eb^{(j)\top} \big) \ub| & < \frac{\rho\overline w(\cT_K(\vb)\cap\mathbb{S}^{p-1})}{\sqrt{n}} + \frac{\delta}{n}\|\bSigma^{1/2}\|^2_{op}(\sqrt{n} + \sqrt{p})^2.
\end{align*}
Let $\epsilon=\frac{\delta}{n}\|\bSigma^{1/2}\|^2_{op}(\sqrt{n} + \sqrt{p})^2$. Since we can find such a $\delta$ for any $\epsilon > 0$, the ball $\BB_{\delta}(\eb^{(j)\top}\bSigma^{-1})$ is inside of $Q$. Thus $Q$ has a non-empty interior with high probability.
\end{proof}

%%%%%%%%%
%%
%%%%%%%%%
%\section{}
%{\color{red} this is slightly wrong. needs to change the $\psi(\bmeta_n)$ to be sup over the unit ball (since its absolute value) and then Lemma C.7 needs to be changed over the ball.}
\begin{proof}[Proof of Lemma \ref{sg_psi}]
Let $\lambda = \frac{\rho\overline w(\cT_K(\vb)\cap\mathbb{S}^{p-1})}{\sqrt{n}}$. We have
\begin{align*}
    \psi(\bmeta_n) & = \sup_{\ub \in \cT_K(\vb) \cap \mathbb{S}^{p-1}} |(\bmeta_n\T \hat \bSigma - \eb^{(j)\top}) \ub| - \lambda\\
    & = \sup_{\ub \in \cT_K(\vb) \cap \mathbb{B}_2^p} |(\bmeta_n\T \hat \bSigma - \eb^{(j)\top}) \ub| - \lambda\\
    & = \max\{ \sup_{\ub \in \cT_K(\vb) \cap \mathbb{B}_2^p} (\bmeta_n\T \hat \bSigma - \eb^{(j)\top}) \ub - \lambda, \sup_{\ub \in -\cT_K(\vb) \cap \mathbb{B}_2^p} (\bmeta_n\T \hat \bSigma - \eb^{(j)\top}) \ub - \lambda \}.
\end{align*}
Let $\psi_0(\bmeta) = \sup_{\ub \in \cT_K(\vb) \cap\mathbb{B}_2^p} (\bmeta\T \hat \bSigma - \eb^{(j)\top}) \ub - \lambda$, and $\psi_1(\bmeta) = \sup_{\ub \in -\cT_K(\vb) \cap \mathbb{B}_2^p} (\bmeta\T \hat \bSigma - \eb^{(j)\top}) \ub - \lambda$. 
The subgradient of $\psi_0(\bmeta)$is
\begin{align*}
    \partial\psi_0(\bmeta) = \hat\bSigma \argsup_{\ub \in \cT_K(\vb) \cap \mathbb{B}_2^p} (\bmeta\T \hat \bSigma - \eb^{(j)\top}) \ub,
\end{align*}
% {\color{blue} why is $\argsup_{\ub \in \cT_K(\vb) \cap \mathbb{S}^{p-1}} (\bmeta\T \hat \bSigma - \eb^{(j)\top}) \ub = \Pi_{\cT_K(\vb) \cap \mathbb{S}^{p-1}}(\hat\bSigma\bmeta_n - \eb^{(j)})$? this should be a unit vector so you probably want to normalize. it also needs a little proof why this is the largest dot product with the cone you cna have.} 
since for any $\yb\in\RR^{p}$,
\begin{align*}
    \psi_0(\yb) - \psi_0(\xb) & = \sup_{\ub \in \cT_K(\vb) \cap \mathbb{B}_2^p} (\yb\T \hat \bSigma - \eb^{(j)\top}) \ub - \sup_{\ub \in \cT_K(\vb) \cap \mathbb{B}_2^p} (\xb\T \hat \bSigma - \eb^{(j)\top}) \ub\\
    & \geq \langle \hat \bSigma \yb - \eb^{(j)}, \argsup_{\ub \in \cT_K(\vb) \cap \mathbb{B}_2^p} (\xb\T \hat \bSigma - \eb^{(j)\top}) \ub \rangle - \sup_{\ub \in \cT_K(\vb) \cap \mathbb{B}_2^p} (\xb\T \hat \bSigma - \eb^{(j)\top}) \ub\\
    & = \langle \yb - \xb, \hat\bSigma \argsup_{\ub \in \cT_K(\vb) \cap \mathbb{B}_2^p} (\xb\T \hat \bSigma - \eb^{(j)\top}) \ub \rangle.
\end{align*}
In the above observe that the ``$\argsup$'' is actually ``$\argmax$'' since the set $\cT_K(\vb) \cap \mathbb{B}_2^p$ is compact an the function $\ub \mapsto (\bmeta\T \hat \bSigma - \eb^{(j)\top}) \ub$ is continuous. Similarly, the subgradient of $\psi_1(\bmeta)$ is
\begin{align*}
    \partial\psi_1(\bmeta) = \hat\bSigma \argsup_{\ub \in -\cT_K(\vb) \cap \mathbb{B}_2^p} (\bmeta\T \hat \bSigma - \eb^{(j)\top}) \ub.
\end{align*}

By Lemma \ref{project_argsup}, the subgradient of $\psi_0$ and $\psi_1$ are equivalent to
\begin{align*}
    \partial\psi_0(\bmeta) &= \hat\bSigma \argsup_{\ub \in \cT_K(\vb) \cap \mathbb{B}_2^p} (\bmeta\T \hat \bSigma - \eb^{(j)\top}) \ub = \hat\bSigma \,\phi_0(\bmeta),\\%\frac{\Pi_{\cT_K(\vb)}(\hat\bSigma\bmeta - \eb^{(j)})}{\|\Pi_{\cT_K(\vb)}(\hat\bSigma\bmeta - \eb^{(j)})\|} \\
    \partial\psi_1(\bmeta) &= \hat\bSigma \argsup_{\ub \in -\cT_K(\vb) \cap \mathbb{B}_2^p} (\bmeta\T \hat \bSigma - \eb^{(j)\top}) \ub = \hat\bSigma \,\phi_1(\bmeta)%\frac{\Pi_{-\cT_K(\vb)}(\hat\bSigma\bmeta - \eb^{(j)})}{\|\Pi_{-\cT_K(\vb)}(\hat\bSigma\bmeta - \eb^{(j)})\|}
\end{align*}

By the pointwise maximum rule of subgradient \cite[Theorem 1.13]{shor2012minimization}, the subgradeint of $\psi$ at $\bmeta$ is $\partial\psi_0(\bmeta)$ if $\psi_0(\bmeta) > \psi_1(\bmeta)$, is $\partial\psi_1(\bmeta)$ otherwise. It is simple to see that $\psi_0(\bmeta) = (\bmeta\T \hat \bSigma - \eb^{(j)\top}) \phi_0(\bmeta)$ and similarly $\psi_1(\bmeta) = (\bmeta\T \hat \bSigma - \eb^{(j)\top}) \phi_1(\bmeta)$ even when $\phi_0(\bmeta)$ or $\phi_1(\bmeta)$ are zero vectors. This completes the proof.
\end{proof}

%%%%%%%%%
%%%
%%%%%%%%%
%\section{}
\begin{proof}[Proof of Lemma \ref{cvg_sbgrad}]
Let $\bmeta^* \in \argmin_{\bmeta\in Q} \|\hat\bSigma^{\frac{1}{2}}\bmeta\|$ be a constrained minima such that $\|\bmeta^*\|$ is the smallest. Note that this implies that $\bmeta^* \in \operatorname{col}(\hat\bSigma^{\frac{1}{2}})$. Let $\bmeta_1$ be the initial point with a finite $\ell_2$ norm. By Corollary \ref{nonempty_int_Q} there exists a strictly feasible point $\bmeta^{sf}$ such that $\psi(\bmeta^{sf})<0$. It is not hard to see that $\|\bmeta^*\|$ is bounded, since $\|\hat\bSigma^{\frac{1}{2}}\bmeta^*\|\leq\|\hat\bSigma^{\frac{1}{2}}\bmeta_1\|$ is bounded and $\|\bmeta^*\|\leq\|\hat\bSigma^{\frac{1}{2}}\bmeta^*\|(\lambda_{\min}^+(\hat \bSigma^{1/2}))^{-1}$ where $\lambda_{\min}^+(\hat \bSigma^{1/2})$ is the smallest positive eigenvalue of $\hat \bSigma^{1/2}$. The latter holds by the definition of $\bmeta^*$, and the fact that $\bmeta^* \in \operatorname{col}(\hat \bSigma^{1/2})$. Furthermore, there exists at least one $\bmeta^{sf}$ which is $\|\bmeta^{sf}\|$ bounded, since according to Corollary \ref{nonempty_int_Q} $\bmeta^{sf}=\eb^{(j)\top}\bSigma^{-1}$ is a choice of $\bmeta^{sf}$. Thus $\|\bmeta_1-\bmeta^*\|$ and $\|\bmeta_1-\bmeta^{sf}\|$ are bounded. Let $C_1$ be such a constant satisfying $\|\bmeta_1-\bmeta^*\|\leq C_1$ and $\|\bmeta_1-\bmeta^{sf}\|\leq C_1$. 

We also note that $\|\gb_n\| \leq \hat\bSigma^{\frac{1}{2}}\frac{\hat\bSigma^{\frac{1}{2}}\bmeta_n}{\|\hat\bSigma^{\frac{1}{2}}\bmeta_n\|} \leq \|\hat\bSigma^{\frac{1}{2}}\|_{\operatorname{op}}$ for $\bmeta_n\in Q$; and obviously $\|\gb_n\| \leq \|\hat\bSigma\|_{\operatorname{op}}$ for $\bmeta_n\notin Q$. Define a constant $C_2=\max\{\|\hat\bSigma^{\frac{1}{2}}\|_{\operatorname{op}}, \|\hat\bSigma\|_{\operatorname{op}}\}$, so that $\|\gb_n\| \leq C_2$.

Now we show that such a subgradient method converges in finite iterations. Let $f(\bmeta)\coloneqq\|\hat\bSigma^{\frac{1}{2}}\bmeta_i\|$. At every step of iteration, we record the best candidate found so far as
\begin{align*}
    \bmeta_n^{best} = \argmin \big\{f(\bmeta_i)\,\big|\,\bmeta_i\in Q,\,i\in[n]\big\}.
\end{align*}
Arbitrarily choose $\epsilon>0$. Let $k$ be the iteration number such that after $k$ the best value is $\epsilon$-suboptimal: $f(\bmeta_n^{best})<f(\bmeta^*)+\epsilon$ for $n>k$. Also the best value before $k$ is outside of the $\epsilon$-neighborhood: $f(\bmeta_k^{best})\geq f(\bmeta^*)+\epsilon$. Consequently $f(\bmeta_n)\geq f(\bmeta^*)+\epsilon$ for $n<k$ and $\bmeta_n\in Q$.
\begin{enumerate}
    \item Find a point $\tilde\bmeta$ and a constant $c>0$ such that $f(\tilde\bmeta)\leq f(\bmeta^*)+\epsilon/2,\text{ and } \psi(\tilde\bmeta)\leq-c$.\\
    Such a point $\tilde\bmeta$ can be chosen as $$\tilde\bmeta=(1-\theta)\bmeta^*+\theta\bmeta^{sf},$$
    where $\theta=\min\{1, (\epsilon/2)/(f(\bmeta^{sf})-f(\bmeta^*))\}$. One can see
    \begin{align*}
        f(\tilde\bmeta) & \leq(1-\theta)f(\bmeta^*)+\theta f(\bmeta^{sf}) \leq f(\bmeta^*)+\epsilon/2,\\
        \psi(\tilde\bmeta) & \leq(1-\theta)\psi(\bmeta^*)+\theta\psi(\bmeta^{sf}) \leq \theta\psi(\bmeta^{sf}).
    \end{align*}
    so the constant $c$ can be chosen as $c=-\theta\psi(\bmeta^{sf})$.
    
    \item Show that before $k$, for every iteration $\|\bmeta_{n+1}-\tilde\bmeta\|^2\leq\|\bmeta_n-\tilde\bmeta\|^2-h_n\delta+h_n^2\|\gb_n\|^2$ where $\delta=\min\{\epsilon, 2c\}$.\\
    If $\bmeta_n\in Q$, then $\gb_n=\partial f(\bmeta_n)$, and by the definition of subgradient we have $f(\tilde\bmeta)-f(\bmeta_n) \geq\gb_n\T(\tilde\bmeta-\bmeta_n)$. Since $f(\tilde\bmeta)\leq f(\bmeta^*)+\epsilon/2$ and $f(\bmeta_n)\geq f(\bmeta^*)+\epsilon$, we have $f(\bmeta_n)-f(\tilde\bmeta)\geq\epsilon/2$. Thus
    \begin{align*}
        \|\bmeta_{n+1}-\tilde\bmeta\|^2 & = \|\bmeta_{n}-h_{n}\bg_n-\tilde\bmeta\|^2\\
        & = \|\bmeta_n-\tilde\bmeta\|^2 - 2h_n\gb_n\T(\bmeta_n-\tilde\bmeta) + h_n^2\|\gb_n\|^2\\
        & \leq \|\bmeta_n-\tilde\bmeta\|^2 - 2h_n(f(\bmeta_n)-f(\tilde\bmeta)) + h_n^2\|\gb_n\|^2\\
        & \leq \|\bmeta_n-\tilde\bmeta\|^2 - h_n\epsilon + h_n^2\|\gb_n\|^2.
    \end{align*}
    If $\bmeta_n\notin Q$, then $\gb_n=\partial \psi(\bmeta_n)$, and by the definition of subgradient we have $\psi(\tilde\bmeta)-\psi(\bmeta_n) \geq \gb_n\T(\tilde\bmeta-\bmeta_n)$. Since $\psi(\tilde\bmeta)\leq -c$ and $\psi(\bmeta_n)>0$, we have $\psi(\bmeta_n)-\psi(\tilde\bmeta)\geq c$. Thus
    \begin{align*}
        \|\bmeta_{n+1}-\tilde\bmeta\|^2 & = \|\bmeta_{n}-h_{n}\gb_n-\tilde\bmeta\|^2\\
        & = \|\bmeta_n-\tilde\bmeta\|^2 - 2h_n\gb_n\T(\bmeta_n-\tilde\bmeta) + h_n^2\|\gb_n\|^2\\
        & \leq \|\bmeta_n-\tilde\bmeta\|^2 - 2h_n(\psi(\bmeta_n)-\psi(\tilde\bmeta)) + h_n^2\|\gb_n\|^2\\
        & \leq \|\bmeta_n-\tilde\bmeta\|^2 - 2h_n c + h_n^2\|\gb_n\|^2.
    \end{align*}
    Define $\delta=\min\{\epsilon, 2c\}$ we have
    \begin{align}
        \|\bmeta_{n+1}-\tilde\bmeta\|^2 \leq \|\bmeta_n-\tilde\bmeta\|^2 - h_n\delta + h_n^2\|\gb_n\|^2.
    \label{dist_each_iter}
    \end{align}
    
    \item %Get the order of $k$ when $h_n=1/n$.\\
    Recursively apply \eqref{dist_each_iter} to get
    \begin{align*}
        \|\bmeta_{n+1}-\tilde\bmeta\|^2 \leq \|\bmeta_1-\tilde\bmeta\|^2 - \delta\sum_{n=1}^k h_n + \sum_{n=1}^k h_n^2\|\gb_n\|^2,
    \end{align*}
    so that
    \begin{align*}
        0 \leq C_1^2 - \delta\sum_{n=1}^k h_n + C_2^2\sum_{n=1}^k h_n^2.
    \end{align*}
    When $\epsilon$ is chosen to be small, $\delta$ has the same order as $\epsilon$, since $\delta=\min\{\epsilon, 2c\}$ and $c=-\theta\psi(\bmeta^{sf})=\epsilon\,\frac{\psi(\bmeta^{sf})}{2(f(\bmeta^{sf})-f(\bmeta^*))}$. Thus we have
    \begin{align*}
        \epsilon \lesssim \frac{C_1^2 + C_2^2\sum_{n=1}^k h_n^2}{\sum_{n=1}^k h_n}.
    \end{align*}
\end{enumerate}
\end{proof}

\section{Proofs of Section \ref{asymp:sec}}

%%%%%%%%%%%%%
%%%
%%%%%%%%%%%%%
%\section{}
\begin{proof}[Proof of Theorem \ref{debiase_formula_applicable_unkowncov}]
The debiased estimator $\hat\bbeta_{d}$ is constructed as
\begin{align*}
    \hat\bbeta_{d} = \vb + n^{-1}\hat\bmeta \tilde\Xb\T (\tilde\bY - \tilde\Xb\vb).
\end{align*}
Using simple rearrangements the above can be seen to be equivalent to
\begin{align*}
\sqrt{n}(\hat\bbeta_{d} - \bbeta^*) = \frac{1}{\sqrt{n}}\hat\bmeta \tilde\Xb\T\bvarepsilon +  \sqrt{n}(\hat\bmeta \hat\bSigma - \Ib)(\bbeta^* - \vb).
\end{align*}
If we are interested in the $j$\textsuperscript{th} coefficient $\sqrt{n}(\hat\bbeta_{d}^{(j)} - \bbeta^{*(j)})$ we can multiply the above by $\eb^{(j)\top} = (0,\ldots,\underbrace{1}_j,\ldots,0)$ to obtain
\begin{align}\label{jth:equation}
\sqrt{n}(\hat\bbeta_{d}^{(j)} - \bbeta^{*(j)}) = \frac{1}{\sqrt{n}}\hat\bmeta\T \tilde\Xb\T\bvarepsilon +  \sqrt{n}(\hat \bmeta\T \hat \bSigma - \eb^{(j)\top})( \bbeta^* - \vb).
\end{align}
In \eqref{jth:equation}, we can see the first term is Gaussian conditional on $\overline\Xb, \overline\bY,\tilde\Xb$. The vector $\hat\bmeta$ depends on $\overline\Xb, \overline\bY$ since the constraint of the optimization \eqref{opt_step3} in step 2 involves $\vb$, which is obtained in step 1 and is dependent on $\overline\Xb, \overline\bY$. Since the noise $\bvarepsilon$ is assumed to be normal we have:
\begin{align*}
    Z_j = \frac{1}{\sqrt{n}}\hat\bmeta\T \tilde\Xb\T\bvarepsilon |\overline\Xb, \overline\bY, \tilde\Xb \sim N(0, \sigma^2\hat\bmeta\T\hat\bSigma\hat\bmeta).
\end{align*}
One can see that the solution of the optimization program \eqref{opt_step3} minimizes the variance of the first term in \eqref{jth:equation}. Next, we would like the second term in \eqref{jth:equation} to converge to zero in order to achieve the asymptotic distribution of the debiased coefficient. Notice that the vector $\frac{\bbeta^* - \vb}{\|\bbeta^* - \vb\|} \in \cT_K(\vb)\cap\mathbb{S}^{p-1}$, so the second term $\Delta_j$ can be bounded as
\begin{align}
\label{2nd_term_bound1}
|\Delta_j| = |\sqrt{n}(\hat \bmeta\T \hat \bSigma - \eb^{(j)\top})( \bbeta^* - \vb)| \leq \sqrt{n}\sup_{\ub \in \cT_K(\vb) \cap\mathbb{S}^{p-1}} |(\hat \bmeta\T \hat \bSigma - \eb^{(j)\top}) \ub| \|\vb - \bbeta^*\|.
\end{align}
Since $\hat \bmeta$ is chosen so that the constraint in \eqref{opt_step3} is satisfied, the above will be at most
\begin{align*}
    \sqrt{n} \frac{\rho\overline w(T_K(\vb)\cap\mathbb{S}^{p-1})}{\sqrt{n}} \|\vb - \bbeta^*\|.
\end{align*}
Since $\overline w(\cT_K(\vb)\cap \mathbb{S}^{p-1})\|\vb-\bbeta^*\|=o_p(1)$ as required in step 1, we have $\Delta_j=o_p(1)$.
\end{proof}

%%%%%%%%%
%%%
%%%%%%%%%
%\section{}
\begin{proof}[Proof of Theorem \ref{sigma_hat_rate}]
By the triangle inequality we have
\begin{align*}
\bigg|\frac{1}{n}\sum_{i \in [n]} (Y_i - \bX_i\T \hat \bbeta)^2 - \sigma^2\bigg| \leq \bigg| \frac{1}{n}\sum_{i \in [n]} \varepsilon_i^2 - \sigma^2 \bigg| + \bigg|\frac{1}{n}\sum_{i \in [n]} (Y_i - \bX_i\T \hat \bbeta)^2 - \frac{1}{n}\sum_{i \in [n]} \varepsilon_i^2\bigg|.
\end{align*}
Let $T_n = \frac{\sqrt{n}}{\sqrt{\Var(\epsilon_i^2)}}\bigg(\frac{1}{n}\sum_{i \in [n]} \varepsilon_i^2 - \sigma^2\bigg)$. Notice that $T_n$ converges to a standard normal distribution by central limit theorem. 
Suppose $\EE \epsilon^6 < +\infty$. Let $\rho = \frac{\EE|\epsilon_i^2-\sigma^2|^3}{\Var(\epsilon_i^2)^3}$, and $z\sim N(0,1)$. By the Berry-Esseen central limit theorem \cite[Theorem 2.1.3]{vershynin2018high}, we have
\begin{align*}
    & \Big| \PP\{T_n>\delta\} - \PP\{z>\delta\} \Big| \leq \frac{\rho}{\sqrt{n}}\\
    \Rightarrow\quad & \PP\{T_n>\delta\} \leq \PP\{z>\delta\} + \frac{\rho}{\sqrt{n}}.
\end{align*}
By a tail bound of a standard normal random variable \cite[Example 2.1]{wainwright2019high}, the above inequality can be written as
\begin{align*}
    \PP\{T_n>\delta\} \leq e^{\frac{-\delta^2}{2}} + \frac{\rho}{\sqrt{n}}.
\end{align*}
Thus plug in $T_n = \frac{\sqrt{n}}{\sqrt{\Var(\epsilon_i^2)}}\bigg(\frac{1}{n}\sum_{i \in [n]} \varepsilon_i^2 - \sigma^2\bigg)$ we get
\begin{align}
\label{pf3.1_part1}
    \PP\Big\{\frac{1}{n}\sum_{i \in [n]} \varepsilon_i^2 - \sigma^2 \geq \frac{\sqrt{\Var(\epsilon_i^2)}\,\delta}{\sqrt{n}}\Big\} \leq e^{\frac{-\delta^2}{2}} + \frac{\rho}{\sqrt{n}}.
\end{align}
The second term can be bounded as
\begin{align*}
\bigg|\frac{1}{n}\sum_{i \in [n]} (Y_i - \bX_i\T \hat \bbeta)^2 - \frac{1}{n}\sum_{i \in [n]} \varepsilon_i^2\bigg| & = \bigg| \frac{1}{n}\sum_{i \in[n]}\bigg((\bX_i\T\bbeta^* - \bX_i\T \hat \bbeta + \epsilon_i)^2-\epsilon_i^2\bigg) \bigg|\\
& = \bigg| \frac{1}{n}\sum_{i \in[n]}\bigg((\bX_i\T\bbeta^* - \bX_i\T\hat\bbeta)^2-2(\bX_i\T\bbeta^* - \bX_i\T\hat\bbeta)\epsilon_i\bigg) \bigg|\\
& \leq  \frac{1}{n} \|\Xb(\hat \bbeta - \bbeta^*)\|^2 + \frac{2}{n} \|\Xb (\hat \bbeta - \bbeta^*)\| \|\bvarepsilon\|.
\end{align*}
Since we have $\frac{1}{\sqrt{n}}\|\Xb(\hat \bbeta - \bbeta^*)\|\lesssim\frac{\sigma\delta}{\sqrt{n}}$, and $\frac{\|\epsilon\|}{\sqrt{n}}$ can be bounded by $\sqrt{2}\sigma$ according to \eqref{chebyshev_sigma}, so that
\begin{align*}
    \frac{1}{n}\|\Xb(\hat \bbeta - \bbeta^*)\|^2 \lesssim \frac{\sigma^2\delta^2}{n},
    \quad \text{and  }\frac{2}{n}\|\Xb(\hat \bbeta - \bbeta^*)\|\|\epsilon\| \lesssim \frac{\sigma^2\delta}{\sqrt{n}}.
\end{align*}
By the fact $\delta=o(\sqrt{n})$, we have $\delta^2/n\leq\delta/\sqrt{n}$. Thus with probability converging to one we have
\begin{align}
\label{pf3.1_part2}
    \bigg|\frac{1}{n}\sum_{i \in [n]} (Y_i - \bX_i\T \hat \bbeta)^2 - \frac{1}{n}\sum_{i \in [n]} \varepsilon_i^2\bigg| & \lesssim \frac{\sigma^2\delta}{\sqrt{n}}.
\end{align}
Combine \eqref{pf3.1_part1} and \eqref{pf3.1_part2}, with probability converging to one
\begin{align*}
    |\hat\sigma^2-\sigma^2| \lesssim \frac{(\sqrt{\Var(\epsilon_i^2)}\vee\sigma^2)\,\delta}{\sqrt{n}}.
\end{align*}
\end{proof}

\section{Proofs of Section \ref{cvs_ls:sec}}
%%%%%%%%%
%%%
%%%%%%%%%
%\section{}
\begin{proof}[Proof of Theorem \ref{main_rst_cvs_ls}]
In the optimization program \eqref{step2_cvs_ls}, $\vb$ is the minima, so by the fact $\vb'\in K$ is a feasible point, we have 
\begin{align*}
\|\hat \bbeta - \vb\| & \leq \|\hat \bbeta - \vb'\| +  \frac{\overline w(T_K(\vb')\cap\mathbb{S}^{p-1})}{\sqrt{n}} -  \frac{\overline w(T_K(\vb)\cap\mathbb{S}^{p-1})}{\sqrt{n}} \\
& \leq \|\hat \bbeta - \vb'\| + \frac{\overline w(T_K(\vb')\cap\mathbb{S}^{p-1})}{\sqrt{n}},
\end{align*}
and by triangle inequality 
\begin{align*}
    \|\hat \bbeta - \vb'\| \leq \|\hat \bbeta - \bbeta^*\| + \|\vb' - \bbeta^*\|.
\end{align*}
Plug in $\vb'$ in Lemma \ref{corollary2.6} (and use Remark \ref{remark:after:important:lemma:gaussian:width} after it), to obtain with probability at least $1-e^{-\overline w(\cT_K(\vb')\cap\mathbb{S}^{p-1})}-3e^{-\frac{\big( \overline  w(\cT_K(\vb')\cap\mathbb{S}^{p-1})\big)^2}{2}}-\frac{\Var(\epsilon_i^2)}{n\sigma^4}$ we have
$$\|\bSigma^{1/2}(\hat \bbeta - \bbeta^*)\| \lesssim \|\bSigma^{1/2}(\vb' - \bbeta^*)\| +  \frac{\sigma \overline w(\bSigma^{1/2}(T_K(\vb')\cap\mathbb{S}^{p-1}))}{\sqrt{n}}.$$
By Lemma \ref{connect_gs_comp}, Remark 1.7 of \cite{plan2016generalized} and the fact that $\bSigma$ has bounded spectrum we conclude that 
$$\|\hat \bbeta - \bbeta^*\| \lesssim \|\vb' - \bbeta^*\| +  \frac{\sigma\overline w(T_K(\vb')\cap\mathbb{S}^{p-1})}{\sqrt{n}},$$
so that 
\begin{align*}
    \|\hat \bbeta - \vb'\| \lesssim \|\vb' - \bbeta^*\| +  \frac{\sigma\overline w(T_K(\vb')\cap\mathbb{S}^{p-1})}{\sqrt{n}}
    \text{  and  }
    \|\hat \bbeta - \vb\| \lesssim \|\vb' - \bbeta^*\| +  \frac{(\sigma+1)\overline w(T_K(\vb')\cap\mathbb{S}^{p-1})}{\sqrt{n}}.
\end{align*}
Again by triangle inequality
\begin{align*}
    \|\vb - \bbeta^*\| & \leq \|\vb - \hat\bbeta\| + \|\hat \bbeta - \bbeta^*\|\\
    & \lesssim \|\vb' - \bbeta^*\| +  \frac{(\sigma+1)\overline w(T_K(\vb')\cap\mathbb{S}^{p-1})}{\sqrt{n}}.
\end{align*}
Obviously the order of $\overline w(T_K(\vb)\cap\mathbb{S}^{p-1})$ is also controlled by $\|\vb' - \bbeta^*\|$ and $\overline w(T_K(\vb')\cap\mathbb{S}^{p-1})$ since
\begin{align*}
     \frac{\overline w(T_K(\vb)\cap\mathbb{S}^{p-1})}{\sqrt{n}} & \leq \|\hat \bbeta - \vb'\| +  \frac{\overline w(T_K(\vb')\cap\mathbb{S}^{p-1})}{\sqrt{n}} - \|\hat \bbeta - \vb\|\\
     & \leq \|\hat \bbeta - \vb'\| +  \frac{\overline w(T_K(\vb')\cap\mathbb{S}^{p-1})}{\sqrt{n}}\\
     & \lesssim \|\vb' - \bbeta^*\| +  \frac{(\sigma+1)\overline w(T_K(\vb')\cap\mathbb{S}^{p-1})}{\sqrt{n}}.
\end{align*}
Finally
\begin{align*}
    \overline w(\cT_K(\vb)\cap \mathbb{S}^{p-1})\|\vb-\bbeta^*\| & \lesssim \frac{1}{\sqrt{n}}\Big[\sqrt{n}\|\vb' - \bbeta^*\| + (\sigma+1)\overline w(T_K(\vb')\cap\mathbb{S}^{p-1})\big]^2\\
    & \lesssim \sqrt{n}\|\vb'-\bbeta^*\|^2\vee\frac{(\sigma+1)^2\overline w^2(T_K(\vb')\cap\mathbb{S}^{p-1})}{\sqrt{n}}.
\end{align*}
According to the condition of $\|\vb'-\bbeta^*\|$ and $\overline w(T_K(\vb')\cap\mathbb{S}^{p-1})$, with probability at least $1-e^{-\overline  w(\cT_K(\vb')\cap\mathbb{S}^{p-1})}-3e^{-\frac{\big( \overline  w(\cT_K(\vb')\cap\mathbb{S}^{p-1}\big)^2}{2}}-\frac{\Var(\epsilon_i^2)}{n\sigma^4}$
\begin{align*}
    \overline w(\cT_K(\vb)\cap \mathbb{S}^{p-1})\|\vb-\bbeta^*\| = o_p(1).
\end{align*}
\end{proof}

%%%%%%%%%
%%%
%%%%%%%%%
%\section{}
\begin{proof}[Proof of Lemma \ref{sigma_hat_cvs_ls}]
By an intermediate result \eqref{bound_X(beta-v)} in the proof of Lemma \ref{corollary2.6}, with probability $1-e^{-t}-e^{-\frac{t^2}{2}}$ we have
\begin{align*}
    \frac{1}{\sqrt{n}}\|\Xb(\hat\bbeta - \vb')\| \leq \frac{4}{\sqrt{n}}\|\Xb(\vb' - \bbeta^*)\| + \frac{\sqrt{2}(w(\bSigma^{\frac{1}{2}}\cT_K(\vb')\cap\mathbb{S}^{p-1})+\sqrt{2t})\frac{\|\bvarepsilon\|}{\sqrt{n}}}{\sqrt{n-1}-w(\bSigma^{\frac{1}{2}}\cT_K(\vb')\cap\mathbb{S}^{p-1})-t}.
\end{align*}
Set $t = w(\bSigma^{\frac{1}{2}}\cT_K(\vb')\cap\mathbb{S}^{p-1})$, and $\frac{\|\bvarepsilon\|}{\sqrt{n}}$ can be bounded by $\sqrt{2}\sigma$ according to \eqref{chebyshev_sigma}. The above inequality becomes
\begin{align*}
    \frac{1}{\sqrt{n}}\|\Xb(\hat\bbeta - \vb')\| \leq \frac{4}{\sqrt{n}}\|\Xb(\vb' - \bbeta^*)\| + \frac{w(\bSigma^{\frac{1}{2}}\cT_K(\vb')\cap\mathbb{S}^{p-1})\sigma}{\sqrt{n}},
\end{align*}
and by triangle inequality
\begin{align*}
    \frac{1}{\sqrt{n}}\|\Xb(\hat \bbeta - \bbeta^*)\| & \leq \frac{1}{\sqrt{n}}\|\Xb(\hat \bbeta - \vb')\| + \frac{1}{\sqrt{n}}\|\Xb(\vb' - \bbeta^*)\| \\
    & \leq \frac{5}{\sqrt{n}}\|\Xb(\vb' - \bbeta^*)\| + \frac{w(\bSigma^{\frac{1}{2}}\cT_K(\vb')\cap\mathbb{S}^{p-1})\sigma}{\sqrt{n}}.
\end{align*}

Now what's left is to bound $\frac{1}{\sqrt{n}}\|\Xb(\vb' - \bbeta^*)\|$. For the Gaussian case $\bX_i\sim N(0, \bSigma)$, we can rewrite it as
\begin{align*}
    \|\Xb(\bbeta^*-\vb')\| & = \Big\| \Xb\bSigma^{-\frac{1}{2}}\frac{\bSigma^{\frac{1}{2}}(\bbeta^*-\vb')}{\|\bSigma^{\frac{1}{2}}(\bbeta^*-\vb')\|} \Big\| \|\bSigma^{\frac{1}{2}}(\bbeta^*-\vb')\|.
\end{align*}
By Gordon's escape through mesh (Lemma \ref{gordon_escape}), since $\frac{\bSigma^{\frac{1}{2}}(\bbeta^*-\vb')}{\|\bSigma^{\frac{1}{2}}(\bbeta^*-\vb')\|} \in\bSigma^{\frac{1}{2}}\cT_K(\vb')\cap\mathbb{S}^{p-1}$ with probability at least $1-e^{w^2(\bSigma^{\frac{1}{2}}\cT_K(\vb')\cap\mathbb{S}^{p-1})/2}$ we have
\begin{align*}
    \Big\| \Xb\bSigma^{-\frac{1}{2}}\frac{\bSigma^{\frac{1}{2}}(\bbeta^*-\vb')}{\|\bSigma^{\frac{1}{2}}(\bbeta^*-\vb')\|} \Big\| & \leq \sup_{\ub\in\bSigma^{\frac{1}{2}}\cT_K(\vb')\cap\mathbb{S}^{p-1}}\|\Xb\ub\|\\
    & \leq \sqrt{n} + 2w(\bSigma^{\frac{1}{2}}\cT_K(\vb')\cap\mathbb{S}^{p-1}).
\end{align*}
Thus
\begin{align*}
    \frac{1}{\sqrt{n}}\|\Xb(\bbeta^*-\vb')\| & \leq \frac{\sqrt{n} + 2w(\bSigma^{\frac{1}{2}}\cT_K(\vb')\cap\mathbb{S}^{p-1})}{\sqrt{n}}\|\bSigma^{\frac{1}{2}}\|_{\operatorname{op}}\|\vb'-\bbeta^*\|\\
    & \lesssim \|\vb'-\bbeta^*\|,
\end{align*}
consequently
\begin{align*}
    \frac{1}{\sqrt{n}}\|\Xb(\hat \bbeta - \bbeta^*)\| & \lesssim \|\vb'-\bbeta^*\| + \frac{w(\bSigma^{\frac{1}{2}}\cT_K(\vb')\cap\mathbb{S}^{p-1})\sigma}{\sqrt{n}}.
\end{align*}
By the fact $ w(\bSigma^{\frac{1}{2}}T_K(\vb')\cap\mathbb{S}^{p-1}) \leq \|\bSigma^{-\frac{1}{2}}\|_{\operatorname{op}}\|\bSigma^{\frac{1}{2}}\|_{\operatorname{op}} \overline w(T_K(\vb')\cap\mathbb{S}^{p-1})$ \citep[Remark 1.7]{plan2016generalized}, and $\bSigma$ has bounded eigenvalues we have
\begin{align*}
    \frac{1}{\sqrt{n}}\|\Xb(\hat \bbeta - \bbeta^*)\| \lesssim \|\vb'-\bbeta^*\| + \frac{\overline w(\cT_K(\vb')\cap\mathbb{S}^{p-1})\sigma}{\sqrt{n}},
\end{align*}
and
\begin{align*}
    \delta \asymp \frac{\sqrt{n}}{\sigma}\|\vb'-\bbeta^*\| + \overline w(\cT_K(\vb')\cap\mathbb{S}^{p-1}).
\end{align*}
To show that $\delta=o(\sqrt{n})$, since $\|\vb'-\bbeta^*\|^2=o(1/\sqrt{n})$ and $\sigma$ is finite, the first term is $o(\sqrt{n})$. The second term is $o(\sqrt{n})$ by the given condition.
\end{proof}

%%%%%%%%%
%% Known Cov
%%%%%%%%
%\section{}
\begin{proof}[Proof of Lemma \ref{debiase_formula_applicable}]
Using simple rearrangement the equation
\begin{align*}
    \hat\bbeta_d = \hat\bbeta + n^{-1}\bSigma^{-1} \tilde\Xb\T (\tilde\bY - \tilde\Xb\hat \bbeta),
\end{align*}
can be seen to be equivalent to
\begin{align}
\label{debiase_formula_split}
\sqrt{n}(\hat\bbeta_d - \bbeta^*) = \frac{1}{\sqrt{n}}\bSigma^{-1}\tilde\Xb\T\bvarepsilon +  \sqrt{n}(\bSigma^{-1} \hat\bSigma - \Ib)(\bbeta^* - \hat\bbeta).
\end{align}
The first term is Gaussian condition on $\tilde\Xb$:
$$Z = \frac{1}{\sqrt{n}}\bSigma^{-1}\tilde\Xb\T\bvarepsilon |\tilde\Xb \sim N(0, \sigma^2\bSigma^{-1}\hat\bSigma\bSigma^{-1}).$$
What remains to show is that the second term in \eqref{debiase_formula_split} converges to zero with high probability. Let $\ub = \bbeta^* - \hat\bbeta$, and $\eb^{(j)\top} = (0,\ldots,\underbrace{1}_j,\ldots,0)$. The $j$\textsuperscript{th} coordinate of the second term can be written as
\begin{align*}
    \sqrt{n}(\eb^{(j)\top}\bSigma^{-1} \hat\bSigma - \eb^{(j)\top})\ub & = \frac{1}{\sqrt{n}}\sum_{i=1}^n (\eb^{(j)\top}\bSigma^{-1}\tilde\bX_i\tilde\bX_i\T\ub - u_j).
\end{align*}
Let $g_i = \eb^{(j)\top}\bSigma^{-1}\tilde\bX_i\tilde\bX_i\T\ub - u_j$. Notice that $(\tilde\Xb, \tilde\bY)$ is independent from $(\overline\Xb, \overline\bY)$, and $\ub$ is constant conditionally on $(\overline\Xb, \overline\bY)$, so $\EE(g_i|\overline\Xb, \overline\bY) = \eb^{(j)\top}\bSigma^{-1}\bSigma\ub - u_j = 0$. Moreover, $\eb^{(j)\top}\bSigma^{-1}\tilde\bX_i$ and $\tilde\bX_i\T\ub$ are Gaussian random variables condition on $(\overline\Xb, \overline\bY)$. Let $\|\cdot\|_{\psi_2}$ be the sub-gaussian norm defined in \cite[Definition 2.5.6]{vershynin2018high}. The sub-gaussian norm of a Gaussian random variable is up to a constant of its standard deviation \cite[Example 2.5.8]{vershynin2018high}, so we have
\begin{align*}
    \|\eb^{(j)\top}\bSigma^{-1}\tilde\bX_i\|_{\psi_2} & \leq C_1 \|\bSigma^{-\frac{1}{2}}\eb^{(j)}\|\\
    \|\tilde\bX_i\T\ub\|_{\psi_2} & \leq C_2 \|\bSigma^{\frac{1}{2}}\ub\|.
\end{align*}
Let $\|\cdot\|_{\psi_1}$ be the sub-exponential norm defined in \cite[Definition 2.7.5]{vershynin2018high}. The product of two sub-gaussian random variables is a sub-exponential random variable, and the corresponding sub-exponential norm is less than the product of sub-Gaussian norms \cite[Lemma 2.7.7]{vershynin2018high}. Thus
\begin{align*}
    \|\eb^{(j)\top}\bSigma^{-1}\tilde\bX_i\tilde\bX_i\T\ub\|_{\psi_1} & \leq \|\eb^{(j)\top}\bSigma^{-1}\tilde\bX_i\|_{\psi_2}\|\tilde\bX_i\T\ub\|_{\psi_2}\\
    & \leq C_1 C_2 \|\bSigma^{-\frac{1}{2}}\|_{\operatorname{op}} \|\bSigma^{\frac{1}{2}}\ub\|.
\end{align*}
Additionally, the sub-exponential norm of a centered sub-exponential random variable is up to a constant to the original one \cite[Exercise 2.7.10]{vershynin2018high}
\begin{align*}
    \|\eb^{(j)\top}\bSigma^{-1}\tilde\bX_i\tilde\bX_i\T\ub - u_j\|_{\psi_1} & \leq C_3 \|\eb^{(j)\top}\bSigma^{-1}\tilde\bX_i\tilde\bX_i\T\ub\|_{\psi_1}\\
    & \leq C_1 C_2 C_3 \|\bSigma^{-\frac{1}{2}}\|_{\operatorname{op}} \|\bSigma^{\frac{1}{2}}\ub\|.
\end{align*}
Let $C = C_1 C_2 C_3$. Given the sub-exponential norm of $g_i=\eb^{(j)\top}\bSigma^{-1}\tilde\bX_i\tilde\bX_i\T\ub - u_j$, use Bernstein's inequality \cite[Theorem 2.8.1]{vershynin2018high} to get the conditional concentration inequality
\begin{align*}
    \PP\bigg( \Big|\frac{1}{\sqrt{n}}\sum_{i=1}^n g_i\Big| \geq t \,\bigg|\, \overline\bX, \overline\bY \bigg) \leq 2 \exp{\Big[ -c\min\Big( \frac{t^2}{C^2\|\bSigma^{-\frac{1}{2}}\|_{\operatorname{op}}^2\|\bSigma^{\frac{1}{2}}\ub\|^2}, \frac{t\sqrt{n}}{C\|\bSigma^{-\frac{1}{2}}\|_{\operatorname{op}}\|\bSigma^{\frac{1}{2}}\ub\|}\Big) \Big]}.
\end{align*}
The unconditional concentration inequality can be obtained by
{\footnotesize
\begin{align*}
    \PP\bigg( \Big|\frac{1}{\sqrt{n}}\sum_{i=1}^n g_i\Big| \geq t \bigg) & = \int \PP\bigg( \Big|\frac{1}{\sqrt{n}}\sum_{i=1}^n g_i\Big| \geq t \,\bigg|\, \overline\bX, \overline\bY \bigg) d\mu(\overline\bX, \overline\bY) \\
    & \leq \int 2 \exp{\Big[ -c\min\Big( \frac{t^2}{C^2\|\bSigma^{-\frac{1}{2}}\|_{\operatorname{op}}^2\|\bSigma^{\frac{1}{2}}\ub\|^2}, \frac{t\sqrt{n}}{C\|\bSigma^{-\frac{1}{2}}\|_{\operatorname{op}}\|\bSigma^{\frac{1}{2}}\ub\|}\Big) \Big]} d\mu(\overline\bX, \overline\bY)\\
    & =  \int_{\|\bSigma^{\frac{1}{2}}\ub\|\leq\theta} 2 \exp{\Big[ -c\min\Big( \frac{t^2}{C^2\|\bSigma^{-\frac{1}{2}}\|_{\operatorname{op}}^2\|\bSigma^{\frac{1}{2}}\ub\|^2}, \frac{t\sqrt{n}}{C\|\bSigma^{-\frac{1}{2}}\|_{\operatorname{op}}\|\bSigma^{\frac{1}{2}}\ub\|}\Big) \Big]} d\mu(\overline\bX, \overline\bY) +\\ 
    & \quad \mathbb{P}\Big[ \|\bSigma^{\frac{1}{2}}\ub\| > \theta \Big].
\end{align*}
}
The threshold $\theta = \|\bSigma^{\frac{1}{2}}(\vb'-\bbeta^*)\|+\frac{\sigma  w(\bSigma^{\frac{1}{2}}\cT_K(\vb')\cap\mathbb{S}^{p-1})}{\sqrt{n}}$ is chosen according to the result of Lemma \ref{corollary2.6} in order to make the second term vanish. Apply Lemma \ref{corollary2.6} with $\vb=\vb'$, one can see the second term of RHS vanishes as $n\rightarrow\infty$. 

For the first term, take $t=\theta\|\bSigma^{-\frac{1}{2}}\|_{\operatorname{op}}\,a_n$, we can see that $\Big|\frac{1}{\sqrt{n}}\sum_{i=1}^n g_i\Big|$ is bounded as
\begin{align*}
    \PP\bigg( \Big|\frac{1}{\sqrt{n}}\sum_{i=1}^n g_i\Big| \geq \theta\|\bSigma^{-\frac{1}{2}}\|_{\operatorname{op}}\,a_n \bigg) 
    \leq 2 \exp{\Big[ -c\min\Big( \frac{a_n^2}{C^2}, \frac{\sqrt{n}a_n}{C}
    \Big)\Big]},
\end{align*}
where $a_n$ is picked such that $\theta\|\bSigma^{-\frac{1}{2}}\|_{\operatorname{op}}\,a_n=o(1)$ and $a_n\rightarrow\infty$. Specifically we have
\begin{align*}
    \|\bSigma^{\frac{1}{2}}(\vb'-\bbeta^*)\|\,a_n = o(1), \,\,  w(\bSigma^{\frac{1}{2}}\cT_K(\vb')\cap\mathbb{S}^{p-1})\,a_n = o(\sqrt{n}), \,\,a_n\rightarrow\infty.
\end{align*}
The first condition reduces to $\|\vb'-\bbeta^*\|a_n=o(1)$ since $\lambda_{\min}(\bSigma^{1/2})\|\vb'-\bbeta^*\| \leq \|\bSigma^{\frac{1}{2}}(\vb'-\bbeta^*)\| \leq \|\bSigma^{1/2}\|_{\operatorname{op}}\|\vb'-\bbeta^*\|$. The condition $ w(\bSigma^{\frac{1}{2}}\cT_K(\vb')\cap\mathbb{S}^{p-1})a_n=o(\sqrt{n})$ reduces to $\overline w(\cT_K(\vb')\cap\mathbb{S}^{p-1})a_n=o(\sqrt{n})$ by the fact $ w(\bSigma^{\frac{1}{2}}T_K(\vb')\cap\mathbb{S}^{p-1}) \leq \|\bSigma^{-\frac{1}{2}}\|_{\operatorname{op}}\|\bSigma^{\frac{1}{2}}\|_{\operatorname{op}} \overline w(T_K(\vb')\cap\mathbb{S}^{p-1})$ \citep[Remark 1.7]{plan2016generalized}.
\end{proof}

%%%%%%%%%%
%%
%%%%%%%%%
\begin{proof}[Proof of Proposition \ref{decomp_tan_pos}]
By definition, $\cT(M^{p+}(\vb)) = \{\ub - t\vb:\, t\geq0,\,\ub\in M^{p+}\}$. 
If $\vb$ is a non-zero constant, it is trivial that $\cT(M^{p+}(\vb)) = M^{p}$. Moreover if all the coordinates of $\vb$ are zeros, the positiveness is also preserved so that $\cT(M^{p+}(\vb)) = M^{p+}$. Now it is sufficient to consider the case where $\vb$ has at least two constant pieces. 

Firstly, suppose the first constant piece of $\vb$ doesn't consist of zeros. Within each constant piece, the monotonicity of $u_i - tv_i$ is preserved, but not necessarily the positiveness, so that $\cT(M^{p+}(\vb)) \subset M^{p_1} \times M^{p_2} \times \ldots \times M^{p_{l}}$. To show the other direction, arbitrarily choose $\xb\in M^{p_1} \times M^{p_2} \times \ldots \times M^{p_{l}}$. Let $\epsilon_1 = \min_{i\in S} (v_{i+1}-v_i)$, where $S = \{i : v_{i + 1} > v_i\}$ and $\epsilon_2 = 2\min_{i\in[p]}v_i$. Pick $t = \frac{2\|\xb\|_{\infty}}{\epsilon_1\wedge\epsilon_2}$, then for all $i\in[p]$ we have
\begin{align*}
    x_i + tv_i & \geq x_i + \|x\|_{\infty} \geq 0,
\end{align*}
and for $i \in S$:
\begin{align*}
     t(v_{i+1}-v_i) = \frac{2\|\xb\|_{\infty}}{\epsilon_1\wedge\epsilon_2}(v_{i+1}-v_i) \geq x_i - x_{i+1} \quad \Rightarrow \quad x_i + tv_i \leq x_{i+1} + tv_{i+1}.
\end{align*}
For $i \in [p-1]\setminus S$ we have $v_{i + 1} = v_i$ and $x_{i} \leq x_{i + 1}$ so that $x_i + tv_i \leq x_{i+1} + tv_{i+1}$ also holds. Thus for any $\xb\in M^{p_1} \times M^{p_2} \times \ldots \times M^{p_{l}}$ there is a $t$ such that $\xb + t\vb\in M^{p+}$. The direction $\cT(M^{p+}(\vb)) \supset M^{p_1} \times M^{p_2} \times \ldots \times M^{p_{l}}$ holds.

When the first constant piece of $\vb$ is zero valued, within it $u_i - tv_i = u_i$ is always positive and monotone. For the other constant pieces, $u_i - tv_i$ is still monotone, so that $\cT(M^{p+}(\vb)) \subset M^{p_1+} \times M^{p_2} \times \ldots \times M^{p_{l}}$. For the other direction, let $\epsilon_1 = \min_{i\in S} (v_{i+1}-v_i)$, and $\epsilon_2$ be two times the minimum non-zero $v_i$. Also let $t = \frac{2\|\xb\|_{\infty}}{\epsilon_1\wedge\epsilon_2}$. it is easy to verify that $\xb + t\vb\in M^{p+}$. 
\end{proof}

%%%%%%%%%%%%%
%%
%%%%%%%%%%%%
%\section{}
\begin{proof}[Proof of Lemma \ref{beta_proj_onto_s}]
By definition, $\vb_s \in \argmin_{\wb\in T}\|\wb - \hat \bbeta\|$. For brevity let $\vb'$ be any vector in $\argmin_{\wb\in T}\|\wb - \hat \bbeta\|$. First for each coordinate of $\vb'$, we have either sign$(\vb'_{(i)})\,=\,$sign$(\hat \bbeta_{(i)})$, or sign$(\vb_{(i)})\,=0$, because otherwise we can always reverse the sign to make the $\ell_2$-norm of difference $\|\vb' - \hat \bbeta\|$ smaller.

%Then we show that $\vb_s$ has the same non-zero indices as $\bbeta^*$, and the values are as we stated before. 
Fix a set $S'$ of $s$ coordinates which is the assumed support for the vector $\vb_{(i)}'$. Consider the following optimization problem 
\begin{align}\label{unrelaxed:problem}
    \min_{\vb'}\,\, \sum_{i\in S'}(|\hat \bbeta_{(i)}| - |\vb'_{(i)}|)^2 + \sum_{i\notin S'} \hat \bbeta_{(i)}^2
    \quad \text{subject to} \quad
    \sum_{i\in S'} |\vb'_{(i)}| = \|\bbeta^*\|_1.
\end{align}
Relax this to the following problem which can potentially get a smaller objective function value
\begin{align*}
    \min_{\vb'}\,\, \sum_{i\in S'}(|\hat \bbeta_{(i)}| - a_i)^2 + \sum_{i\notin S'} \hat \bbeta_{(i)}^2
    \quad \text{subject to} \quad
    \sum_{i\in S'} a_i = \|\bbeta^*\|_1,
\end{align*}
where $a_i \in \RR$ (here we lose the positivity of $a_i$ from problem \eqref{unrelaxed:problem}).

Use Lagrange multipliers we obtain the Lagrangian
\begin{align*}
    L = \sum_{i\in S'}(|\hat\bbeta_{(i)}| - a_i)^2 + \sum_{i\notin S'} \hat \bbeta_{(i)}^2 + \lambda\big(\sum_{i\in S'} a_i- \|\bbeta^*\|_1\big),
\end{align*}
and solve $\frac{\partial L}{\partial a_i} = 0$ to get
$$a_i = |\hat \bbeta_{(i)}| + \lambda \text{ for all } i\in S'.
$$
Combine it with the fact that $\sum_{i\in S'} a_i = \|\bbeta^*\|_1$, we have 
$$\lambda = \frac{\|\bbeta^*\|_1 - \sum_{i\in S'}|\hat \bbeta_i|}{s} > 0,$$
where the last inequality follows since $\|\bbeta^*\|_1 \geq \|\hat \bbeta\|_1$. It follows that $a_i \geq 0$, and thus the minimum for problem \eqref{unrelaxed:problem} is also achieved at the same point. 
Hence at the optimal point we have $\|\vb' - \hat \bbeta\| = \sqrt{s\lambda^2 + \sum_{i \not \in S'} \hat \bbeta_{(i)}^2}$. Note that when $S' = S$ is the set of indices of the $s$ most significant coordinates both $\lambda$ and $\sum_{i \not \in S'} \hat \bbeta_{(i)}^2$ are minimized. This completes the proof.
\end{proof}

\section{Proofs From Section \ref{minkowski:section}}

\begin{proof}[Proof of Lemma \ref{find:v:minkowski:lemma}]
We first begin by noticing that $\bbeta^*$ will be a feasible point to \eqref{find:v:minkowski}, since 
\begin{align*}
    \bigg(\frac{\bar s \overline w(\Xb K') }{n}\bigg) \bigg(\frac{n}{\bar s^2 \overline w(\Xb K') \overline w(K')}\bigg)^{\gamma} \gg \frac{\sigma \gamma_{K'}(\bbeta^*) \overline w(\Xb K') }{n}
\end{align*} 
(assuming $\sigma, \lambda_{\min}(\bSigma), \lambda_{\max}(\bSigma)$ do not scale with $n$), and by assumption $\bar \gamma_{K'}(\bbeta^*) \leq \bar s$. Hence we will have $\rho_{K'}(\vb) \geq \rho_{K'}(\bbeta^*)$ which ensures $\bbeta^* \in \rho_{K'}(\vb)K'$, and in addition 

\begin{align*}
    \|\vb - \bbeta^*\| w(\overline \cT_{K'}(\vb) \cap \mathbb{S}^{p-1}) \leq  \bigg(\frac{\bar s^2 \overline w(\Xb K') \overline w(K')}{n}\bigg) \bigg(\frac{n}{\bar s^2 \overline w(\Xb K')\overline w(K')}\bigg)^{\gamma} = o_p(1),
\end{align*}
where the last bound follows by the following logic:
We should note here that the quantity $w(\overline \cT_{K'}(\vb) \cap \mathbb{S}^{p-1}) \leq \gamma_{K'}(\vb)\overline w( K')$ is an upper bound on the Gaussian complexity of the tangent cone. This is so because 
    \begin{align*}
        w(\overline \cT_{K'}(\vb) \cap \mathbb{S}^{p-1})& = \EE \sup_{\wb \in \overline \cT_{K'}(\vb), \|\wb\| = 1} \gb\T \wb \leq \EE \sup_{\wb \in \overline \cT_{K'}(\vb), \|\wb\| = 1}  \rho_{K'}^*(\gb) \rho_{K'}(\wb) \\
        &  \leq \EE \rho_{K'}^*(\gb) \gamma_{K'}(\vb) \leq \overline w(K')\gamma_{K'}(\vb).
    \end{align*}
\end{proof}

\begin{proof}[Proof of Lemma \ref{lemma:find:v:trace:reg}]
Let $\vb$ has an SVD given by $\Ub \bLambda \Vb\T$. Using the formula
\begin{align*}
    \|\Ab\|^2_F = \tr(\Ab\T\Ab),
\end{align*}
we obtain that
\begin{align*}
    \|\hat \bbeta - \vb\|_F^2 = \sum_{i = 1}^{R} \lambda_i^2 + \sum_{i = 1}^{\hat r} \hat \lambda_i^2 - 2\sum_{i \in [R],j \in [\hat r]} (\vb_i\T \hat \vb_j) (\ub_i\T \hat \ub_j) \lambda_i \hat \lambda_j
\end{align*}
Above $R\leq \bar r$ denotes the rank of $\vb$ and $\hat r$ denotes the rank of $\hat \bbeta$. First, suppose $\hat r > \bar r$. We note that 
\begin{align*}
    2\sum_{i \in [R],j \in [\hat r]} (\vb_i\T \hat \vb_j) (\ub_i\T \hat \ub_j) \lambda_i \hat \lambda_j\leq \sum_{i \in [R],j \in [\hat r]} ((\vb_i\T \hat \vb_j)^2 +  (\ub_i\T \hat \ub_j)^2) \lambda_i \hat \lambda_j
\end{align*}

Also, it is clear that $\sum_i (\vb_i\T \hat \vb_j)^2 \leq 1$ and $\sum_i (\ub_i\T \hat \ub_j)^2 \leq 1$ and $\sum_j (\vb_i\T \hat \vb_j)^2 \leq 1$ and $\sum_j (\ub_i\T \hat \ub_j)^2 \leq 1$. Since $\hat r \geq R$, we can bound the above as
\begin{align*}
    \sum_{i \in [R],j \in [\hat r]} ((\vb_i\T \hat \vb_j)^2 +  (\ub_i\T \hat \ub_j)^2) \lambda_i \hat \lambda_j & \leq \sum_{i \in [R]} (\sum_{j = R}^{\hat r}(\vb_i\T \hat \vb_j)^2 + \sum_{j = R}^{\hat r}(\ub_i\T \hat \ub_j)^2) \lambda_i \hat \lambda_R \\
    & + \sum_{i \in [R],j < R} ((\vb_i\T \hat \vb_j)^2 +  (\ub_i\T \hat \ub_j)^2) \lambda_i \hat \lambda_j,
\end{align*}
where we are assuming $\hat \lambda_1 \geq \hat \lambda_2 \geq \ldots \geq \hat \lambda_{\hat r}$. Let $\alpha_{ij} = (\vb_i\T \hat \vb_j)^2$ for $i \in [R], j < R$ and $\alpha_{iR} = \sum_{j = R}^{\hat r}(\vb_i\T \hat \vb_j)^2$, and  $\alpha'_{ij} = (\ub_i\T \hat \ub_j)^2$ for $i \in [R], j < R$ and $\alpha'_{iR} = \sum_{j = R}^r(\ub_i\T \hat \ub_j)^2$. We have $\sum_i \alpha_{ij} \leq 1$ and $\sum_j \alpha_{ij} \leq 1$, and similarly for $\alpha_{ij}'$. Hence we can upper bound the summations by letting $\sum_i \alpha_{ij} = 1$ and $\sum_j \alpha_{ij} = 1$, and similarly for $\alpha_{ij}'$. Hence we have a doubly stochastic matrix. As is well known each doubly stochastic matrix is a convex combination of permutation matrices. By the rearrangement inequality it follows that the upper bound is sharpest if the permutation matrix respects the orders of the $\lambda_i$ and $\hat \lambda_i$. Hence, it is most beneficial to have $\Ub$ and $\Vb$ coincide with $\hat \Ub$ and $\hat \Vb$ on the set of the $R$-th largest singular values. Also whatever is left from the Frobenius norm, should be distributed equally to the singular values in order to maximize their sum.

Suppose now $\hat r < \bar r$. If $R < \hat r$ then the same proof as before applies. If $R \geq \hat r$, by reversing the roles of $R$ and $\hat r$ and $i$ and $j$ in the previous proof we realize that it is most beneficial to match the highest $\hat r$ of the singular values $\lambda_i$ with the $\hat \bbeta$ matrix. Whatever is left from the Frobenius norm should be distributed equally to maximize the value of the sum of the singular values.

Following the above strategy results in the following maximal values for each $R$. For $R < \hat r$ we have $\sum_{i = 1}^R \hat \lambda_i + \sqrt{U - \sum_{i = R + 1}^{\hat r} \hat \lambda_i^2}\sqrt{R}$, where $U = \bigg(\frac{2\sqrt{2\bar r} \overline w(\Xb K') }{n}\bigg) \bigg(\frac{n}{(2\sqrt{2\bar r})^2 \overline w(\Xb K') \overline w(K')}\bigg)^{\gamma}$. For $R > \hat r$ (we will be forced to select $R = 
\overline r$) we have $\sum_{i = 1}^{\hat r} \hat \lambda_i + \sqrt{U}\sqrt{\overline r}$. One can see, that in both cases, it is most beneficial to set $R = \overline r$ which completes the proof.
\end{proof}

\begin{proof}[Proof of Proposition \ref{trace:reg:projection:proposition}]
We have
\begin{align*}
    \|\Ab - \Bb\|_F^2 = \tr(\Ab\T \Ab) - 2 \tr(\Ab\T \Bb) + \tr(\Bb\T \Bb).
\end{align*}
We thus need to focus on minimizing 
\begin{align*}
    -2 t \tr(\Vb\T \Ab\T \Ub) - 2\sum_{i}(\tilde \vb_i\T \Ab\T \tilde \ub_i) \sigma_i + t^2 r + \sum_{i = 1}^{\min(p_1,p_2) - r} \sigma_i^2,
\end{align*}
where we wrote $\Bb = t \Ub \Vb\T + \Wb$ and $\Wb = \sum_{i = 1}^{\min(p_1, p_2) - r} \sigma_i \tilde \ub_i \tilde\vb_i\T$, where $\tilde \ub_i \in \col(\Ub)^{\perp}$ and $\tilde \vb_i \in \col(\Vb)^{\perp}$. The above is clearly the same as
\begin{align*}
    -2 t \tr(\Vb\T \Ab\T \Ub) - 2\sum_{i}(\tilde \vb_i\T \Pb_{\Vb^{\perp}}\Ab\T\Pb_{\Ub^{\perp}} \tilde \ub_i) \sigma_i + t^2 r + \sum_{i = 1}^{\min(p_1,p_2) - r} \sigma_i^2,
\end{align*}
Let $(\Pb_{\Vb^{\perp}}\Ab\T\Pb_{\Ub^{\perp}})\T = \sum_{i = 1}^{\min(p_1, p_2) - r} \bar \lambda_i \bar \ub_i \bar \vb_i\T$, for  $\bar \ub_i \in \col(\Ub)^{\perp}$ and $\bar \vb_i \in \col(\Vb)^{\perp}$.

Then the expression 
\begin{align*}
    \sum_{i}(\tilde \vb_i\T \Pb_{\Vb^{\perp}}\Ab\T\Pb_{\Ub^{\perp}} \tilde \ub_i) \sigma_i  = \sum_{i,j} (\tilde \vb_i\T\bar \vb_j)(\tilde \ub_i \bar \ub_i) \bar \lambda_j \sigma_i.
\end{align*}
Using the same logic as in the proof of Lemma \ref{find:v:minkowski:lemma}, we can upper bound this by $\sum_{i} \bar \lambda_i \sigma_i$ where we are assuming that $\bar\lambda_1 \geq \bar\lambda_2 \geq \ldots \bar\lambda_{\min(p_1,p_2) - r}$ and $\sigma_1 \geq \sigma_2 \geq \ldots \sigma_{\min(p_1,p_2) - r}$, and the $\tilde \vb_i, \tilde \ub_i$ must coincide with the corresponding $\bar \vb_i, \bar \ub_i$. It then follows that the projection reduces to the following problem
\begin{align*}
    -2 t \tr(\Vb\T \Ab\T \Ub) - 2\sum_{i}\bar\lambda_i \sigma_i + t^2 r + \sum_{i = 1}^{\min(p_1,p_2) - r} \sigma_i^2,\mbox{s.t. } t \geq \sigma_1 \geq \sigma_2 \geq \ldots \geq \sigma_{\min(p_1,p_2) - r},
\end{align*}
which is clearly equivalent to an isotonic regression problem, and can be solved by PAVA.
\end{proof}

\section{Proofs from Section \ref{minkowski:gauge:regularization:sec}}

\begin{proof}[Proof of Theorem \ref{funamental:theorem:minkowski:reg}] 
%{\color{red} Change all ocurrances of $w$ with $\overline w$!}

Consider solving the constrained optimization problem:
\begin{align*}
\hat \bbeta := \argmin_{\bbeta} n^{-1} \sum_{i \in [n]} (Y_i - \bX_i\T \bbeta)^2, \mbox{ s.t. } \rho_{K'}(\bbeta) \leq  \rho_{K'}(\wb).
\end{align*}
By duality theory we know that there exists a $\lambda'$ of the regularized problem such that the solution $\hat \bbeta_{\lambda'} \equiv \hat \bbeta$. In addition observe that if $\lambda$ and $\bar \lambda$ are two values such that $\lambda < \bar \lambda$, then $\rho_{K'}(\hat \bbeta_{\lambda}) \geq \rho_{K'}(\hat \bbeta_{\bar \lambda})$. This is so because
\begin{align*}
n^{-1} \sum_{i \in [n]} (Y_i - \bX_i\T \hat \bbeta_{\lambda})^2 + \lambda \rho_{K'}(\hat \bbeta_{\lambda}) \leq n^{-1} \sum_{i \in [n]} (Y_i - \bX_i\T \hat \bbeta_{\bar \lambda})^2 + \lambda \rho_{K'}(\hat \bbeta_{\bar \lambda})\\
  n^{-1}\sum_{i \in [n]} (Y_i - \bX_i\T \hat \bbeta_{\bar \lambda})^2 + \bar \lambda \rho_{K'}(\hat \bbeta_{\bar \lambda})\leq n^{-1} \sum_{i \in [n]} (Y_i - \bX_i\T \hat \bbeta_{\lambda})^2 + \bar \lambda \rho_{K'}(\hat \bbeta_{\lambda}),
\end{align*}
from where we have that 
\begin{align*}
\MoveEqLeft \frac{ n^{-1}\sum_{i \in [n]} (Y_i - \bX_i\T \hat \bbeta_{\bar \lambda})^2 -   n^{-1}\sum_{i \in [n]} (Y_i - \bX_i\T \hat \bbeta_{\lambda})^2}{ \lambda} \geq\\& \rho_{K'}(\hat \bbeta_{\lambda}) - \rho_{K'}(\hat \bbeta_{\bar \lambda}) \\
& \geq \frac{  n^{-1}\sum_{i \in [n]} (Y_i - \bX_i\T \hat \bbeta_{\bar \lambda})^2 -   n^{-1}\sum_{i \in [n]} (Y_i - \bX_i\T \hat \bbeta_{\lambda})^2}{\bar \lambda}.
\end{align*}
We conclude that $\rho_{K'}(\hat \bbeta_{\lambda}) \geq \rho_{K'}(\hat \bbeta_{\bar \lambda})$. This implies that for all values of $\lambda > \lambda'$ we will have $\rho_{K'}(\hat \bbeta_{\lambda}) \leq \rho_{K'}(\wb)$. This shows the existence of $\lambda_{\wb}^*$ (by taking the $\inf$ among all $\lambda$ for which $\rho_{K'}(\hat \bbeta_{\lambda}) \leq \rho_{K'}(\wb)$). We will now assume $\lambda \geq \lambda_{\wb}^*$ so that we have $\rho_{K'}(\hat \bbeta_\lambda) \leq \rho_{K'}(\wb)$.

 By the basic inequality
 \begin{align*}
     n^{-1}\sum_{i \in [n]} (Y_i - \bX_i\T \hat \bbeta)^2 + \lambda \rho_{K'}(\hat \bbeta) \leq 2 n^{-1}\sum_{i \in [n]} (Y_i - \bX_i\T  \wb)^2 + \lambda \rho_{K'}( \wb),
 \end{align*}
 we have
 \begin{align*}
     \MoveEqLeft n^{-1}\sum_{i \in [n]} (\bX_i\T (\hat \bbeta - \bbeta^*))^2 -   n^{-1}\sum_{i \in [n]} (\bX_i\T (\wb - \bbeta^*))^2  \\& \leq n^{-1}\sum_{i \in [n]} \varepsilon_i \bX_i\T (\hat \bbeta - \wb) + \lambda \rho_{K'}( \wb) - \lambda \rho_{K'}(\hat \bbeta),
 \end{align*}
 This is equivalent to 
 
\begin{align*}
     \MoveEqLeft n^{-1}\sum_{i \in [n]} (\bX_i\T (\hat \bbeta - \wb))^2  + 2n^{-1}\sum_{i \in [n]} (\bX_i\T (\hat \bbeta - \wb)) (\bX_i\T (\wb - \bbeta^*)) \\
     & \leq n^{-1}\sum_{i \in [n]} \varepsilon_i \bX_i\T (\hat \bbeta - \wb) + \lambda \rho_{K'}( \wb) - \lambda \rho_{K'}(\hat \bbeta),
 \end{align*}

 Using $2ab \geq -1/2 a^2 - 2 b^2$ we can lower bound $2n^{-1}\sum_{i \in [n]} (\bX_i\T (\hat \bbeta - \wb)) (\bX_i\T (\wb - \bbeta^*))  \geq -1/2n^{-1}\sum_{i \in [n]} (\bX_i\T (\hat \bbeta - \wb))^2 -2n^{-1}\sum_{i \in [n]} (\bX_i\T (\wb - \bbeta^*))^2$. It follows that
 \begin{align*}
n^{-1}\sum_{i \in [n]} (\bX_i\T (\hat \bbeta - \wb))^2 \lesssim n^{-1}\sum_{i \in [n]} (\bX_i\T (\wb - \bbeta^*))^2 + n^{-1}\sum_{i \in [n]} \varepsilon_i \bX_i\T (\hat \bbeta - \wb) + \lambda \rho_{K'}( \wb) - \lambda \rho_{K'}(\hat \bbeta).
\end{align*}
 Conditionally on the noise we then have $\sum_{i \in [n]} \varepsilon_i \bX_i \sim N(0, \bSigma \frac{\sum \varepsilon_i^2}{n^2})$. Since $\hat \bbeta - \wb \in \overline \cT_{K'} (\wb)$, it follows that $\frac{\bSigma^{\frac{1}{2}}(\hat \bbeta - \wb)}{\|\bSigma^{\frac{1}{2}}(\hat \bbeta - \wb)\|} \in \bSigma^{\frac{1}{2}}\overline \cT_{K'}(\wb) \cap \mathbb{S}^{p-1}$. Hence
 \begin{align*}
     n^{-1}\sum_{i \in [n]} \varepsilon_i \bX_i\T (\hat \bbeta - \wb) \leq \sqrt{\frac{\sum \varepsilon_i^2}{n^2}} \|\bSigma^{\frac{1}{2}}(\hat \bbeta - \wb)\|\sup_{\ub \in \bSigma^{\frac{1}{2}}\overline \cT_{K'}(\wb) \cap \mathbb{S}^{p-1}}\bZ\T \ub ,
 \end{align*}
 where $\bZ$ is a standard normal vector. Next note that $\EE \sup_{\ub \in \bSigma^{\frac{1}{2}}\overline \cT_{K'}(\wb) \cap \mathbb{S}^{p-1}}\bZ\T \ub = w(\bSigma^{\frac{1}{2}}\overline \cT_{K'}(\wb) \cap \mathbb{S}^{p-1})$ and by the concentration of Lipschitz functions of Gaussian Variables \cite[Theorem 2.26]{wainwright2019high}, we then obtain
 \begin{align}\label{lipschitz:concentration}
 \PP\Big(\sup_{\ub \in \bSigma^{\frac{1}{2}}\overline \cT_{K'}(\wb) \cap \mathbb{S}^{p-1}}\bZ\T \ub - \EE \sup_{\ub \in \bSigma^{\frac{1}{2}}\overline \cT_{K'}(\wb) \cap \mathbb{S}^{p-1}}\bZ\T \ub \geq t\Big) \leq \exp \Big( -\frac{t^2}{2 L^2} \Big),
 \end{align}
 where $L$ is the Lipschitz constant of the map $\bZ \mapsto \sup_{\ub \in \bSigma^{\frac{1}{2}}\overline \cT_{K'}(\wb) \cap \mathbb{S}^{p-1}}\bZ\T \ub$ which is $1$ in this case. Hence, by setting $t = \overline w(\overline \cT_{K'}(\wb) \cap \mathbb{S}^{p-1})$ we conclude that with probability at least $1 - \exp(-\overline w(\overline \cT_{K'}(\wb) \cap \mathbb{S}^{p-1})^2/2)$ we have 
 \begin{align*}
     n^{-1}\sum_{i \in [n]} \varepsilon_i \bX_i\T (\hat \bbeta - \wb) \leq \sqrt{\frac{\sum \varepsilon_i^2}{n^2}} \|\bSigma^{\frac{1}{2}}(\hat \bbeta - \wb)\| (\overline w(\overline \cT_{K'}(\wb) \cap \mathbb{S}^{p-1}) + w(\bSigma^{\frac{1}{2}}\overline \cT_{K'}(\wb) \cap \mathbb{S}^{p-1}))\\
     \lesssim \sqrt{\frac{\sum \varepsilon_i^2}{n^2}} \|\bSigma^{\frac{1}{2}}(\hat \bbeta - \wb)\| \overline w(\overline \cT_{K'}(\wb) \cap \mathbb{S}^{p-1}),
 \end{align*}
 where we used the bound in Remark 1.7 of \cite{plan2016generalized}, that $w(\bSigma^{\frac{1}{2}}\overline \cT_{K'}(\wb)\cap\mathbb{S}^{p-1}) \leq \|\bSigma^{1/2}\|_{\operatorname{op}}\|\bSigma^{-1/2}\|_{\operatorname{op}}\overline w(\overline \cT_{K'}(\wb)\cap\mathbb{S}^{p-1})$. Using the reasoning around inequality \eqref{chebyshev_sigma} we have
 \begin{align*}
 \frac{\|\bvarepsilon\|^2}{n} \leq 2 \sigma^2,
 \end{align*}
 with probability at least $1-\frac{\Var(\bvarepsilon_i^2)}{n\sigma^4}$. Therefore
 \begin{align*}
     n^{-1}\sum_{i \in [n]} \varepsilon_i \bX_i\T (\hat \bbeta - \wb) \lesssim \frac{\sigma \|\bSigma^{\frac{1}{2}}(\hat \bbeta - \wb)\|\overline w(\overline \cT_{K'}(\wb) \cap \mathbb{S}^{p-1})}{\sqrt{n}}.
 \end{align*}
 Now the quantity $\|\Xb (\hat \bbeta - \wb)\|^2$ can be bounded with Gordon's escape through the mesh result as in \eqref{gordon:lower:bound}. We obtain that with with probability at least $1-e^{-\frac{t^2}{2}}$ 
 \begin{align*}
     \Big\| \Xb\bSigma^{-\frac{1}{2}}\frac{\bSigma^{\frac{1}{2}}(\hat\bbeta-\wb)}{\|\bSigma^{\frac{1}{2}}(\hat\bbeta-\wb)\|} \Big\| \geq 
     \inf\limits_{\ub\in\bSigma^{\frac{1}{2}}\overline \cT_{K'}(\wb)\cap\mathbb{S}^{p-1}}\|\Xb\bSigma^{-\frac{1}{2}}\ub\| \geq \sqrt{n-1}-w(\bSigma^{\frac{1}{2}}\overline \cT_{K'}(\wb)\cap\mathbb{S}^{p-1}) - t.
 \end{align*}
 set $t = \overline w(\overline \cT_{K'}(\wb)\cap\mathbb{S}^{p-1})$ to obtain that with high probability 
 $\|\Xb (\hat \bbeta - \wb)\|^2 \geq \|\bSigma^{\frac{1}{2}}(\hat \bbeta - \wb)\|^2 (\sqrt{n-1} - w(\bSigma^{\frac{1}{2}}\overline \cT_{K'}(\wb)\cap\mathbb{S}^{p-1}) - \overline w(\overline \cT_{K'}(\wb)\cap\mathbb{S}^{p-1}))^2/n$. Under the assumption $\overline w(\overline \cT_{K'}(\wb)\cap\mathbb{S}^{p-1}) = o(\sqrt{n})$ (which implies $ w(\bSigma^{\frac{1}{2}}\overline \cT_{K'}(\wb)\cap\mathbb{S}^{p-1}) = o(\sqrt{n})$ assuming bounded spectrum of $\bSigma$) we conclude that
 \begin{align*}
     \|\bSigma^{\frac{1}{2}}(\hat \bbeta - \wb)\|^2 \lesssim \|\bSigma^{\frac{1}{2}}(\hat \bbeta - \wb)\| \frac{\sigma \overline w(\overline \cT_{K'}(\wb)\cap\mathbb{S}^{p-1}) }{\sqrt{n}} + \lambda\rho_K(\wb) - \lambda \rho_K(\hat \bbeta),
 \end{align*}
 where the sign $\lesssim$ hides absolute constants (that also depend on the assumption $w(\bSigma^{\frac{1}{2}}\overline \cT_{K'}(\wb)\cap\mathbb{S}^{p-1}) = o(\sqrt{n})$) and constants depending on the maximum and minimum eigenvalues of $\bSigma$. Next since we are assuming $\bSigma$ has a bounded spectrum we may write
 \begin{align*}
     \|\hat \bbeta - \wb\|^2 \lesssim n^{-1}\sum_{i \in [n]} (\bX_i\T (\wb - \bbeta^*))^2 + \|\hat \bbeta - \wb\| \frac{\sigma \overline  w(\overline \cT_{K'}(\wb)\cap\mathbb{S}^{p-1}) }{\sqrt{n}} + \lambda\rho_K(\wb) - \lambda \rho_K(\hat \bbeta).
 \end{align*}
 The term $n^{-1}\sum_{i \in [n]} (\bX_i\T (\wb - \bbeta^*))^2$ can be upper bounded by $\|\bSigma^{1/2}(\wb - \bbeta^*)\| \lesssim \|\wb - \bbeta^*\|$, with high probability since $n^{-1}\sum_{i \in [n]} (\bX_i\T (\wb - \bbeta^*))^2/\|\bSigma^{1/2}(\wb - \bbeta^*)\|^2$ is an average of i.i.d. chi-squared random variables which concentrate abound their mean value of $1$ with high probability \cite[see Lemma 1 of][e.g.]{Laurent00adaptive}. Hence
\begin{align*}
     \|\hat \bbeta - \wb\|^2 \lesssim \|\wb - \bbeta^*\|^2 + \|\hat \bbeta - \wb\| \frac{\sigma \overline w(\overline \cT_{K'}(\wb)\cap\mathbb{S}^{p-1}) }{\sqrt{n}} + \lambda\rho_K(\wb) - \lambda \rho_K(\hat \bbeta).
 \end{align*}
 Now, by assumption $\rho_K(\wb)  - \rho_K(\hat \bbeta) \leq s(\wb) \|\wb - \hat \bbeta\|$. It follows that either $\|\hat \bbeta - \wb\| \lesssim \|\wb - \bbeta^*\|$, or $\|\hat \bbeta - \wb\| \lesssim  \frac{\sigma \overline  w(\overline \cT_{K'}(\wb)\cap\mathbb{S}^{p-1}) }{\sqrt{n}} + \lambda s(\wb)$, which grants the conclusion of the theorem statement in both cases since for positive numbers $a,b$ we have $\max(a,b) \lesssim a + b \lesssim \max(a,b)$. 
 
 We now prove the final part of the theorem. First, if $\lambda^*_{\wb} = 0$ there is nothing to prove. Next, suppose that $\lambda_{\wb}^*$ is strictly larger than $0$. Note that this means that at $\lambda_{\wb}^*$ we must have $\rho_{K'}(\hat \bbeta_{\lambda_{\wb}^*}) = \rho_{K'}(\wb)$. By reapplying the previous part of the proof, and taking into account $\rho_{K'}(\hat \bbeta_{\lambda_{\wb}^*}) = \rho_{K'}(\wb)$ we conclude that at $\lambda_{\wb}^*$ we have:
 \begin{align}\label{some:useful:intermediate:inequality}
\|\hat \bbeta_{\lambda_{\wb}^*} - \bbeta^*\| \lesssim \|\bbeta^* - \wb\| + \frac{\sigma \overline w(\overline \cT_{K'}(\wb) \cap \mathbb{S}^{p-1})}{\sqrt{n}}.
\end{align}
We will now consider the difference %{\color{red} here there is some sort of continuity assumption which needs to be checked.}
\begin{align*}
\MoveEqLeft 0 \geq \sum_{i} (Y_i - \bX_i\T \hat \bbeta_{\lambda_{\wb}^*})^2 + \lambda_{\wb}^* \rho_{K'}(\hat \bbeta_{\lambda_{\wb}^*}) - \sum_{i} (Y_i - \bX_i\T \wb)^2 - \lambda_{\wb}^* \rho_{K'}(\wb) \\
& \geq \sum_{i} (Y_i - \bX_i\T \hat \bbeta_{\lambda_{\wb}^*})^2 - \sum_{i} (Y_i - \bX_i\T \wb)^2.
\end{align*}
This difference is negative, and by repeating the proof of the previous part of the theorem, and taking into account \eqref{some:useful:intermediate:inequality}, can be lower bounded by $\gtrsim -\|\bbeta^* - \wb\|^2 -\frac{\sigma^2 \overline w^2(\overline \cT_{K'}(\wb) \cap \mathbb{S}^{p-1})}{n}$ with high probability. On the other hand, let $\wb'$ be the vector that achieves $s(\wb, r')$ for $r' = c_0 (\|\bbeta^* - \wb\| + \frac{\overline w(\overline \cT_{K'}(\wb) \cap \mathbb{S}^{p-1})}{\sqrt{n}})$. Consider the difference
\begin{align*}
\MoveEqLeft \sum_{i} (Y_i - \bX_i\T \wb')^2 + \lambda_{\wb}^* \rho_{K'}(\wb') - \sum_{i} (Y_i - \bX_i\T \wb)^2 - \lambda_{\wb}^* \rho_{K'}(\wb) \\
& \leq \bigg|\sum_{i} (Y_i - \bX_i\T \wb')^2 - \sum_{i} (Y_i - \bX_i\T \wb)^2\bigg| + \lambda_{\wb}^* (\rho_{K'}(\wb') - \rho_{K'}(\wb))
\end{align*}
By repeating the first part of the proof the difference $\bigg|\sum_{i} (Y_i - \bX_i\T \wb')^2 - \sum_{i} (Y_i - \bX_i\T \wb)^2\bigg|$ can be upper bounded by $\lesssim \|\bbeta^* - \wb\|^2 + \frac{\sigma^2 \overline w^2(\overline \cT_{K'}(\wb) \cap \mathbb{S}^{p-1})}{n}$ with high probability, since $\wb' \in \overline \cT_{K'}(\wb)$ and $\|\wb' - \wb\| \leq r'$. However, the difference $ \lambda_{\wb}^* (\rho_{K'}(\wb') - \rho_{K'}(\wb)) = -\lambda_{\wb}^* s(\wb,r')$. Hence if $\lambda_{\wb}^* \geq C_0( \|\bbeta^* - \wb\|^2 + \frac{\sigma^2 \overline w^2(\overline \cT_{K'}(\wb) \cap \mathbb{S}^{p-1})}{n})/s(\wb,r')$ for some large $C_0$, the vector $\wb'$ will have a smaller cost function than the vector $\hat \bbeta_{\lambda_{\wb}^*}$ which is a contradiction. Hence with high probability $\lambda_{\wb}^* \lesssim (\|\bbeta^* - \wb\|^2 + \frac{\sigma^2 \overline w^2(\overline \cT_{K'}(\wb) \cap \mathbb{S}^{p-1})}{n})/s(\wb, r')$.

 %Furthermore as before we have $w(\bSigma^{\frac{1}{2}}\overline \cT_{K'}(\bbeta^*)\cap\mathbb{S}^{p-1}) \leq \|\bSigma^{1/2}\|_{\operatorname{op}}\|\bSigma^{-1/2}\|_{\operatorname{op}}w(\overline \cT_{K'}(\bbeta^*)\cap\mathbb{S}^{p-1})$ (see Remark 1.7 \cite{plan2016generalized}) implying that

 % \begin{align*}
 %     \lambda_{\min}(\bSigma)\|\hat \bbeta - \bbeta^*\|_2^2 \lesssim \frac{\lambda_{\max}(\bSigma)}{\lambda_{\min}(\bSigma)^{\frac{1}{2}}}\|\hat \bbeta - \bbeta^*\|_2 \frac{\sigma w(\overline \cT_{K'}(\bbeta^*)\cap\mathbb{S}^{p-1}) }{\sqrt{n}} + \lambda\rho_K(\bbeta^*) - \lambda \rho_K(\hat \bbeta),
 % \end{align*}

% Next we have $\rho_K(\bbeta^*) - \rho_K(\hat \bbeta) \leq C(\bbeta^*) \|\hat \bbeta - \bbeta^*\|$, and therefore we obtain
% \begin{align*}
%     \|\hat \bbeta - \bbeta^*\|_2 \lesssim \frac{\sigma w(\overline \cT_{K'}(\bbeta^*)\cap\mathbb{S}^{p-1}) }{\sqrt{n}}.
% \end{align*}
% which is what we wanted to show.

\end{proof}

\begin{proof}[Proof of Lemma \ref{lemma:group:lasso:svr}] We first show an upper bound on $s(\wb)$. Observe that if $\wb'$ has any non-zero entries outside of the group support of $\wb$, $\rho_{K'}(\wb')$ can be decreased, $\rho_{K'}(\wb) - \rho_{K'}(\wb')$ increased and $\|\wb - \wb'\|$ can be decreased if we make those elements zero. Hence we may assume that $\wb'$ has the same group support as $\wb$. Since $\|\cdot\|_{2,1}$ is a norm, $\rho_{K'}(\wb) - \rho_{K'}(\wb') \leq \rho_{K'}(\wb - \wb') = \|\wb - \wb'\|_{2,1}$. Using Cauchy-Schwartz we conclude that $s(\wb) \leq \sqrt{s}$.%Suppose $\wb' = \wb + \hb$, and note that $\hb \in \overline \cT_{K'}(\wb)$. 

For the second part, construct a vector $\wb'$ from $\wb$ in the following manner. For all $G \in S$ (where $S$ denotes the set of active groups in $\wb$) take $\wb'_G = \wb_G - \frac{\wb_G}{\|\wb_G\|}h$, where $h \leq \|\wb_G\|$, and take $\wb'_G = \mathbf{0}$ otherwise. It follows that $\rho_{K'}(\wb) - \rho_{K'}(\wb') = s h$. This completes the proof. 
\end{proof}

%%%%%%%%%%%%%
%% SLOPE
%%%%%%%%%%%%

\section{Proofs of Section \ref{slope_sqrtslope:sec}}
%\section{Proof of Theorem \ref{main_rst_slope}}
\begin{proof}[Proof of Theorem \ref{main_rst_slope}]
% First we need to introduce a definition $\vartheta(s, c_0)$ as the minimal \textit{weighted restricted eigenvalue} of the design matrix $\Xb$ \citep[Page 10]{bellec2018slope}
% \begin{align*}
%     \vartheta(s, c_0) = \min_{\bdelta\in\{\bdelta:\sum_{j=1}^p\lambda_j|\delta_{\#j}|\leq (1+c_0)\|\bdelta\|\sqrt{\sum_{j=1}^s\lambda_j^2}\}, \bdelta\neq\mathbf{0}}\,\,\frac{1}{\sqrt{n}} \frac{\|\Xb\bdelta\|}{\|\bdelta\|}
% % \label{def_vartheta}
% \end{align*}
% and 
% \begin{align}
% \label{def_vartheta_star}
%     \vartheta^* = \min_{c_0>0, s\in[1, s^u]} \vartheta(s, c_0)
% \end{align}

For the SLOPE estimator, we combine the results in Corollary 6.2 in \cite{bellec2018slope}. %and Theorem 8.3 in \cite{bellec2018slope}. 
With probability at least $1-\frac{1}{2}(\frac{s^u}{2p})^{\frac{s^u}{\vartheta^*}}$ we have
\begin{align}\label{slope_rate_vartheta}
    \|\hat\bbeta - \bbeta^*\| \lesssim \frac{\sigma }{\vartheta^*}\sqrt{\frac{s^u\log(2ep/s^u)}{n}}.
\end{align}
    
For the square-root SLOPE estimator, we use the result in \cite[Corollary 6.2]{derumigny2018improved}. %and \cite[Theorem 8.3]{bellec2018slope}. 
With probability at least $1-(\frac{s^u}{p})^{s^u}-(1+e^2)e^{-n/24}$ we have the same rate as \eqref{slope_rate_vartheta}. 

%{\color{red} Not sure if we need the 2 lines below:}
It follows that when $C \gtrsim \sigma/\vartheta^*$, $\bbeta^*$ will be a feasible point. Hence $\bbeta^* \in K = \{\bbeta: \|\bbeta\|_1 \leq \|\vb\|_1\}$. Next, since $\|\vb - \hat \bbeta\| \leq C \sqrt{s^u \log (2ep/s^u)/n}$ is guaranteed in step 1, by triangle inequality we have
\begin{align*}
    \|\bbeta^* - \vb\| & \leq \|\hat \bbeta - \vb\| + \|\hat \bbeta - \bbeta^*\|\\
    & \lesssim (C+\sigma/\vartheta^*)\sqrt{\frac{s^u\log(2ep/s^u)}{n}}.
\end{align*}
% By the fact that $\vartheta^*$ should be larger than the smallest eigenvalue of the design matrix, we have
% \begin{align*}
%     \|\bbeta^* - \vb\| \leq
%     (C+\sigma\|\bSigma^{\frac{1}{2}}\|_{\operatorname{op}})\sqrt{\frac{s^u\log(ep/s^u)}{n}}
% \end{align*}

For $\overline w(\cT_K(\vb)\cap\mathbb{S}^{p-1})$, since $K$ is constructed as $K=\{\bbeta: \|\bbeta\|_1 \leq \|\vb\|_1\}$ and $\vb$ is at least $s^u$ sparse, by \cite[Proposition 3.10]{chandrasekaran2012convex} we have
\begin{align*}
    \overline w(\cT_K(\vb)\cap\mathbb{S}^{p-1}) \lesssim \sqrt{s^u\log\frac{e p}{s^u}}.
\end{align*}
Finally since $s^u=o(\sqrt{n}/\log (ep/s^u))$ we have
\begin{align*}
    \overline w(\cT_K(\vb)\cap \mathbb{S}^{p-1})\|\vb-\bbeta^*\| \lesssim \frac{s^u \log ep/s^u}{\sqrt{n}}C= o_p(1),
\end{align*}
by assumption. This completes the proof.
\end{proof}

%%%%%%%%%%%%%
%%
%%%%%%%%%%%%
%\section{}
\begin{proof}[Proof of Lemma \ref{slope_unknownl1_step2_sol}]
Let $\vb_s$ be a vector the set of vectors with $s$ non-zero coordinates. Recall the optimization problem
\begin{align*}
    \argmax \|\vb\|_1, \mbox{ s.t. } \|\vb - \hat \bbeta\| \leq C \sqrt{\frac{s^u \log 2ep/s^u}{n}} \text{ and } \|\vb\|_0 \leq s^u.
\label{slope_unknownl1_step2}
\end{align*}
% can be reduced to an optimization over finite candidates as
% \begin{align}
%     \hat{s} = \argmax_{s\in[0,s^u]}\,\|\vb_s\|_1, \mbox{ s.t. } \|\vb_s - \hat \bbeta\| \leq C \sqrt{\frac{s^u \log 2ep/s^u}{n}}
% \label{slope_step2_reduce}
% \end{align}
% and the desired $\vb = \vb_{\hat{s}}$, where $\vb_s$ is exactly $s$-sparse.

Let $\vb_s$ be an $s < s^u$ sparse vector candidate for being the solution of the program above. First for each coordinate of $\vb_s$, we have either  $\sign (\vb_s^{(i)})\,=\,\sign(\hat\bbeta_{(i)})$, or $\sign(\vb_s^{(i)})\,=0$, because otherwise we can always change that coordinate to $-\sgn(\vb_s^{(i)})(|\vb_s^{(i)}|+2|\hat\bbeta_{(i)}|)$ to make $\|\vb_s\|_1$ larger while keeping $\|\vb_s - \hat \bbeta\|$ unchanged.

Then we show that the non-zero indices in $\vb_s$ have the form $\vb_s^{(i)} = \hat \bbeta_{(i)} + \sign(\hat \bbeta_{(i)}) c$ for some $c \geq 0$. Let $S'$ with $|S'| = s$ be the set of non-zero coordinates of $\vb_s$. The optimization program becomes
\begin{align*}
    \argmax_{s} \sum_{i\in S'}|\vb_s^{(i)}| , \mbox{ s.t. } \sum_{i\in S'}(|\hat \bbeta_{(i)}| - |\vb_s^{(i)}|)^2 + \sum_{i\notin S'} \hat \bbeta_{(i)}^2\leq C \sqrt{\frac{s^u \log 2ep/s^u}{n}}.
\end{align*}
Relax the above problem to 
\begin{align*}
    \argmax_{s} \sum_{i\in S'}a_i , \mbox{ s.t. } \sum_{i\in S'}(|\hat \bbeta_{(i)}| - a_i)^2 + \sum_{i\notin S'} \hat \bbeta_{(i)}^2\leq C \sqrt{\frac{s^u \log 2ep/s^u}{n}},
\end{align*}
where $a_i$ need not be positive. Using Lagrange multipliers we obtain
\begin{align*}
    L = \sum_{i\in S'}a_i + \lambda \sum_{i\in S'}(|\hat \bbeta_{(i)}| - a_i)^2,
\end{align*}
and solve $\frac{\partial L}{\partial a_i} = 0$ to get
\begin{align*}
    a_i = |\hat \bbeta_{(i)}| + \frac{1}{2\lambda} \text{ for all } i\in S'.
\end{align*}
Let $c=\frac{1}{2\lambda}$. We have that 
\begin{align*}
    s c^2 + \sum_{i \not \in S'} \hat \bbeta^2_{(i)} \leq C \sqrt{\frac{s^u \log 2ep/s^u}{n}}.
\end{align*}
Hence the maximal value of $c$ satisfies $c^2 = \frac{C \sqrt{\frac{s^u \log 2ep/s^u}{n}} - \sum_{i \not \in S'} \hat \bbeta^2_{(i)}}{s} \geq 0$. The latter is $\geq 0$ if there exists a feasible point in the program. When $C \sqrt{\frac{s^u \log 2ep/s^u}{n}} < \sum_{i \not \in S'} \hat \bbeta^2_{(i)}$ then the vector with support $S'$ can never be feasible in any case.

Note that our objective function is
\begin{align*}
    \sum_{i \in S'} a_i = \sum_{i \in S'} |\hat \bbeta_{(i)}| + s c,
\end{align*}
which is maximized when $c = \sqrt{\frac{C \sqrt{\frac{s^u \log 2ep/s^u}{n}} - \sum_{i \not \in S'} \hat \bbeta^2_{(i)}}{s}}$. It is also clear that in the above, one should pick $S'$ which minimizes the coefficients of $\sum_{i \not \in S'} \hat \bbeta^2_{(i)}$ and at the same time, maximizes $\sum_{i \not \in S'} |\hat \bbeta_{(i)}|$. Clearly, this set corresponds to the maximal in magnitude elements in the vector $\hat \bbeta$. Since $a_i$ are positive then one can find the corresponding maximal values of $|\vb_{(i)}| = a_i$, and $\vb_{(i)} = \hat \bbeta_{(i)} + \sign(\hat \bbeta_{(i)})c$ on the set $S'$ where the largest $s$ coefficients of $\hat \bbeta$ are located. Furthermore, the bigger the $s$ is the bigger the objective function. Hence we take $s = s^u$. This completes the proof.
\end{proof}

% Combine with the fact that for $i\in S'$ sign$(\vb_s^{(i)})\,=\,$sign$(\hat\bbeta_{(i)})$ we have $\vb_s^{(i)} = \hat \bbeta_{(i)} + \sign(\hat \bbeta_{(i)})c$ for some $c \geq 0$. Notice that the cardinality of set $S'$ is $s$, so the optimization \eqref{slope_step2_reduce} is converted to
% \begin{align*}
%     \argmax_{s, c}\, sc  , \mbox{ s.t. } s c^2 \leq C^2 \frac{s^u \log 2ep/s^u}{n} - \sum_{i\notin S'}\hat\bbeta_{(i)}^2
% \end{align*}
% it is not hard to see that the analytical solution of the above program is
% \begin{align*}
%     sc = \sqrt{s}\sqrt{C^2\frac{s^u \log 2ep/s^u}{n} - \sum_{i \notin S'} \hat \bbeta_{(i)}^2}
% \end{align*}
% Notice that the RHS of above inequality is increasing with $s$, so we actually have $\hat s=s^u$ as the solution of \eqref{slope_step2_reduce}. Consequently $c=\sqrt{\bigg(C^2\frac{s^u \log 2ep/s^u}{n} - \sum_{i = s^u + 1}^p \hat \bbeta_{(i)}^2\bigg)/s^u}$. Also since the RHS is decreasing with $\sum_{i \notin S'} \hat \bbeta_{(i)}^2$, the non-zero indices of $\vb_s$ should be the same as indices with largest absolute values in $\hat\bbeta$.

%%%%%%%%%%%%%
%%
%%%%%%%%%%%%
%\section{}
\begin{proof}[Proof of Lemma \ref{sigma_hat_slope}]
According to the results in \cite[Corollary 6.2]{bellec2018slope} and \cite[Corollary 6.2]{derumigny2018improved}, with probability converging to 1, the quantity $\frac{1}{\sqrt{n}}\|\overline{\Xb}(\hat \bbeta - \bbeta^*)\|$ can be bounded as
    \begin{align*}
        \frac{1}{\sqrt{n}}\|\overline{\Xb}(\hat \bbeta - \bbeta^*)\| \lesssim \frac{\sigma}{\vartheta^*}\sqrt{\frac{s^u \log(2ep/s^u)}{n}}.
    \end{align*}
    conditional on $\overline{\Xb}$ satisfying the WRE with $\vartheta^*$, where $\vartheta^*$ is defined in the main text and is $\vartheta(s^u, 3)$ for the LASSO, and $\vartheta(s^u, 20)$ for square-root SLOPE. From \cite[Theorem 8.3]{bellec2018slope} and the assumptions of Remark \ref{slope:remark:important}, we know that $\vartheta^* \geq \kappa/\sqrt{2}$ with high probability and $\overline{\Xb}$ satisfies the WRE condition. This is what we wanted to show.
 \end{proof}
%     we can see $\vartheta^*$ is always larger than the minimum eigenvalue of $\bSigma^{\frac{1}{2}}$, so that
%     \begin{align*}
%         \frac{1}{\sqrt{n}}\|\Xb(\hat \bbeta - \bbeta^*)\| \lesssim \sigma\|\bSigma^{\frac{1}{2}}\|_{\operatorname{op}}\sqrt{\frac{s^u \log(2ep/s^u)}{n}}
%     \end{align*}
% Thus 
% \begin{align*}
%     \delta \sim \|\bSigma^{\frac{1}{2}}\|_{\operatorname{op}}\sqrt{s^u \log(2ep/s^u)}
% \end{align*}
% Since $s^u=o(\sqrt{n}/\log p)$ and $\bSigma$ has bounded eigenvalues, we have $\delta=o(\sqrt{n})$.

%%%%%%%%%%%%%%%%%%
%% non-Gaussian
%%%%%%%%%%%%%%%%%
\section{Proofs of Section \ref{subGaussian:noise:section}}
%\section{}
\begin{proof}[Proof of Lemma \ref{non_empty_int_nonGaussian}]
It suffices to show that for a sufficiently large $\rho'$ the constraint $\|\tilde\Xb\bSigma^{-1}\eb^{(j)}\|_{\infty}\leq \rho'\sqrt{\log n}$ contains a $\delta$ $\ell_2$-ball, since we have proved that the other set contains a small ball around the point $\bSigma^{-1}\eb^{(j)}$ in Corollary \ref{nonempty_int_Q}.
\begin{enumerate}
    \item \textbf{Feasible Point:}\\
    We argue that $\bmeta = \bSigma^{-1}\eb^{(j)}$ is a feasible point since $\|\tilde\Xb\bSigma^{-1}\eb^{(j)}\|_{\infty}\leq \rho'\sqrt{\log n}$ with probability converging to one. Notice that each coordinate of $\tilde\Xb\bSigma^{-1}\eb^{(j)}$ is a sub-Gaussian variable since
    \begin{align*}
        \|(\tilde\Xb\bSigma^{-1}\eb^{(j)})_i\|_{\psi_2} & \leq \|\eb^{(j)\top}\bSigma^{-1}\|\|\tilde\bX_i\|_{\psi_2}\\
        & = \sqrt{\bSigma^{-2}_{jj}}\,\|\tilde\bX_i\|_{\psi_2} = O(1).
    \end{align*}
    Since $\bSigma$ has bounded eigenvalues so does $\bSigma^{-2}$, and hence all of its entries should be bounded, thus $\bSigma^{-2}_{jj}$ is bounded. And since $\tilde\bX_i$ is either a bounded or Gaussian, which both belong to the sub-Gaussian category, $\tilde\bX_i$ is sub-Gaussian. %, thus $\|\bSigma^{-1/2}\tilde\bX_i\|_{\psi_2}$ is bounded. 
    Therefore, $\|(\tilde\Xb\bSigma^{-1}\eb^{(j)})_i\|_{\psi_2}$ is bounded for all $i\in[n]$, or in other words each coordinate of $\tilde\Xb\bSigma^{-1}\eb^{(j)}$ is sub-Gaussian. 
    By the concentration inequality of maximum sub-Gaussian variables \citep[p. 14]{duchi2017lecture}, with probability converging to one
    \begin{align*}
        \max_{i\in[n]} |(\tilde\Xb\bSigma^{-1}\eb^{(j)})_i| \lesssim \sqrt{\log n}.
    \end{align*}
    Thus for a sufficiently large $\rho'$ we have
    \begin{align*}
        \|\tilde\Xb\bSigma^{-1}\eb^{(j)}\|_{\infty}\leq \rho'\sqrt{\log n}.
    \end{align*}
    
    \item \textbf{Non-empty Interior:}\\
    We are able to find $\bmeta = \bSigma^{-1}\eb^{(j)}$ as a feasible point. Now the idea is to show that there exists a small $\delta>0$ such that $\BB_{\delta}(\eb^{(j)\top}\bSigma^{-1})$ is still inside of the feasible region with high probability. Now let $\bx$ be a unit vector. We have
    \begin{align*}
        \|\tilde\Xb(\bSigma^{-1}\eb^{(j)}+\delta\bx)\|_{\infty} & \leq \rho'\sqrt{\log n} + \delta\|\tilde\Xb\bx\|_{\infty}.
    \end{align*}
    Picking $\delta = \rho'\sqrt{\log n}/\sup_{\xb: \|\xb\| \le 1}\|\tilde \Xb \bx\|_{\infty}$ shows that for the value $2 \rho'$ the set has non-empty interior. This completes the proof.
    % Notice that each coordinate of $\tilde\Xb\bx$ is a standard Gaussian random variable. By the standard concentration results of maximum Gaussian variables, with probability converging to one
    % \begin{align*}
    %     \|\tilde\Xb\bx\|_{\infty} = \max_{i\in[n]}|\tilde\bX_i\bx| \leq 2\sqrt{\log n}
    % \end{align*}
    % Thus let $\epsilon=2\delta\sqrt{\log n}$, we have
    % \begin{align*}
    %     \|\tilde\Xb(\bSigma^{-1}\eb^{(j)}+\delta\bx)\|_{\infty} & \leq \rho'\sqrt{\log n} + \epsilon
    % \end{align*}
    % Since we can find such a $\delta$ for any $\epsilon > 0$, the ball $\BB_{\delta}(\eb^{(j)\top}\bSigma^{-1})$ is inside of the feasible region. Thus the set 
    % $$\{\bmeta:\,\,\|\Xb\bmeta\|_{\infty} \leq \rho' \sqrt{\log n}\}$$
    % has a non-empty interior with high probability.
\end{enumerate}
\end{proof}

%%%%%%%%%%%%%%%%%%
%%
%%%%%%%%%%%%%%%%%
\begin{proof}[Proof of Lemma \ref{sg_psi_prime}]

This fact follows by a direct calculation. We omit the details.

\end{proof}

%{\color{red} NO PROOF HERE?}

%%%%%%%%%%%%%%%%%%
%%
%%%%%%%%%%%%%%%%%
\begin{proof}[Proof of Theorem \ref{non_gaussian}]

%{\color{red} Defs of $c$ vs $c_n$ has changed, $C_2$ is called $C'$.}

We state and prove the following result. Its proof rests on an argument from \cite[Lemma 3.1]{javanmard2014confidence}.
\begin{lemma} 
\label{jm_lemma3.1}
The following holds:
\begin{align*}
\|\hat \bSigma^{1/2}\hat\bmeta\|^2 \geq \sup_{\ub \in \cT_K(\vb) \cap \mathbb{S}^{p-1}}\frac{(|u_j| - \rho\lambda)^2\mathbbm{1}\{|u_j|\geq \rho\lambda\}}{\ub \T \hat \bSigma \ub } \geq \frac{(|u^*_j| - \rho\lambda)^2\mathbbm{1}\{|u_j|\geq \rho\lambda\}}{\ub^{*\top} \hat \bSigma \ub^*},
\end{align*}
where $\lambda = \frac{\bar w(\cT_K(\vb) \cap \mathbb{S}^{p-1})}{\sqrt{n}}$, and $\ub^* = \frac{\bbeta^* - \vb}{\|\bbeta^* - \vb\|}$.
\end{lemma}
\begin{proof}
%This proof is similar to the proof of \cite[Lemma 3.1]{javanmard2014confidence}.

Let $\lambda=\frac{\overline w(\cT_K(\vb)\cap\mathbb{S}^{p-1})}{\sqrt{n}}$. The constraint $\sup_{\ub \in \cT_K(\vb) \cap \mathbb{S}^{p-1}} |(\bmeta\T \hat \bSigma - \eb^{(j)\top}) \ub|  \leq \rho\frac{\overline w(\cT_K(\vb)\cap\mathbb{S}^{p-1})}{\sqrt{n}}$ implies 
\begin{align*}
    & u_j-\langle \ub, \hat\bSigma\bmeta\rangle \leq \rho\lambda,\quad \ub\in\cT_K(\vb)\cap \mathbb{S}^{p-1}\,\,\text{or}\\
    & -\rho\lambda\leq u_j-\langle \ub, \hat\bSigma\bmeta\rangle,\quad \ub\in\cT_K(\vb)\cap \mathbb{S}^{p-1}.
\end{align*}
Consider the first case $u_j-\langle \ub, \hat\bSigma\bmeta\rangle \leq \rho\lambda$. Then for any feasible $\tilde\bmeta$ and $c\geq0$, when $\ub\in\cT_K(\vb)\cap \mathbb{S}^{p-1}$ we have
\begin{align*}
    \tilde\bmeta\hat\bSigma\tilde\bmeta & \geq \tilde\bmeta\hat\bSigma\tilde\bmeta + c(u_j-\rho\lambda) - c\langle \ub, \hat\bSigma\tilde\bmeta\rangle\\
    & \geq \min_{\bmeta:\,\, \sup_{\ub \in \cT_K(\vb) \cap \mathbb{S}^{p-1}} (u_j-\rho\lambda) - \langle \ub, \hat\bSigma\bmeta\rangle \leq0} [\bmeta\hat\bSigma\bmeta + c(u_j-\rho\lambda) - c\langle \ub, \hat\bSigma\bmeta\rangle].
\end{align*}
%{\color{red} THERE ARE QUESTION MARKS BELOW?}
Thus the optimal value of the optimization \eqref{optimization:algo6} in step \ref{step2_non_gaussian_noise} satisfies 
\begin{align*}
    \|\hat\bSigma^{1/2}\hat\bmeta\|^2 \geq \min_{\bmeta:\,\, \sup_{\ub \in \cT_K(\vb) \cap \mathbb{S}^{p-1}} (u_j-\rho\lambda) - \langle \ub, \hat\bSigma\bmeta\rangle \leq 0} [\bmeta\hat\bSigma\bmeta + c(u_j-\rho\lambda) - c\langle \ub, \hat\bSigma\bmeta\rangle].
\end{align*}
When $\bmeta=c\ub/2$, the RHS is minimized. Thus
\begin{align*}
    \|\hat\bSigma^{1/2}\hat\bmeta\|^2 \geq c(u_j-\rho\lambda) - \frac{c^2}{4}\ub\T\hat\bSigma\ub,\quad\text{ if }u_j-\rho\lambda\leq \frac{c}{2}\ub\T\hat\bSigma\ub.
\end{align*}
We then optimize over $c$. When $c=2(u_j-\rho\lambda)/\ub\T\hat\bSigma\ub$, the condition $u_j-\rho\lambda\leq \frac{c}{2}\ub\T\hat\bSigma\ub$ holds for any $\ub$. And since we need $c\geq0$, the condition $u_j\geq \rho\lambda$ should hold. Plug in the value of $c$ to the RHS, we get
\begin{align*}
    \|\hat\bSigma^{1/2}\hat\bmeta\|^2 \geq \frac{(u_j-\rho\lambda)^2}{\ub\T\hat\bSigma\ub}\mathbbm{1}\{u_j\geq \rho\lambda\}.
\end{align*}
Similarly for the second case $-\rho\lambda\leq u_j-\langle \ub, \hat\bSigma\bmeta\rangle$ we will get
\begin{align*}
    \|\hat\bSigma^{1/2}\hat\bmeta\|^2 \geq \frac{(-u_j-\rho\lambda)^2}{\ub\T\hat\bSigma\ub}\mathbbm{1}\{-u_j\geq \rho\lambda\}.
\end{align*}
Finally 
\begin{align*}
    \|\hat\bSigma^{1/2}\hat\bmeta\|^2 \geq \frac{(|u_j|-\rho\lambda)^2}{\ub\T\hat\bSigma\ub}\mathbbm{1}\{|u_j|\geq \rho\lambda\}.
\end{align*}
\end{proof}

\begin{lemma}
\label{bdd_u*Sigu*}
Suppose $\bX_i$ has a covariance matrix $\bSigma$, and the eigenvalues of $\bSigma$ are bounded. $\ub^*$ is defined as in Lemma \ref{jm_lemma3.1}. Then conditionally on $\overline{\Xb}$ we have that $ \lambda_{\min}(\bSigma)/2 \leq \ub^{*\top}\hat \bSigma \ub^* \leq 3/2\|\bSigma\|_{\operatorname{op}}$ with high probability.
%{\color{blue} (which should be $\lambda_{\min}(\bSigma)/2$ e.g.).This holds since conditional on $\bar \Xb$,  $u^*$ is just some fixed vector. }
\end{lemma}

\begin{proof}
Since conditionally on $\overline{\Xb}$ we have that $\vb$ is independent of $\hat \bSigma$, and $(\ub^{*\top} \bX_i)^2$ is a sub-exponential random variable (with norm less than $K := \|\Xb_i\|^2_{\psi_2}$ which is bounded by assumption), we can use a Bernstein type of concentration inequality to claim that 
\begin{align*}
    \PP\bigg(\bigg|\frac{1}{n} \sum_{i \in [n]} (\ub^{*\top} \bX_i)^2 - \EE [(\ub^{*\top} \bX_i)^2| \overline{\Xb}]\bigg| \geq t\bigg) \leq \exp(-c n t^2/K^2 \wedge t/K).
\end{align*}
Choose $t = \lambda_{\min}(\bSigma)/2$, and note that $\lambda_{\max}(\bSigma) \geq \EE [(\ub^{*\top} \bX_i)^2 | \overline{\Xb}] \geq \lambda_{\min}(\bSigma)$, completing the proof.
\end{proof}

\begin{theorem}[Lindeberg-Feller CLT]\citep[p. 901]{greene2003econometric}
\label{lindeberg_feller_clt}
Let $\bX_1,\ldots,\bX_n$ be independent but not necessarily identically distributed random variables with $\mathbb{E}[\bX_i]=\mu_i$ and $\var(\bX_i)=\sigma_i^2<\infty$. Define $\overline\mu_n=n^{-1}\sum_{i=1}^n\mu_i$ and $\overline\sigma_n^2=n^{-1}\sum_{i=1}^n\sigma_i^2$. Suppose
\begin{align*}
    \lim_{n\rightarrow\infty}\frac{\max_i\sigma_i^2}{n\overline\sigma^2_n}=0, \quad 
    \lim_{n\rightarrow\infty}\overline\sigma_n^2 < \infty.
\end{align*}
Then
\begin{align*}
    \sqrt{n}(\frac{\bar\bX-\overline\mu_n}{\overline\sigma_n}) \xrightarrow{\text{d}} Z \sim N(0,1).
\end{align*}
% (Theorem 8 here: \href{https://faculty.washington.edu/ezivot/econ583/asymptoticsprimerslides.pdf}{link1} or Theorem 6.3 here \href{https://sites.math.washington.edu/~morrow/336_18/papers18/jainul.pdf}{link2})
\end{theorem}

The proof of Theorem \ref{non_gaussian} starts here. We divide the proof into two cases in terms of the scale of $\|\hat\bSigma^{1/2}\hat\bmeta\|$. A sufficiently large $\|\hat\bSigma^{1/2}\hat\bmeta\|$ is required if one would like to use Lindeberg-Feller CLT to derive the limiting distribution of $\sqrt{n}(\hat\bbeta_d^{(j)} - \bbeta^{*(j)})$. Let $a_n = o(1)$ be any slowly converging to $0$ rate such that $\frac{1}{a_n}=o(\frac{n}{\log n})$. 
\begin{enumerate}
    \item 
    Suppose now that $$\|\hat\bSigma^{1/2}\hat\bmeta\| \leq C_1 \frac{\sqrt{\log n}/\sqrt{(\|\bbeta^* - \vb\|\sqrt{\log n}) \vee a_n}}{\sqrt{n}},$$ 
    for some constant $C_1$. Then by Lemma \ref{jm_lemma3.1} and Lemma \ref{bdd_u*Sigu*}, for some constant $C'$ we have
    \begin{align*}
        (|u_j^*| - \rho\lambda)\mathbbm{1}(|u_j^*| > \rho\lambda) \leq C' \frac{\sqrt{\log n}/\sqrt{(\|\bbeta^* - \vb\|\sqrt{\log n}) \vee a_n}}{\sqrt{n}}.
    \end{align*}
    Plug in $\ub^* = \frac{\bbeta^* - \vb}{\|\bbeta^* - \vb\|}$ to get
    \begin{align*}
        |\bbeta^*_j - \vb_j| \leq \|\bbeta^* -\vb\|C'\frac{\sqrt{\log n}/\sqrt{(\|\bbeta^* - \vb\|\sqrt{\log n})\vee a_n}}{\sqrt{n}} + \|\bbeta^* -\vb\|\rho\lambda.
    \end{align*}
    Given that $\|\bbeta^* -\vb\|\max\{\overline w(\cT_K(\vb)\cap \mathbb{S}^{p-1}), \sqrt{\log n}\} = o_p(1)$, we have 
    $$|\bbeta^*_j - \vb_j|=o_p(1/\sqrt{n}),$$ 
    so $\vb_j$ is more precise than what we need already.
    
    Then we show that the debiased estimator $\hat\bbeta_{d}^{(j)} \leftarrow \eb^{(j)\top}\vb + n^{-1}\hat\bmeta\T \tilde\Xb\T (\tilde\bY - \tilde\Xb\vb)$ is still $o_p(1/\sqrt{n})$ close to $\bbeta^*_j$ since the correction term
    \begin{align*}
        n^{-1}\hat\bmeta\T \tilde\Xb\T (\tilde\bY - \tilde\Xb\vb)=o_p(1/\sqrt{n}).
    \end{align*}
    We have
\begin{align*}
    n^{-1}\hat\bmeta\T \tilde\Xb\T (\tilde\bY - \tilde\Xb\vb) & \leq \hat\bmeta\T\frac{\tilde \Xb\T\tilde \Xb}{n}(\bbeta^*-\vb) + \hat\bmeta\T\frac{\tilde \Xb\T\bvarepsilon}{n}\\
    & \leq |\hat\bmeta\T \hat \bSigma \ub^*|\,\|\bbeta^* - \vb\| + 
    \frac{1}{\sqrt{n}}\|\hat\bSigma^{1/2}\hat\bmeta\| \big|\sum_{i \in n}\frac{(\tilde \Xb\hat\bmeta)_i}{\|\tilde \Xb\hat\bmeta\|} \varepsilon_i\big|,
\end{align*}
where $\ub^* = \frac{\bbeta^* - \vb}{\|\bbeta^* - \vb\|}$.

The first term, can be bounded as follows. The first line uses the first constraint in step \ref{step2_non_gaussian_noise}, and the second line uses Lemma \ref{jm_lemma3.1} and Lemma \ref{bdd_u*Sigu*}. Suppose the upper bound of $\sqrt{{\ub^*}\T \hat \bSigma \ub^*}$ is $C_3$ for a constant $C_3>0$.
\begin{align*}
    |\hat\bmeta\T \hat \bSigma \ub|\|\bbeta^* - \vb\| & \leq (\rho\lambda + |u^*_j|)\|\bbeta^* - \vb\| \\
    & \leq (\rho\lambda + C_3(\|\hat \bSigma^{1/2}\hat\bmeta\| + \rho\lambda))\|\bbeta^* - \vb\|,
\end{align*}
Since $\|\bbeta^* - \vb\|\overline w(\cT_K(\vb)\cap \mathbb{S}^{p-1})=o_p(1)$, we have $\lambda\|\bbeta^* - \vb\|=o_p(1/\sqrt{n})$. And by the condition of $\|\hat\bSigma^{1/2}\hat\bmeta\|$ we have $\|\hat \bSigma^{1/2}\bmeta\|\|\bbeta^* - \vb\| = o_p(1/\sqrt{n})$ as well. Thus the above quantity is $o_p(1/\sqrt{n})$. 

For the second term, by the condition $\frac{1}{a_n}=o(\frac{n}{\log n})$ we have $\|\hat\bSigma^{1/2}\hat\bmeta\|=o_p(1)$. Notice that $\big|\sum_{i \in n}\frac{(\tilde\Xb\hat\bmeta)_i}{\|\tilde\Xb\hat\bmeta\|} \varepsilon_i\big|=O_p(1)$ since it is sub-Gaussian condioned on $\tilde \Xb$. This is because $\epsilon_i$ is sub-Gaussian, it is independent of $\tilde \Xb$ and the coefficients satisfies $\sum_{i \in n}\big(\frac{(\tilde \Xb\hat\bmeta)_i}{\|\tilde \Xb\hat\bmeta\|}\big)^2 = 1$.

Hence we have established that 
$$\hat\bbeta_d^{(j)}-\bbeta^*_j = o_p(1/\sqrt{n}),$$ so any confidence interval centering at $\hat\bbeta_d^{(j)}$ with length $O(1/\sqrt{n})$ will contain $\bbeta^*_j$. Even though such a confidence interval might not be very efficient since $\hat\bbeta_d^{(j)}$ converges faster than the rate $1/\sqrt{n}$.

To make sure the confidence interval is of the length $O(1/\sqrt{n})$, one can pick some small constant $c > C'(\log n)^{1/2}/\sqrt{(\|\bbeta^* - \vb\|\sqrt{\log n}) \vee a_n}$ and make the confidence intervals as \eqref{ci_non_gaussian_noise}.
    
    \item
    Suppose now that $$\|\hat \bSigma^{1/2}\bmeta\| \geq C_1 \frac{\sqrt{\log n}/\sqrt{(\|\bbeta^* - \vb\|\sqrt{\log n}) \vee a_n}}{\sqrt{n}}.$$ 
    In that case it follows 
    $$\|\tilde\Xb\hat\bmeta\|_{\infty}/(\sqrt{n}\|\hat \bSigma^{1/2}\hat\bmeta\|) \lesssim (\|\bbeta^* - \vb\|\sqrt{\log n}) \vee a_n = o_p(1),$$ 
    so we can apply the Lindeberg-Feller CLT (Theorem \ref{lindeberg_feller_clt}). 
    Let $Z_j = \frac{1}{\sqrt{n}}\hat\bmeta\T \tilde\Xb\T\bvarepsilon$, we have
\begin{align*}
    \sqrt{n}(\hat\bbeta_d^{(j)} - \bbeta^{*(j)}) = Z_j + \Delta_j,\quad
    \Delta_j = \sqrt{n}(\hat \bmeta\T \hat \bSigma - \eb^{(j)\top})(\bbeta^* - \vb).
\end{align*}
$\Delta_j$ converges to zero with probability converging to one since $\|\bbeta^* - \vb\|\overline w(\cT_K(\vb)\cap \mathbb{S}^{p-1})=o_p(1)$. And $Z_j$ is Gaussian conditional on $\overline \Xb, \overline \bY,\tilde \Xb$ by the Lindeberg-Feller CLT (Theorem \ref{lindeberg_feller_clt})
\begin{align*}
     \frac{Z_j}{\sigma\|\hat \bSigma^{1/2} \hat\bmeta\|} \xrightarrow{\text{d}}N(0,1).
\end{align*}
Thus the confidence interval \eqref{ci_non_gaussian_noise} also applies in this case.
\end{enumerate}
\end{proof}

\section{Time Series}

In this section we provide time series examples of positive monotone cone regression.

\subsection{Real Data Example}

In this example we look at temperatures from different cities in the month of January. The dataset used is freely available online on Kaggle (it is entitled ``Daily Temperature of Major Cities'') or alternatively provide a link \href{https://www.kaggle.com/sudalairajkumar/daily-temperature-of-major-cities}{here}. The dataset was first cleaned of missing values by dropping observations. Next we selected $5$ large US cities for the experiment. These are Los Angeles, Chicago, Philadelphia, Seattle and Las Vegas. For each city we observe $26$ years of temperature data --- one daily average temperature for each day of the year from 1995 to 2020. We look at the regression of the average daily temperature on January 5 $\sim$ the average daily temperatures from January 1 through January 4 in each year. This results in a $\bY$ value of size $26 \times 1$ and an $\Xb$ predictor matrix of dimension $26 \times 4$. The consecutive days in January were selected to avoid seasonality effects, and we considered only $5$ days in order to be able to also fit an unrestricted linear model. We believe that in this example it is reasonable to assume that the vector $\bbeta$ belongs to a positive monotone cone, since we expect the average daily temperatures in closer dates to January 5 to predict the average daily temperature on January 5 better, and in addition intuitively one would not expect the coefficients in such a regression to admit negative values. Below we show two tables. The first table, Table \ref{our:algo:table} contains the p-values (of two sided tests against $0$) generated by our algorithm for each of the cities and dates. The second table, Table \ref{lm:pvals:table} contains p-values from running an unrestricted linear model. The reported p-values are raw --- there has not been any multiple testing adjustment. We can see that both approaches provide p-values which are largely in agreement. Our approach may appear slightly more conservative, which may be explained by the fact that we are splitting the data.

\begin{table}[ht]%\label{our:algo:table}
\centering
\begin{tabular}{rrrrr}
  \hline
 & Jan 1 & Jan 2 & Jan 3 & Jan 4 \\ 
  \hline
Los Angeles & 0.72 & 0.95 & 0.93 & 0.00 \\ 
  Chicago & 0.30 & 0.86 & 0.75 & 0.00 \\ 
  Philadelphia & 0.78 & 0.12 & 0.82 & 0.07 \\ 
  Seattle & 0.71 & 0.96 & 0.42 & 0.01 \\ 
  Las Vegas & 0.76 & 0.93 & 0.21 & 0.00 \\ 
   \hline
\end{tabular}
  \caption{p-values of the positive monotone cone regression}
  \label{our:algo:table}
 \end{table}

\begin{table}[ht]
\centering
\begin{tabular}{rrrrr}
  \hline
 & Jan 1 & Jan 2 & Jan 3 & Jan 4 \\ 
  \hline
Los Angeles & 0.71 & 0.37 & 0.07 & 0.00 \\ 
  Chicago & 0.37 & 0.18 & 0.82 & 0.00 \\ 
  Philadelphia & 0.53 & 0.04 & 0.94 & 0.06 \\ 
  Seattle & 0.80 & 0.68 & 0.16 & 0.00 \\ 
  Las Vegas & 0.99 & 0.51 & 0.09 & 0.00 \\ 
   \hline
\end{tabular}
  \caption{p-values of the unrestricted linear model}
\label{lm:pvals:table}
\end{table}
\subsection{Sufficient Conditions}

In this section we show some sufficient conditions under which a covariance matrix of stationary time series with monotone $\bbeta^*$ coefficients has bounded spectrum.

Suppose we observe
\begin{align*}
    Y = \sum_{i = 1}^p\beta_i^* X_i + \varepsilon,
\end{align*}
where $\varepsilon \sim N(0, \sigma^2)$, and each $X_i$ has the same distribution as $Y$ (since the series is stationary). Here, in order to match the notation from the main text, the closer the index $i$ is to $p$ the more recent an observation from the series is. According to \cite{lutkepohl2005new} the autocovariance of such a time series is 
\begin{align*}
    \bSigma = \EE \bX\bX\T = \sum_{i = 0}^{\infty} \sigma^2 \Ab^i \eb_1 \eb_1\T \Ab^{i^\top},
\end{align*}
where $\Ab = \begin{pmatrix}
     \beta_p^* &  \beta_{p-1}^* & \ldots & \beta^*_{2} & \beta^*_1 \\
     1 & 0 & \ldots & 0 & 0\\
     0 & 1 & \ddots & 0 & 0\\
     0 & 0 &\ldots & 1 & 0
\end{pmatrix}$, and $\eb_i = (0,\ldots0,\underbrace{1}_{i},0\ldots,0)\T$. The process is stable (which implies that it is stationary \citep{lutkepohl2005new}) if the equation $1 - \sum a_i \lambda^i$ has solution only outside of the unit disk. It is simple to see that this is implied when $\sum_i |\beta^*_i| \leq 1$. Since $\bSigma = \sigma^2 \eb_1 \eb_1\T + \sigma^2 \Ab \bSigma \Ab\T$ it is easy to verify that $\bSigma$ satisfies the identities
\begin{align*}
    \eb_i\T \bSigma \eb_i = \sigma^2 + \sigma^2 \bbeta_{\downarrow}^{*\top} \bSigma \bbeta_{\downarrow}^{*}
\end{align*}
and for $j > i$
\begin{align*}
    \eb_i\T \bSigma \eb_j = \sigma^2 \bbeta_{\downarrow}^* \bSigma \eb_{j-i} = \sigma^2 \bbeta_{\downarrow}^{*\top} \Ab^{j-i}\bSigma \bbeta_{\downarrow}^*,
\end{align*}
where $\bbeta^*_{\downarrow} = (\beta_p^*, \beta_{p-1}^*,\ldots,\beta_1^*)$.

For simplicity we now restrict to the case when all $\beta^*_i \geq 0$ and WLOG we assume $\sigma^2 = 1$. By Gershgorin's disk theorem we have that 
\begin{align*}
    \lambda_{\max}(\bSigma) \leq 1 + \bbeta_{\downarrow}^{*\top} \bSigma \bbeta_{\downarrow}^{*} + \sum_{j \neq i} \bbeta_{\downarrow}^{*\top} \Ab^{|j-i|}\bSigma \bbeta_{\downarrow}^* \leq 1 + 2\bbeta_{\downarrow}^{*\top} (\mathbb{I} - \Ab)^{-1}\bSigma \bbeta_{\downarrow}^*,
\end{align*}
where the last inequality follows since all coefficients in the expression $\bbeta_{\downarrow}^{*\top}  + \sum_{j \neq i} \bbeta_{\downarrow}^{*\top} \Ab^{|j-i|}\leq 2 \sum_{j = 0}^{\infty}\bbeta_{\downarrow}^{*\top}\Ab^{j}$ are positive (and we are assuming the matrix $\mathbb{I} - \Ab$ is invertible).

It is simple to see that
\begin{align*}(\mathbb{I} -\Ab)^{-1} = \begin{pmatrix}
         \frac{1}{1 - \sum_i \beta_i^*} &  \frac{1-\beta^*_p}{1 - \sum_i \beta_i^*} - 1 & \frac{1-\beta^*_p -\beta^*_{p-1}}{1 - \sum_i \beta_i^*} - 1& \ldots & \frac{1-\sum_{i=2}^{p}\beta_i^*}{1 - \sum_i \beta_i^*} - 1\\
         \frac{1}{1 - \sum_i \beta_i^*} & \frac{1-\beta^*_p}{1 - \sum_i \beta_i^*} & \frac{1-\beta^*_p -\beta^*_{p-1}}{1 - \sum_i \beta_i^*} - 1 & \ldots & \frac{1-\sum_{i=2}^{p}\beta_i^*}{1 - \sum_i \beta_i^*} - 1\\
         \frac{1}{1 - \sum_i \beta_i^*} & \frac{1-\beta^*_p}{1 - \sum_i \beta_i^*} & \frac{1-\beta^*_p -\beta^*_{p-1}}{1 - \sum_i \beta_i^*}  & \ldots & \frac{1-\sum_{i=2}^{p}\beta_i^*}{1 - \sum_i \beta_i^*} - 1\\
         \vdots & \ddots & \ddots & \ldots & \vdots\\
         \frac{1}{1 - \sum_i \beta_i^*} & \frac{1-\beta^*_p }{1 - \sum_i \beta_i^*} & \frac{1-\beta^*_p -\beta^*_{p-1}}{1 - \sum_i \beta_i^*}  & \ldots & \frac{1-\sum_{i=2}^{p}\beta_i^*}{1 - \sum_i \beta_i^*}\\
         \end{pmatrix}
\end{align*}
         
It is easy to verify that $\|\bbeta_{\downarrow}^{*\top} (\mathbb{I} -\Ab)^{-1}\|^2 = \frac{\sum_{i = 1}^p (s - \sum_{j = p - i + 2}^{p} \beta^*_j)^2}{(1-s)^2}$, where $s = \sum_i \beta_i^* < 1$. Assume now that $\bbeta_{\downarrow}^*$ is such that $\|\bbeta_{\downarrow}^{*\top} (\mathbb{I} -\Ab)^{-1}\|\|\bbeta^*_{\downarrow}\| \leq \epsilon$ where $\epsilon < \frac{1}{4}$. Then
\begin{align*}
    \lambda_{\max}(\bSigma) \leq 1/(1 - 2\epsilon).
\end{align*}
In addition, once again by Gershgorin's circle theorem we have
\begin{align*}
    \lambda_{\min}(\bSigma) & \geq 1 + \bbeta_{\downarrow}^{*\top} \bSigma \bbeta_{\downarrow}^{*} - 2\bbeta_{\downarrow}^{*\top} (\mathbb{I} - \Ab)^{-1}\bSigma \bbeta_{\downarrow}^* \geq 1 - 2\|\bbeta_{\downarrow}^{*\top} (\mathbb{I} -\Ab)^{-1}\|\|\bbeta^*_{\downarrow}\|\lambda_{\max}(\bSigma)\\
    & \geq 1 - \frac{2\epsilon}{1 - 2 \epsilon}.
\end{align*}

We give two simple examples where the above conditions can be met. The first one is when $\|\bbeta^*_{\downarrow}\|_2 \leq \frac{1}{C\sqrt{p}}$ for some large $C$. In that case we have $s^2 \leq \frac{1}{C^2}$. Also,  $\|\bbeta_{\downarrow}^{*\top} (\mathbb{I} -\Ab)^{-1}\|^2 = \frac{\sum_{i = 1}^p (s - \sum_{j = p - i + 2}^{p} \beta^*_j)^2}{(1-s)^2} \leq \frac{p s^2}{(1-s)^2}$, so that $\|\bbeta_{\downarrow}^{*\top} (\mathbb{I} -\Ab)^{-1}\|\|\bbeta_{\downarrow}\| \leq \frac{s}{C(1-s)} < 1/(C(C-1))$. 

Next suppose that only the first few entries (say $\ell$) of $\bbeta_{\downarrow}^*$ are non-zero and $s < c/\sqrt[4]{\ell}$ for some small $c$. Then  $\|\bbeta_{\downarrow}^{*\top} (\mathbb{I} -\Ab)^{-1}\|^2 = \frac{\sum_{i = 1}^p (s - \sum_{j = p-i+2}^{p} \beta^*_j)^2}{(1-s)^2} \leq \ell s^2/(1-s)^2$, while $\|\bbeta_{\downarrow}^*\|^2 \leq s^2$. Thus $\|\bbeta_{\downarrow}^{*\top} (\mathbb{I} -\Ab)^{-1}\|\|\bbeta_{\downarrow}^*\| \leq \sqrt{\ell} \frac{s^2}{1 - s} \leq c/(1 - c/\sqrt[4]{\ell})$.

The above are only crude sufficient conditions, which are by no means necessary. In fact there are many more examples which we have confirmed have bounded spectrum via numerical verification. 
% \restoregeometry
\end{changemargin}

\end{document}